%% file: ripser++-arxiv.tex
\documentclass[11pt]{article}
\usepackage[letterpaper, margin=1in]{geometry}  

\usepackage{fancyhdr}
\usepackage[lined,noend,algonl,boxed,ruled]{algorithm2e}
\usepackage{xcolor}
\usepackage{graphicx}
\usepackage[font=small]{subcaption}
\usepackage{amssymb}
\usepackage{amsmath}
\usepackage{latexsym}
\usepackage{standalone} 
\usepackage{float}
\usepackage{booktabs}
\usepackage{pdflscape}
\usepackage{mathtools}
\usepackage{tikz}
\usepackage{xstring} 
\usepackage{color}
\usepackage{amsthm}
\usepackage{comment}
\usepackage{url}
\usepackage{hyperref}

\newcommand{\ra}[1]{\renewcommand{\arraystretch}{#1}}
\newtheorem{definition}{Definition}
\newtheorem{lemma}{Lemma}
\newtheorem{proposition}{Proposition}
\newtheorem{theorem}{Theorem}
\newtheorem{corollary}{Corollary}
\newtheorem{example}{Example}
\newtheorem{observation}{Observation}
\newtheorem{condition}{Condition}
\newtheorem{assumption}{Assumption}

\newtheorem{remark}{Remark}

\input{math_commands.tex}

\title{GPU-Accelerated Computation of Vietoris-Rips Persistence Barcodes} 

\author{Simon Zhang\thanks{Department of Computer Science and Engineering, The Ohio State University. Email: zhang.680@osu.edu} \qquad Mengbai Xiao\thanks{Department of Computer Science and Engineering, The Ohio State University. Email: xiao.736@osu.edu} \qquad  Hao Wang\thanks{Department of Computer Science and Engineering, The Ohio State University. Email: wang.2721@osu.edu}}

\global\newcommand{\ccsdesc}[2]{\small{#1~\relax$\rightarrow$~\relax#2;~~\relax}}

\newcommand*\Copyright[1]{%
  \def\@copyrightholder{#1}
  \def\@Copyright{%
    \setbox\@tempboxa\hbox{\includegraphics[height=14\p@,clip]{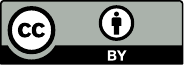}}%
    \@rightskip\@flushglue \rightskip\@rightskip
    \hangindent\dimexpr\wd\@tempboxa+0.5em\relax
    \href{https://creativecommons.org/licenses/by/3.0/}%
         {\smash{\lower\baselineskip\hbox{\unhcopy\@tempboxa}}}\enskip
    \textcopyright\ %
    \ifx!#1!\textcolor{red}{Author: Please fill in the \string\Copyright\space macro}\else#1\fi
    ;\\%
    licensed under Creative Commons License CC-BY\ifx!#1!\\\null\fi\par}}
\Copyright{\textcolor{red}{Author: Please provide a copyright holder}}
\def\subjclassHeading{%
  \textcolor{darkgray}{\fontsize{9}{12}\sffamily\bfseries
                       2012 ACM Subject Classification\enskip}}
\def\Keywords{%
  \textcolor{darkgray}{\fontsize{9}{12}\sffamily\bfseries
                       Keywords and phrases\enskip}}
\def\Category{%
  \textcolor{darkgray}{\fontsize{9}{12}\sffamily\bfseries
                       Category\enskip}}
\def\Supplemental{%
  \textcolor{darkgray}{\fontsize{9}{12}\sffamily\bfseries
                       Supplemental Material\enskip}}
\def\Funding{%
  \textcolor{darkgray}{\fontsize{9}{12}\sffamily\bfseries
                       Funding\enskip}}
\def\Acknowledgements{%
  \textcolor{darkgray}{\fontsize{9}{12}\sffamily\bfseries
                       Acknowledgements\enskip}}
\date{}
\begin{document}
	
	\maketitle
	
	\begin{abstract}
	The computation of Vietoris-Rips persistence barcodes is both execution-intensive and memory-intensive. In this paper, we study its computational structure and identify several unique mathematical properties and algorithmic opportunities with connections to the GPU. Mathematically and empirically, we look into the properties of apparent pairs, which are independently identifiable persistence pairs comprising up to 99\% of persistence pairs. We prove tight upper and lower bounds of the apparent pair rate and some probabilistic lower bounds. We also design massively parallel algorithms to take advantage of the very large number of simplices that can be processed independently of each other. Having identified these opportunities, we develop a GPU-accelerated software for computing Vietoris-Rips persistence barcodes, called Ripser++. Under nice sampling conditions, we show that the expected work complexity of our algorithm is near linear in the number of simplices. The expected depth complexity is dependent only on the computation of the expected number of $p$-dimensional homological cycles. The software achieves up to 30x speedup over the total execution time of the original Ripser and also reduces CPU-memory usage by up to 2.0x. We believe our GPU-acceleration based efforts open a new chapter for the advancement of topological data analysis in the post-Moore's Law era.  
	\end{abstract}
    
\newcommand{\mycommentfont}{\color{blue}\small}

{\subjclassHeading}
\ccsdesc{Theory of computation }{Massively parallel algorithms}
\ccsdesc{Software and its engineering}{Massively parallel systems}
\ccsdesc{Theory of computation Randomness}{geometry and discrete structures}
\medskip


\Keywords{\small{Parallel Algorithms, Topological Data Analysis, Vietoris-Rips, Persistent Homology, Apparent Pairs, High Performance Computing, GPU, Random Graphs}}
\medskip

\Category{\small{Extended Version of Regular Paper presented in SoCG 2020}}
\medskip


\Supplemental{\small{Open Source Software: {https://www.github.com/simonzhang00/ripser-plusplus}}}
\medskip 

\Funding{\small{This work has been partially supported by the National Science Foundation under grants CCF-1513944, CCF-1629403, CCF-1718450, and an IBM Fellowship.}}
\medskip

\Acknowledgements{\small{We would like to thank Ulrich Bauer for technical discussions on Ripser and Greg Henselman for discussions on Eirene. We also thank Greg Henselman, Matthew Kahle, and Cheng Xin on discussions about probability and apparent pairs. We acknowledge Birkan Gokbag for his help in developing Python bindings for Ripser++. We appreciate the constructive comments and suggestions of the anonymous reviewers. Finally, we are grateful for the insights and expert judgement in many discussions with Tamal Dey.}}

    \normalsize{
	\input{introduction}}
    \normalsize{

\input{ripser++-body}}
	
	\normalsize{
	\bibliography{lipics-v2019-bibliography}
    }
	\normalsize{
	\input{appendix}}

	
\end{document}

%% file: math_commands.tex
\def\gX{{\mathcal{X}}}

%% file: introduction.tex
\section{Introduction}
\label{sec:introduction}
Topological data analysis (TDA) \cite{carlsson2009topology} is an emerging field in the era of big data, which has a strong mathematical foundation. As one of the core tools of TDA, persistent homology seeks to find topological or qualitative features of data (usually represented by a finite metric space). It has many applications, such as in neural networks~\cite{guss2018characterizing}, sensor networks~\cite{de2007coverage}, bioinformatics~\cite{dabaghian2012topological}, deep learning~\cite{hofer2017deep}, manifold learning~\cite{niyogi2008finding}, and neuroscience~\cite{luetgehetmann2019computing}. One of the most popular and useful topological signatures persistent homology can compute are Vietoris-Rips barcodes. There are two challenges to Vietoris-Rips barcode computation. The first one is its highly computing- and memory-intensive nature in part due to the exponentially growing number of simplices it must process. The second one is its irregular computation patterns with high dependencies such as its matrix reduction step \cite{zhang2019hypha}. 
Therefore, sequential computation is still the norm in computing persistent homology. There are several CPU-based software packages in sequential mode for computing persistent homology~\cite{bauer2014distributed, bauer2017phat,henselman2016matroid,maria2014gudhi,ripser}. Ripser ~\cite{ripser,tralie2018ripser} is a representative and computationally efficient software specifically designed to compute Vietoris-Rips barcodes, achieving state of the art performance~\cite{bauer2019ripser, otter2017roadmap} by using effective and mathematically based algorithmic optimizations.

The usage of hardware accelerators like GPU is inevitable for computation in many areas. To continue advancing the computational geometry field, we must include hardware-aware algorithmic efforts. The ending of Moore’s law~\cite{theis2017end} and the termination of Dennard scaling~\cite{dennard1974design} technically limits the performance improvement of general-purpose CPUs~\cite{esmaeilzadeh2012dark}. The computing ecosystem is rapidly evolving from conventional CPU computing to a new disruptive accelerated computing environment where hardware accelerators such as GPUs play the main roles of computation for performance improvement. 

Our goal in this work is to develop GPU-accelerated computation for Vietoris-Rips barcodes, not only significantly improving the performance, but also to lead a new direction in computing for topological data analysis. We identify two major computational components for computing Vietoris-Rips barcodes, namely filtration construction with clearing and matrix reduction. We look into the underlying mathematical and empirical properties tied to the hidden massive parallelism and data locality of these computations. Having laid mathematical foundations, we develop efficient parallel algorithms for each component, and put them together to create a computational software. Our approach is theoretically sound and, under nice sampling conditions, has a complexity that is nearly linear in expectation in the number of generated simplices. 

Our contributions presented in this paper are as follows:
\begin{enumerate}
\item  We prove theoretical bounds on the number of so-called ``apparent pairs." These apparent pairs have a natural algorithmic connection to the GPU. 
\item 	We design and implement hardware-aware massively parallel algorithms that accelerate the two major computation components of Vietoris-Rips barcodes as well as a data structure for persistence pairs for matrix reduction.
\item  We show in $d_{amb}$-dimensional Euclidean space, that under a uniform sampling condition with $k$ cavities the expected work complexity of our Algorithm is $O(kn_p\log(n_0)+n_{p+1})$ with expected depth complexity of  $O(kn_p\log(n_0)+n_0(p+1))$ for $n_0$ points, $n_p$ $p$-dimensional simplices.
\item  We perform extensive experiments justifying our algorithms' computational effectiveness as well as dissecting the nature of Vietoris-Rips barcode computation, including looking into the expected number of apparent pairs as a function of the number of points.
\item  We achieve up to 30x speedup over the original Ripser software; surprisingly, up to 2.0x CPU memory efficiency and requires, at best, 60\% of the CPU memory used by Ripser on the GPU device memory.
\item     Ripser++ is an open source software in the public domain to serve the TDA community and relevant application areas.
\end{enumerate}

%% file: ripser++-body.tex
    \section{Background}
    A metric is defined as follows:
 \begin{definition}
     A metric $d: X \times X \rightarrow \mathbb{R}^+$ on a point set $X$ is a real-valued function on pairs of points so that:
     \begin{itemize}
         \item $d(x,x)=0, \forall x \in X$ (identity distance is zero)
         \item $d(x,y)>0$ if $x \neq y \in X$ (distance between differing point is positive)
         \item $d(x,y)=d(y,x), \forall x,y \in X$ (symmetry)
         \item $d(x,z)\leq d(x,y)+d(y,z), \forall x,y,z \in X$ (triangle inequality)
     \end{itemize}
      We say a metric $d$ has \textbf{no equidistant points} if 
      \begin{equation}
          d(x,y)\neq d(z,w), \forall x,y,z,w \in X \text{ and } (x,y)\neq (z,w)
      \end{equation}
 \end{definition}
    \begin{definition}
       A \textbf{metric measure space} is the triple $(\gX,d,P)$ where $\gX$ is a point set equipped with a metric $d$ where $P$ is a probability distribution supported on $\gX$.
    \end{definition}

    
	We assume our \emph{data} is a finite set of i.i.d. samples from $P$, denoted as $X$. 
 
 Since $X$ is finite, the metric $d$ over $\gX$ can be represented by its \textbf{distance matrix $D$}:
 \begin{equation}
     D[i,j]= d \text{(point $i$, point $j$) with }D[i,i]=0
 \end{equation} 
 where the points from $X$ are assigned a unique integer from $0,...,n_0-1$.

 \begin{definition}
      With the metric $d$, we can define the aspect ratio of $X$ as 
 \begin{equation}
    a(X)_{d}\triangleq \frac{\max_{x,y \in X} d(x,y)}{\min_{x,y \in X} d(x,y)} 
 \end{equation}
 \end{definition}
	\begin{definition}
    \label{def: simplicialcomplex}
	    Define an (abstract) simplicial complex $K$ as a collection of simplices closed under the subset relation, where a simplex $\sigma$ is defined as a subset of $X$. This means that any simplex $\sigma': \sigma'\subseteq \sigma$ must also belong to $K$.

        The dimension of a simplex $\sigma$ is $\lvert\sigma \rvert -1$.
	\end{definition}
    \begin{definition}
         A filtration of simplicial complexes is a sequence of simplicial complexes ordered by the subset relationship.
    \end{definition}
	 A particularly popular and useful~\cite{aktas2019persistence} filtration is a Vietoris-Rips filtration. On the metric space, we generate simplices over an adjustable threshold $t \in \mathbb{R}$. See Figure \ref{fig: VR-barcodes} for an illustration. A subset of points $\sigma \subseteq X$ forms a simplex if they are simultaneously close enough with respect to $\tau$.  Let
	\begin{equation}
	\textsf{Rips}_t(X)= \{\emptyset \neq \sigma \subseteq X \mid \textsf{diam}(\sigma)\leq t\},
	\end{equation}
	where $t \in \mathbb{R}$ and $\textsf{diam}(\sigma)$ is the maximum 
	distance 
	between pairs of points in $\sigma$ as determined by metric $d$. 
    \begin{definition}
	The \textbf{Vietoris-Rips filtration}, or Rips filtration for short, for some $R\in \mathbb{R}$ is defined as the sequence: $(\textsf{Rips}_t(X))_{t\leq R}$, indexed by growing $0\leq t \leq R$  where $\textsf{Rips}_t(X)$ strictly increases in cardinality for growing $\tau$.  

    A \textbf{full-Rips filtration} is a Vietoris-Rips filtration where $R=\infty$.
    \end{definition}
    \begin{definition}
    The aspect ratio is defined on a Vietoris Rips filtration as follows:
    \begin{equation}
        a(\textsf{Rips}_t(X))_{d}\triangleq \frac{R}{\min_{\sigma \subseteq X}\textsf{diam}(\sigma)}
    \end{equation}
    \end{definition}
    \begin{definition}
    A \textbf{$k$-skeleton} $K^{ \leq k}$ of a simplicial complex $K$ is defined as the abstract subcomplex of $K$ where every $p$-dimensional simplex in $K^{\leq k}$ has $p\leq k$.
    \end{definition}
    \begin{example}
    The $1$-skeleton of a Rips complex $\textsf{Rips}_t(X)$ at radius $t$ in Euclidean space is a $t$-radius graph. This is the graph of all the points with edges between points that are at most distance $\tau$ apart.
    \end{example}
	\begin{figure}[h]
    \centering
    \includegraphics[width=0.9\columnwidth]{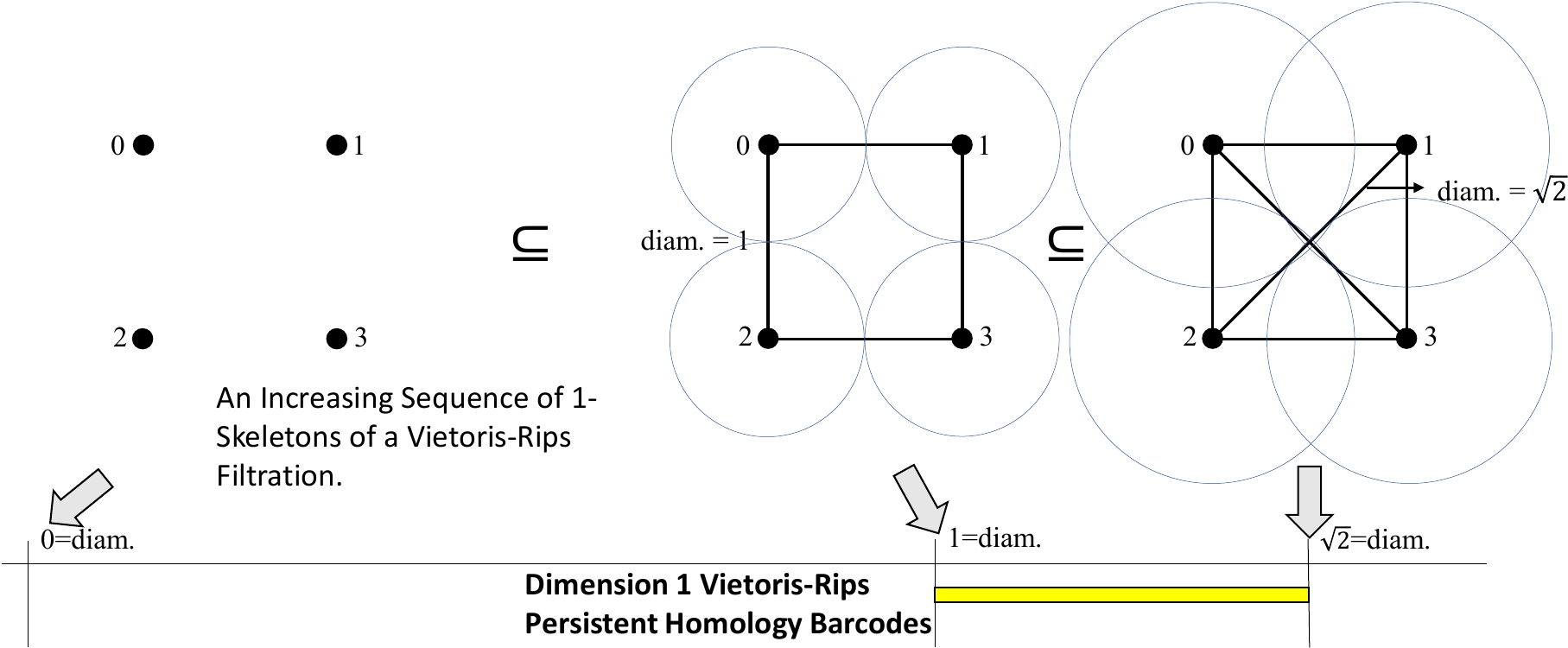}
    \caption{\small{A filtration on an example finite metric space of four points of a square in the plane. The $1$-skeleton, or simplicial complex of only points and unordered pairs of points, at each diameter value where ``creation" or ``destruction" occurs is shown. The 1 dimensional Vietoris-Rips barcode is below it: a $1$-cycle is ``created" at diameter 1 and ``destroyed" at diameter $\sqrt{2}$.}}
    \label{fig: VR-barcodes}
    \centering
    \end{figure}
    
\subsection{The Boundary Operator}
For an (abstract) simplicial complex $K$, as defined in Definition \ref{def: simplicialcomplex}, for each dimension $p, 0\leq p\leq \text{dim}(K)$ the set of all $p$-dimensional simplices freely generate a $p$-chain vector space over some field.

\begin{definition}
    Given a simplicial complex $K$, 
    a $p$-chain $\tau= \sum_{\sigma \in K: \lvert\sigma\rvert=p+1} m_{\sigma}\cdot \sigma$ with coefficients $m_{\sigma}$ is a formal sum of $p$-dimensional simplices from $K$.

    A $p$-chain vector space $(C_p(K), +)$ consists of the set of $p$-chains of $K$ with the formal sum operation $+$ over a coefficient field.
\end{definition}
The formal sum between two $p$-chains geometrically represents a union of two $p$-chains where at their intersection the coefficients add in the coefficient field.

Over the reals, if there are $n$ simplices, then a simplex $\sigma_i$ can be represented as a standard basis vector $e_i \in \mathbb{R}^n$ and a $p$-chain $\sigma = \sum_{i \in [n]} m_i \cdot \sigma_i$ can be represented by $\sum_{i \in [n]} m_i \cdot e_i \in \mathbb{R}^n$. This is since the $p$-chains form a vector space. Thus any $p$-chain can be viewed as a column vector. We denote this column vector by $[\sigma] \in \mathbb{R}^n$. When a simplex $\sigma_i$ is a summand of $p$-chain $\sigma$ then we denote this by $\sigma_i \in \sigma$.

Connectivity amongst these simplices can be defined algebraically through a boundary map. We define this formally here:
\begin{definition}\label{eq: boundary}
A \textbf{boundary map} $\partial_{p+1}^K: C_{p+1}(K) \rightarrow C_{p}(K)$ is defined by 
\begin{equation}
    \partial_{p+1}^K(\tau)\triangleq \sum_{i=1}^{p+1} (-1)^i\cdot [{\tau}_{-i}], \forall \tau=\sum_{\sigma \in K : \sigma=\{v_0,...,v_{p+1}\}}m_{\sigma}\cdot \sigma  \in C_{p+1}(K)
\end{equation}
where ${\tau}_{-i}= \sum_{\sigma \in K : \sigma=\{v_0,...,v_{p+1}\}}m_{\sigma}\cdot (\sigma \setminus \{v_i\})$
\end{definition}
The formal sum from Equation \ref{eq: boundary} is called the $p$-dimensional \textbf{boundary} of the $p+1$-dimensional boundary. 

When it is clear what the simplicial complex $K$ is for a boundary map $\partial_{p+1}^K$, we often omit the $K$ when denoting the boundary map: $\partial_{p+1}$. We can also denote the $p$-dimensional boundaries with the simplicial complex $K$ as follows: 
\begin{equation}
    \partial_{p+1}(\sigma), \forall \sigma \in K, \lvert \sigma \rvert =p+2
\end{equation}

This boundary map has a matrix representation, which is the concatenation of all the column representations of $\partial_{p+1}(\sigma), \sigma \in K, \lvert \sigma\rvert =p+2$. These column representations are denoted $[\partial_{p+1}(\sigma)]$, and their column-wise concatenation is called the $p+1$-boundary matrix. Their concatenation into a matrix is denoted $[\partial_{p+1}]$. If we concatenate all the $p$-boundary matrices $[\partial_{p}]$ for $0\leq p\leq \text{dim}(K)$, we form the boundary matrix 
\begin{equation}
    [\partial]\triangleq [\partial_{0}\text{ }  \| \cdots \| \text{ }  \partial_{\text{dim}(K)} ]
\end{equation}
\subsection{Homology}
Given an abstract simplicial complex $K$, at each dimension $p: 0\leq p\leq \text{dim}(K)$, we can define a $p$-dimensional cycle space, or space of $p$-dimensional disconnections as:
    $\textsf{ker}(\partial_p)$. The $p$-dimensional boundary space, or space of $p$ dimensional connectivity: $\textsf{im}(\partial_p)$.
    
Homology can be defined by the algebraic measurement of the  disconnections in $p$ dimensions up to connectivity from $p+1$ dimensions. This is formally defined by the following quotient vector space over a coefficient field at dimension $p$:
\begin{equation}
    H_{p}(K)\triangleq \frac{\textsf{ker}(\partial_p)}{\textsf{im}(\partial_{p+1})}
\end{equation}

	\subsection{The Simplex-wise Refinement of the Vietoris-Rips Filtration}
	\label{sec: preliminaries-simplexwise}
	For computation (see Section \ref{sec: preliminaries-computation}) of Vietoris-Rips persistence barcodes, which is a matching of simplices, it is necessary to construct a simplex-wise refinement $S$ 
	of a given filtration $F$. This means that we form a total ordering of the simplicial complexes so that the filtration can be indexed by a single scalar from $\mathbb{R}$. The simplices within each simplicial complex, however, may not have an ordering themselves. Thus, a filtration $F$ is only equivalent to a \textbf{partial order} on the simplices of $K$, where $K$ is the largest simplicial complex of $F$. 
    We would suspect that it is possible to  construct $S$ since the simplex is the smallest unit of connectivity in the Vietoris-Rips complex. 
    
    To construct $S$, we assign a total order on the simplices $\{s_i\}_{i=1,...,\lvert K\rvert}$ of $K$, extending the partial order induced by $F$ so that the increasing sequence of subcomplexes $S=(K_j \triangleq \bigcup_{i \leq j}{\{s_i\}})_{j=1,...,\lvert K \rvert}$ ordered by inclusion grows subcomplexes by one simplex at a time. In category-theoretic terms, this finite construction exists  because the subcategory of simplicial complexes with objects consisting of the subcomplexes $K_i$ from $S$ and the morphisms consisting of inclusion maps from $K_i$ to $K_j$ for $i<j$ form a \textbf{simple acyclic} category~\cite{zhang2025computing}.  
    
    This means that the subcategory exhibits: 
    \begin{enumerate}
        \item (Acylicity~\cite{kozlov2008combinatorial}) The objects can be topologically sorted into a total order so that:
        \begin{enumerate}
            \item The only endomorphism on an object is the identity map and
            \item Only the identity map has an inverse.
        \end{enumerate} 
        \item (Simplicity) Every nontrivial inclusion is a composition of the \emph{unit inclusion maps} which unions a single simplex to a subcomplex.
    \end{enumerate}
    For a Vietoris-Rips complex, there are many ways to order a simplex-wise refinement $S$ of $F$
	, see e.g. ~\cite{luetgehetmann2019computing}; in the case of Ripser and Ripser++, we use the following \textbf{simplex-wise filtration} ordering criterion on simplices:
	
	\label{sec: filtration-order}
	\begin{enumerate}
	\item by increasing diameter: denoted by $\textsf{diam}(\sigma)$,
	\item by increasing dimension: denoted by $\textsf{dim}(\sigma)$, and
	\item by decreasing combinatorial index: denoted by $\textsf{cidx}(\sigma)$ (equivalently, by decreasing lexicographical order on the decreasing sequence of vertex indices)~\cite{siddique2016proof,knuth1997art,pascal1887sopra}.
	\end{enumerate}
	Every simplex in the simplex-wise refinement will correspond to a ``column" in a uniquely associated (co-)boundary matrix for persistence computation. Thus we will use the terms ``column" and ``simplex" interchangeably to explain our algorithms. 

    \subsection{Persistence Pairs}
	Define homological \emph{persistence pairs} over a Vietoris-Rips filtration as a pair of ``creation" and ``destruction" simplices for the creation and destruction of homological rank, or Betti number, of each $\textsf{Rips}_t(X)$ from a Vietoris Rips filtration
    \begin{equation}
         \emptyset\subseteq \cdots \subseteq \textsf{Rips}_t(X)\subseteq \cdots \subseteq \text{Rips}_R(X)
    \end{equation} as $t$ varies, see \cite{edelsbrunner2000topological}. These are the unique pairs of simplices $(\sigma_s,\tau_t), s<t$ occurring during the creation and destruction times for a simplex-wise refinement of the Vietoris-Rips filtration. They are the solution to the (primal) matrix reduction problem, an online linear programming problem, as formulated in the thesis \cite{zhang2025computing}.
	\begin{remark}
	We will show that the combinatorial index maximizes the number of ``apparent pairs" when breaking ties under the conditions of Theorem \ref{theorem: apparent bounds}.
	\end{remark} 
	\subsection{The Combinatorial Number System}
	\label{sec: combinatorial_number_system}
	We use the combinatorial number system to encode simplices. The combinatorial number system is simply a bijection between ordered fixed-length $\mathbb{N}$-tuples and $\mathbb{N}$. It provides a minimal representation of simplices and an easy extraction of simplex facets (see Algorithm \ref{alg:facets}), cofacets (see Algorithm \ref{alg:cofacets_sparse}), and vertices. When not mentioned, we assume that all simplices are encoded by their combinatorial index. 
	The bijection is stated as follows:
	\begin{equation}
	\mathbb{N}^{p+1}\ni(v_p,...,v_0) \iff {v_p \choose p+1} + ... + {v_0 \choose 1} \in \mathbb{N}, v_p>v_{p-1}>...>v_0 \geq 0.
	\end{equation}
	For a proof of this bijection see~\cite{siddique2016proof,knuth1997art,pascal1887sopra}.
	\begin{remark}
	It should be noted that, without mentioning, we will use the notation $(v_p,...,v_0)$ with $v_p>v_{p-1}>...>v_0 \geq 0$ for simplices and a lexicographic ordering on the simplices by \emph{decreasing} sorted vertex indices. One may equivalently view the lexicographic ordering on simplices $(v_0,...,v_p)$ with $0 \leq v_0<v_1<...<v_p$ as being in colexicographic order on the given increasingly ordered vertices. 
	\end{remark}

	\section{Computation}
	\label{sec: preliminaries-computation}
	The general computation of persistent barcodes involves two inter-relatable stages. One stage is to construct a simplex-wise refinement \cite{bauer2017phat} of the given filtration. The other stage is to ``reduce" the corresponding boundary matrix by a ``standard algorithm"~\cite{edelsbrunner2010computational}. 
	In Algorithm \ref{alg:standard-algorithm}, let $\textsf{low}_R(j)$ be the maximum nonzero row of column $j$, $-1$ if column $j$ is zero for a given matrix $R$. The pairs $(\textsf{low}_R(j),j)$ over all $j$ correspond bijectively to \emph{persistence pairs}.
	
	\begin{algorithm}[H]
    \SetKwComment{Comment}{//}{}
    
    \caption{Standard Persistent Homology Computation}
    \label{alg:standard-algorithm}
    
    \KwIn{filtered simplicial complex $\pmb{K}$}
    \KwOut{$\pmb{P}$ persistence barcodes}

    $\pmb{F} \gets \pmb{F}_{\pmb{K}}$
    \Comment{Let $\pmb{F}$ be the filtration of $\pmb{K}$}
    $\pmb{S} \gets \text{simplex-wise-refinement}(\pmb{F})$ 
    \Comment{$\pmb{F} = \pmb{S} \circ r$ where $r$ is injective}
    $R \gets \partial^{\pmb{S}}$ \;
    
    \ForEach{column $j$ in $R$}{
        \While{$\exists k < j$ \text{such that} $\textsf{low}_{R}(j) = \textsf{low}_{R}(k)$}{
            column $j \gets$ column $k$ + column $j$ \\
        \If{$\textsf{low}_{R}(j) \neq -1$}{
            $\pmb{P} \gets \pmb{P} \cup r^{-1}([\textsf{low}_{R}(j), j))$ 
            \Comment{We call the pair $(\textsf{low}_{R}(j), j)$ a pivot in the matrix $R$.} 
        }
    }
    }
    \end{algorithm}

	The construction stage can be optimized \cite{bauer2019ripser, hylton2017performance, maria2014gudhi, gudhi:urm, zomorodian2010fast}. Furthermore, all existing persistent homology software efforts are based on the standard algorithm\cite{adams2011javaplex, bauer2019ripser, bauer2014distributed, bauer2017phat, henselman2016eirene, maria2014gudhi, Morozov, zhang2019hypha}.
	
	\subsection{The Coboundary Matrix}
	\label{sec: coboundary_matrix}
    A coboundary matrix consists of coboundaries (each column is made up of the cofacets of the corresponding simplex) where the columns/simplices are ordered in reverse to the order given in Section \ref{sec: filtration-order} (see ~\cite{de2011dualities}). 
    
    In  Proposition 4.3.5 of \cite{zhang2025computing}, we show that it is possible to compute the standard matrix reduction of Algorithm \ref{alg:standard-algorithm} on the anti-transposed boundary matrix, or coboundary matrix, $[\partial]^{T}_{anti}$. We call this ``computing cohomology" since for every $p\geq 0$ the standard algorithm is iterating through $p$-cochains from the dual of the $p$-chain vector space $C_p(\textsf{Rips}(X))$.
    
	We compute cohomology~\cite{de2011persistent, de2011dualities, dey2014computing} in Ripser++, like in Ripser, for performance reasons specific to Rips filtrations mentioned in~\cite{bauer2019ripser} which will be reviewed in Section \ref{sec: ripser-original}. The memory allocation needed to represent a sparse column-store matrix is lowered when using a coboundary matrix instead of a boundary matrix. This is because there are at most ${n_0 \choose p+2}$ number of $(p+1)$-dimensional simplices (sparsely represented rows) and only ${n_0 \choose p+1}$ $p$-dimensional simplices (columns). 
	Excessive memory allocation can furthermore become a bottleneck to total execution time.
    
    If certain columns can be zeroed/cleared \cite{chen2011persistent} in the coboundary matrix, we will still denote the cleared matrix as a coboundary matrix since the matrix reduction does nothing on zero columns (see Algorithm \ref{alg:standard-algorithm}).
	
	\begin{figure}[!h]
    \includegraphics[width=1.0\columnwidth]{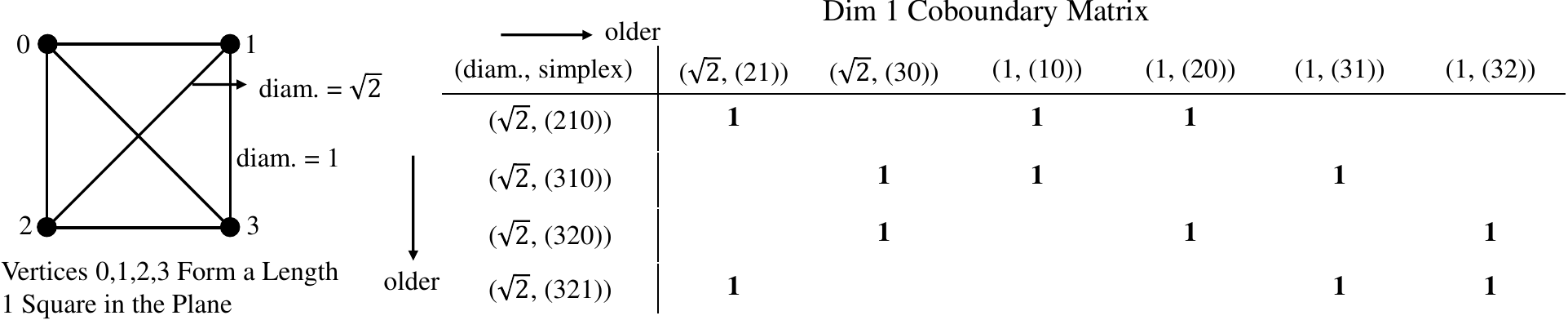}
    \caption{\small{The full 1-skeleton for the point cloud of Figure \ref{fig: VR-barcodes}. Its 1-dimensional coboundary matrix is shown on the right. Let $(e,(a_p,...,a_0))$ be a $p$-dimensional simplex with vertices $a_p$,...,$a_0$ such that $v_p>v_{p-1}>...>v_0 \geq 0$ and diameter $e \in \mathbb{R}^+$. For example, simplex $(1,(10))$ has vertices 1 and 0 with diameter 1. 
    The order of the columns/simplices is the reverse of the simplex-wise refinement of the Vietoris-Rips filtration.}}
    \label{fig: coboundarymatrix}
    \centering
    \end{figure}

 \subsection{Computational Structure in Ripser, a Review}
	\label{sec: ripser-original}
	The sequential computation in Ripser follows the two stages given in Algorithm \ref{alg:standard-algorithm}, with the two stages computed for each dimension's persistence in increasing order. This computation is further optimized with four key ingredients \cite{bauer2019ripser}. We use and build on top of all of these four optimizations for high performance, which are: 
	
	\begin{itemize}
	\item 1. The clearing lemma \cite{chen2011persistent} \cite{bauer2014clear} \cite{zhang2019hypha}, 
	\item 2. Computing cohomology \cite{de2011dualities} \cite{dey2014computing} \cite{zhang2025computing},  
	\item 3. Implicit matrix reduction \cite{bauer2019ripser}, and
	\item 4. The emergent pairs lemma \cite{bauer2019ripser}.
	\end{itemize}

	In this section we review these four optimizations as well as the enclosing radius condition. Our contributions follow in the next section, Section \ref{sec: contributions}. The reader may skip this section if they are very familiar with the concepts from \cite{bauer2019ripser}.
	\subsection{Clearing Lemma}
	\label{sec: clearing-lemma}
	As shown in \cite{carlsson2019persistent}, there is a partition of all simplices into ``creation" and ``destruction" simplices. The introduction of a ``creation" simplex in a simplex-wise refinement of the Rips filtration creates a homology class. On the other hand, ``destruction" simplices zero a homology class or merge two homology classes upon introduction into the simplex-wise filtration. Destruction simplices are of exactly one higher dimension than their corresponding creation simplex. 
	
	Paired creation and destruction simplices are represented by a pair of columns in the boundary matrix. For a boundary matrix, the clearing lemma states that paired creation columns must be zero after reduction~\cite{chen2011persistent, bauer2014clear}. Furthermore, destruction columns are nonzero when fully reduced and their lowest nonzero entry after reduction by the standard algorithm in Algorithm \ref{alg:standard-algorithm}, determines its corresponding paired creation column.
	
	This lemma optimizes the matrix reduction stage of computation and is most effective when it is used before any column additions. For example, when ``computing cohomology" with the clearing lemma, a significant number of paired creator columns can be eliminated due to the application of clearing to the top $p$-dimensional simplices, which dominate the total number of simplices. 
    
    The clearing lemma is widely used in all persistent homology computation software packages to lower the computation time of matrix reduction. As shown in \cite{zhang2019hypha}, the smallest set of columns taking up 50\% of all column additions in Algorithm \ref{alg:standard-algorithm}, also known as ``tail columns" are comprised in a large percentage by columns that can be zeroed by the clearing lemma.  

	\subsection{Zero-Dimensional Persistence Computation: Kruskal's Algorithm}
	\label{sec: 0-persistence}
        The $1$-skeleton of the Rips complex at dimension $1$ is an undirected graph. Viewing the simplex-wise filtration of $1$-skeleton Rips subcomplexes, the edges are introduced in order of distance. The edges that form cycles are viewed as $1$-dimensional creation events and are called complementary edges. The remaining edges destroy zero-dimensional creation events (the points). This allows us to compute zero-dimensional persistent homology. It also simulates Kruskal's algorithm in exactly the same way. Thus, the edges which introduce a destruction of a zero-dimensional creation event form a geometric minimum spanning tree (GMST), and the complementary edges form a $1$-skeleton. 
        
    Zero-dimensional persistence can be computed by a union-find algorithm in Ripser. This algorithm has complexity of  $O(\alpha(n_0) \cdot n_0^2)$, where $n_0$ is the number of points and $\alpha$ is the inverse of the Ackermann's function (essentially a constant). There is no known algorithm that can achieve this low complexity for computing persistence in higher dimensions. This is another reason for computing cohomology. In computing cohomology, clearing is applied from lower dimension to higher dimension (clearing out columns in the higher dimension), the zeroth dimension has no cleared simplices and there are very few $p$-dimensional simplices compared to $(p+1)$-dimension simplicies. Thus the zeroth dimension should be computed as efficiently as possible without clearing. 
    Since it is more efficient to keep the union-find algorithm on CPU, we focus only on dimension $\geq 1$ persistence computation in this paper and in Ripser++. GPU can offer speedup, especially in the top dimension, in this case. 
	\subsection{Implicit Matrix Reduction}
	In Ripser, the coboundary matrix is not fully represented in memory. Instead, the columns or simplices are represented by natural numbers via the combinatorial number system~\cite{knuth1997art, bauer2019ripser, pascal1887sopra} and cofacets are generated as needed and represented by combinatorial indices. This saves on memory allocation along the row-dimension of the coboundary matrix, which scales by a multiplicative factor of $O(n_0)$ more in cardinality than the column-dimension. Furthermore, the generation of cofacets allows us to trade computation for memory. Memory address accesses are replaced by arithmetic and repeated accesses of the much smaller distance matrix and a small binomial coefficient table. Implicit matrix reduction intertwines coboundary matrix construction with matrix reduction.
	
	\subsection{Reduction Matrix vs. Oblivious Matrix Reduction }
	There are two matrix reduction techniques in Ripser, see Algorithm \ref{alg:standard-algorithm}, that can work on top of the implicit matrix reduction. These techniques are applied on a much smaller submatrix of the original matrix in Ripser++ significantly improving performance over full matrix reduction, see Section \ref{sec: performance-optimization}. 
	
	The first is called the reduction matrix technique. This involves storing the column operations on a particular column in a $V$ reduction matrix (see a variant in \cite{boissonnat2013compressed}) by executing the same column operations on $R$ as on the initially identity matrix $V$; $R=\partial \cdot V$ where $\partial$ is the (implicitly represented) boundary operator and where $R$ is the matrix reduction of $\partial$. To obtain a fully reduced column $R_j$ as in Algorithm \ref{alg:standard-algorithm}, the nonzeros $V_{i,j}$ of a column of $V_j$ identify the sum of boundaries $\partial_i$ needed to obtain column $R_j= \Sigma_i \partial_i \cdot V_{i,j}$.
	
	The second is called the oblivious matrix reduction technique. This involves not caching any previous computation with the $R$ or $V$ matrix, see Algorithm \ref{alg:oblivious} for the reduction of a single column $j$. Instead, only the boundary matrix pivot indices are stored. A \emph{pivot} is a row column pair $(r,c)$, being the lowest nonzero entry of a fully reduced nonzero column $c$ and representing a persistence pair. 
	
	In the following, we use the notation $R_j$ to denote a column reduced by the standard algorithm (Algorithm \ref{alg:standard-algorithm}) and $R[j]$ to denote a partially obliviously reduced column. $D_j$ denotes the $j$th column of the boundary matrix.
	
    \begin{algorithm}[H]
    \SetKwComment{Comment}{//}{}
    \caption{Oblivious Column Reduction}
    \label{alg:oblivious}
    
    \KwIn{$j$: column to reduce index, $D=\partial$: boundary matrix, $\textsf{lookup}[$rows $0..j-1]$: lookup table with $\textsf{lookup}[row]= col$ if $(row,col)$ is a pivot, -1 otherwise; $\textsf{low}(j)$: the maximum row index of any nonzero entry in column $j$, -1 if the column $j$ is $\textbf{0}$.}
    \KwOut{$R[j]$ is fully reduced by oblivious reduction and is equivalent to a fully reduced $R_j$ as in Algorithm \ref{alg:standard-algorithm}.}

    \Comment{Assume columns of index $0$ to $j-1$ have all been reduced by the oblivious column reduction algorithm}
    $R[j] \gets D_j$ \;
    \While{$\textsf{lookup}[\textsf{low}(R[j])] \neq -1$}{
        $R[j] \gets R[j] + D_{\textsf{lookup}[\textsf{low}(R[j])]}$ \;
    \If{$R[j] \neq \textbf{0}$}{
        $\textsf{lookup}[\textsf{low}(R[j])] \gets j$ \;
    }
    }
\end{algorithm}

    The reduction matrix technique is correct by the fact that it (re)computes $R_j$ = $\Sigma_i D_i \cdot V_{i,j}$ as needed before adding it with $R_k$ for $k>j$. Thus it involves the same
    sequence  $(R_{j})_{j}$ of columns to do column additions with $R_k$ as in the standard algorithm in Algorithm \ref{alg:standard-algorithm}.
    
    \begin{lemma}
    \label{lemma: oblivious}
    Algorithm 2 (oblivious column reduction from Ripser) is equivalent to a reduction of column $j$ as in Algorithm \ref{alg:standard-algorithm}, namely $R[j] \gets R[j]+R_i$ where $i= \textsf{lookup}[\textsf{low}(R[j])]$.
    \end{lemma}
    
    Our proof of Lemma \ref{lemma: oblivious} is in the Appendix Section \ref{sec: appendix-oblivious-proof}.
    
	
	The reduction matrix technique can lower the column additions (addition of $D_i$ to $R[j]$) needed to reduce any particular column $j$ since after many column additions, many of the nonzeros of $V_j$ will cancel out. 
    This is in contrast to oblivious matrix reduction where there cannot be any cancellation of column additions. Experiments show that datasets with large number of column additions are executed more efficiently with the reduction matrix technique rather than the oblivious matrix reduction technique. In fact, certain difficult datasets will crash on CPU due to too much memory allocation for columns from the oblivious matrix reduction technique but execute correctly in hours by the reduction matrix technique.
	
	\subsection{The Emergent Pairs Lemma}
	As we generate column cofacets during matrix reduction, we may ``skip over" their construction if we can determine that they are ``0-addition columns" or have an ``emergent pair" \cite{zhang2019hypha, bauer2019ripser}. These columns have no column to their left that can add with them. The lemma is stated in \cite{bauer2019ripser} and involves a sufficient condition to find a ``lowest nonzero" or maximal indexed nonzero in a column followed by a check for any columns to its left that can add with it. These nonzero entries correspond to ``shortcut pairs" that form a subset of all persistence pairs. We may pair implicit matrix reduction with the emergent pairs lemma to achieve speedup over explicit matrix reduction techniques~\cite{bauer2017phat, zhang2019hypha}. However, apparent pairs (see Figure \ref{fig:apparentpairs}), when processed massively in parallel by GPU are even more effective for computation than processing the sequentially dependent shortcut pairs (see Figure \ref{fig: apparentvsshortcut}). 
	
	\subsection{Filtering Out Simplices by Diameter Equal to the ``Enclosing Radius"}
	\label{sec: appendix-enclosing-radius}
	We define the enclosing radius $R$ as $\min_{x \in X} \max_{y \in X} d(x,y)$, where $p$ is the underlying metric of our finite metric space $X$. If we compute Vietoris-Rips barcodes up to diameter $\leq$ $R$, then the nonzero persistence pairs will not change after the threshold condition is applied \cite{henselman2016matroid}. Notice that applying the threshold condition is equivalent to truncating the coboundary matrix to a lower right block submatrix, potentially significantly lowering the size of the coboundary matrix to consider. We prove the following claim used in \cite{henselman2016eirene}.
	
	\begin{proposition}
	Computing persistence barcodes for full Rips filtrations with diameter threshold set to the enclosing radius $R$ does not change the nonzero persistence pairs.   
	\end{proposition}
	
	\begin{proof}
	Notice that when we reach the ``enclosing radius" $R$ length in the filtration, every point has an edge to one apex point $p \in  X$. This means that any $p$-dimensional cycle must actually be a boundary. Thus the following two statements are true. 1. Any persistence interval $[creation,destruction)$ with $creation \leq R$ must have $destruction \leq R$. 2. Since no cycles that are not boundaries can form after $R$, there will be no nonzero persistence barcodes $[creation,destruction)$ with $creation>R$.
	
	By statements 1 and 2, we can restrict all simplices to have diameter $\leq$ $R$ and this does not change any of the nonzero persistence intervals of the full Rips filtration.
	\end{proof}

\section{Mathematical and Algorithmic Foundations for GPU Acceleration}
\label{sec: contributions}
	\subsection{Overview of GPU-Accelerated Computation}
	Figure \ref{fig:framework-ripserpp}(a) shows a high-level structure of Ripser, 
	which processes simplices dimension by dimension. In each dimension starting at dimension 1, the filtration is constructed and the clearing lemma is applied followed by a sort operation. 
	The simplices to reduce are further processed in the matrix reduction stage, where the cofacets of each simplex are enumerated to form coboundaries 
	and the column addition is applied iteratively. 
	The matrix reduction is highly dependent among columns. 
	
	
	\begin{figure}[h]
		\includegraphics[width=1.0\textwidth]{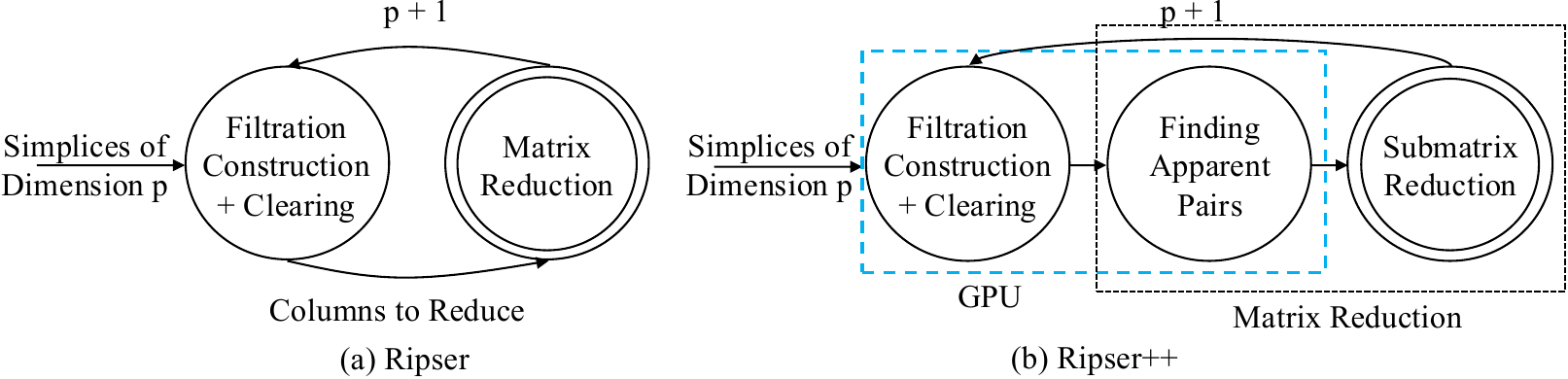}
		\caption{\small{A High-level computation framework comparison of Ripser and Ripser++ starting at dimension $p\geq 1$ (see Section \ref{sec: 0-persistence} for $p=0$). Ripser follows the two stage standard persistence computation of sequential Algorithm \ref{alg:standard-algorithm} with optimizations. In contrast, Ripser++ finds the hidden parallelism in the computation of Vietoris-Rips persistence barcodes, extracts ``Finding Apparent Pairs'' out from Matrix Reduction, and parallelizes ``Filtration Construction with Clearing'' on GPU. These two steps are designed and implemented with new parallel algorithms on GPU, as shown in the Figure \ref{fig:framework-ripserpp}(b) with the dashed rectangle.}}
		\label{fig:framework-ripserpp}
		\centering
		
	\end{figure}
	
    Running Ripser intensively on many datasets, we have observed its inefficiency on CPU. There are two major performance issues. First, in each dimension, the matrix reduction of Ripser uses an enumeration-and-column-addition style to process each simplex. Although the computation is highly dependent among columns, a large percentage of columns (see Table \ref{tab:all-pairs} in Section \ref{sec: experiments}) do NOT need the column addition. Only the cofacet enumeration and a possible persistence pair insertion (into the hashmap of Ripser) are needed on these columns. 
	In Ripser, a subset of these columns are identified by the ``emergent pair" lemma \cite{bauer2019ripser} as columns containing ``shortcut pairs." Ripser follows the sequential framework of Figure \ref{fig:framework-ripserpp}(a) to process these columns one by one, where rich parallelisms, stemming from a large percentage of ``apparent pairs", are hidden. 
	Second, in the filtration construction with clearing stage, applying the clearing lemma and a predefined threshold is independent among simplices. 
	Furthermore, one of the most time consuming part of filtration construction with clearing is sorting. A highly optimized sorting implementation on CPU could be one order of magnitude faster than sorting from the C++ standard library \cite{hou2018framework}. On GPU, the performance of sorting algorithms can be further improved due to the massive parallelism and the high memory bandwidth of GPU \cite{Satish:2009:DES:1586640.1587667, Sintorn:2008:FPG:1412749.1412828}.
	
	In our hardware-aware algorithm design and software implementation, we aim to turn these hidden parallelisms and data localities into reality for accelerated computation by GPU for high performance. Utilizing SIMT (single instruction, multiple threads) parallelism and achieving coalesced device memory accesses are our major intentions because they are unique advantages of GPU architecture. Our efforts are based on mathematical foundation, algorithms development, and effective implementations interacting with GPU hardware, which we will explain in this and the following section. Figure \ref{fig:framework-ripserpp}(b) gives a high-level structure of Ripser++, showing the components of Vietoris-Rips barcode computation offloaded to GPU. This parallelization technique is form of kernalization~\cite{bodlaender2009problems}. We will elaborate on our mathematical and algorithmic contributions in this section.
	
	
\subsection{Matrix Reduction}
\label{sec:matrix_reduction}

    Matrix reduction (see Algorithm \ref{alg:standard-algorithm}) is a fundamental component of computing Rips barcodes. Its computation can be highly skewed \cite{zhang2019hypha}, particularly involving very few columns for column additions. Part of the reason for the skewness is the existence of a large number of ``apparent pairs." We prove and present the Apparent Pairs Lemma and a GPU algorithm to find apparent pairs in an implicitly represented coboundary matrix. We then design and implement a 2-layer data structure that optimizes the performance of the hashmap storing persistence pairs for subsequent matrix reduction on the non-apparent columns, which we term ``submatrix reduction.'' The design of the 2-layer data structure is centered around the existence of the large number of apparent pairs and their usage patterns during submatrix reduction.
\subsection{The Apparent Pairs Lemma}
\begin{figure}[h]
\centering
		\includegraphics[width=0.9\columnwidth]{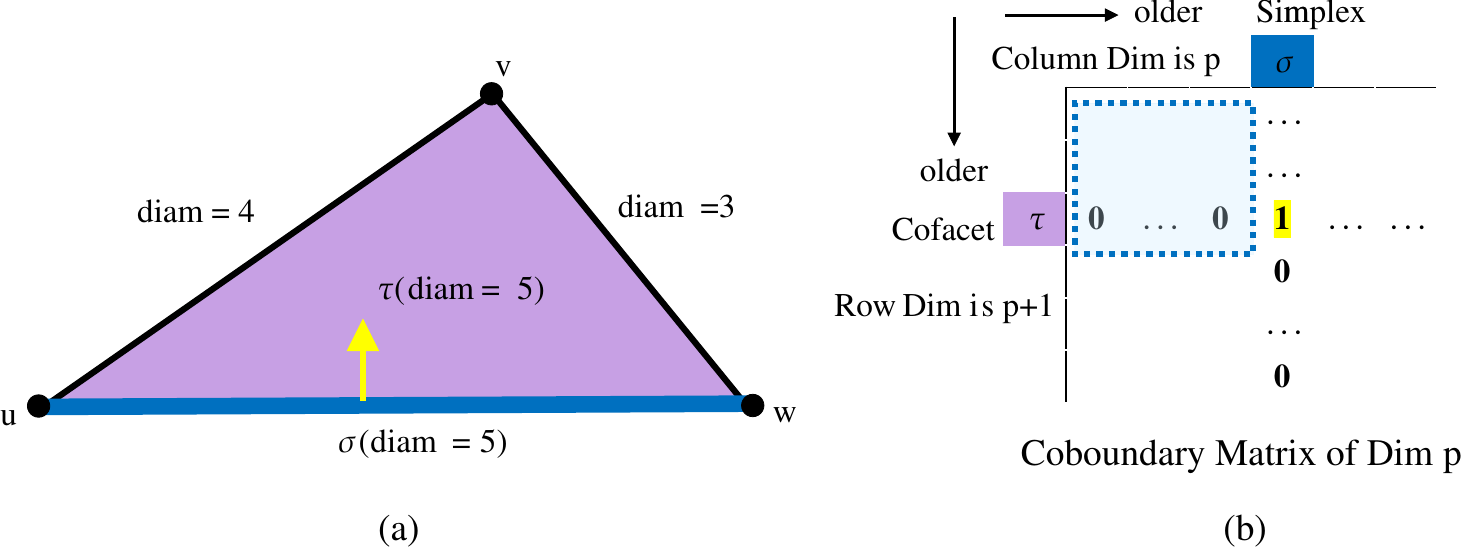}
		\caption{ \small{(a) A dimension $1$ 0-persistence apparent pair $(\sigma,\tau)$ on a single 2-dimensional simplex. $\sigma$ is an edge of diameter $5$ and $\tau$ is a cofacet of $\sigma$ with diameter $5$. The light arrow denotes the pairing between $\sigma$ and $\tau$. (b) The highlighted submatrix of the dimension $p$ coboundary matrix  denotes the dependency of column $[\sigma]$ from Algorithm \ref{alg:standard-algorithm}. In the dimension $p$ coboundary matrix, $(\sigma,\tau)$ is an apparent pair iff entry $(\tau,\sigma)$ has all zeros to its left and below. We color columns/simplices $\sigma$ with blue and their oldest cofacet $\tau$ with purple in (a) and (b). }}
		\label{fig:apparentpairs}
		\centering
	\end{figure}

	Apparent pairs of the form $(\sigma,\tau)$ are particular kinds of persistence pairs, or pairs of ``creation" and ``destruction" simplices as mentioned in Sections \ref{sec: clearing-lemma} and \ref{sec: preliminaries-simplexwise} (see also \cite{edelsbrunner2000topological}). In particular they are determined only by the (co-)boundary relations between $\sigma$ and $\tau$ and the simplex-wise filtration ordering of all simplices.
	
	Apparent pairs \cite{bauer2019ripser} show up in the literature by other names such as close pairs \cite{delgado2014skeletonization} or obvious pairs \cite{henselman2016matroid} but all are equivalent. Furthermore, apparent pairs are known to form a discrete gradient of a discrete Morse function \cite{bauer2019ripser} and have many other mathematical properties. 

	\begin{definition}
\label{def:apparent_pair}

    A pair of simplices $(\sigma,\tau)$ is an apparent pair iff:
    
    1. $\sigma$ is the youngest facet of $\tau$ and 
    
    2. $\tau$ is the oldest cofacet of $\sigma$.
\end{definition}

    We will use the simplex-wise order of Section \ref{sec: filtration-order} for Definition \ref{def:apparent_pair}. In a (co-)boundary matrix, a nonzero entry having all zeros to its left and below is equivalent to an apparent pair. Thus apparent pairs do not require the column reduction of Algorithm \ref{alg:standard-algorithm} to be determined. We call a column containing such a nonzero entry as an apparent column. An example of an apparent pair geometrically and computationally is shown in Figure \ref{fig:apparentpairs}. Furthermore, apparent pairs have zero persistence in Rips filtrations by property 1 of Definition \ref{def:apparent_pair} and that the diameter of a simplex is determined by its maximum length edge.
	
	
	
	In the explicit matrix reduction, where every column of $R$'s nonzeros are stored in memory (see Algorithm \ref{alg:standard-algorithm}), it is easy to determine apparent pairs by checking the positions of $\sigma$ and $\tau$ in the (co-)boundary matrix. However, in the implicit matrix reduction used in Ripser and Ripser++, we need to enumerate cofacets $\tau$ from $\sigma$ and facets $\sigma'$ from $\tau$ at runtime. It is necessary to enumerate both, because even if $\tau$ is found as the oldest cofacet of $\sigma$, $\tau$ may have facets younger than $\sigma$. 
	
	We first notice a property of the facets of a cofacet $\tau$ of simplex $\sigma$ where $\textsf{diam}(\sigma)=\textsf{diam}(\tau)$. Equivalently, we find a property of the nonzeros along particular rows of the coboundary matrix, allowing for a simple criterion for the order of the nonzeros in a row based on simply computing the diameter and combinatorial index of related columns:
	
	\begin{proposition}
		\label{prop: apparent-support}
	Let $\tau$ be the cofacet of simplex $\sigma$ with $\textsf{diam}(\sigma)=\textsf{diam}(\tau)$. 
	
	$\sigma'$ is a strictly younger facet of $\tau$ than $\sigma$ iff
	
	1. $\textsf{diam}(\sigma')=\textsf{diam}(\sigma)=\textsf{diam}(\tau)$ and
	
	2. $\textsf{cidx}(\sigma')<\textsf{cidx}(\sigma)$. ($\sigma'$ is strictly lexicographically smaller than $\sigma$)
	\end{proposition}
	
	\begin{proof}
	  
	($\Longrightarrow$)
	$\sigma'$ as a facet of $\tau$ implies that $\textsf{diam}(\sigma')\leq \textsf{diam}(\tau) = \textsf{diam}(\sigma)$. If $\sigma'$ is strictly younger than $\sigma$, then $\textsf{diam}(\sigma') \geq \textsf{diam}(\sigma)$. Thus 1. $\textsf{diam}(\sigma')=\textsf{diam}(\sigma)=\textsf{diam}(\tau)$.
    
    Furthermore, if $\sigma'$ is strictly younger than $\sigma$ and $\textsf{diam}(\sigma')=\textsf{diam}(\sigma)$, then the only way for $\sigma'$ to be younger than $\sigma$ is if 2. $\textsf{cidx}(\sigma')<\textsf{cidx}(\sigma)$.
	
	($\Longleftarrow$)
	If $\textsf{diam}(\sigma')=\textsf{diam}(\sigma)=\textsf{diam}(\tau)$ and $\textsf{cidx}(\sigma')<\textsf{cidx}(\sigma)$ then certainly $\sigma'$ is a strictly younger facet of $\tau$ than $\sigma$ is as a facet of $\tau$.
	\end{proof}

	We propose the following lemma to find apparent pairs:
	
	\begin{lemma}[The Apparent Pairs Lemma]
		\label{lemma: apparent pairs lemma}
		
	Given simplex $\sigma$ and its cofacet $\tau$,
		
	1. $\tau$ is the lexicographically greatest cofacet of $\sigma$ with $\textsf{diam}(\sigma)=\textsf{diam}(\tau)$, and
		
	2. no facet $\sigma'$ of $\tau$ is strictly lexicographically smaller than $\sigma$ with $\textsf{diam}(\sigma')=\textsf{diam}(\sigma)$,
		
	iff $(\sigma,\tau)$ is an apparent pair.
	\end{lemma}
	\begin{proof}
	($\Longrightarrow$) Since $\textsf{diam}(\tau) \geq \textsf{diam}(\sigma)$ for all cofacets $\tau$, Condition 1 is equivalent to having chosen the cofacet $\tau$ of $\sigma$ of minimal diameter at the largest combinatorial index, by the filtration ordering we have defined in Section \ref{sec: filtration-order}; this implies $\tau$ is the oldest cofacet of $\sigma$.  
	
	Assuming Condition 1, by the negation of the iff in Proposition \ref{prop: apparent-support}, 
	there are no simplices $\sigma'$ with $\textsf{diam}(\sigma')=\textsf{diam}(\sigma)=\textsf{diam}(\tau)$ and $\textsf{cidx}(\sigma')<\textsf{cidx}(\sigma)$ iff $\sigma$ is the youngest facet of $\tau$.
	
	($\Longleftarrow$) If $\textsf{diam}(\tau)>\textsf{diam}(\sigma)$ then there exists a younger $\sigma'$ with same cofacet $\tau$ and thus $\sigma$ is not the youngest facet of $\tau$. Thus $(\sigma,\tau)$ being an apparent pair implies Condition 1. Furthermore, $(\sigma,\tau)$ being apparent ($\sigma$ the youngest facet of $\tau$) with Condition 1 ($\textsf{diam}(\sigma)=\textsf{diam}(\tau)$) implies Condition 2 by negating the iff in Proposition \ref{prop: apparent-support}.
	
	Thus (Conditions 1 and 2) is equivalent to Definition \ref{def:apparent_pair}. 
	\end{proof}
    The Apparent Pairs Lemma can be simplified when $p=1$ and the metric has no equidistant pairs of points. 
    In the following lemma we show that it suffices to check condition (1) of the apparent pair condition for the case of $p=1$ which involves matching edges with triangles.
        \begin{lemma}[Sufficient Condition for an Apparent Pair]\label{lemma: apparentpair-diam-peq1}

    Let $R \in \mathbb{R}\cup\{\infty\}$ and let $d$ be the metric on a set  $X$ with $n_0$ points \textbf{no pair of which are equidistant} and assume $p=1:  p\leq d_{amb}$. 
    
    
    For any $\sigma= \{x_{k_1},...,x_{k_{p+1}}\} \in z_p$, associated with  some $p$-dimensional creation event for  $z_p \in \textsf{ker}(\partial_p)$, 
    
    There exists a $\tau \in C_{p+1}(K_n): \partial_{p+1}(\tau) \in \textsf{im}(\partial_{p+1})$ with 
    $(\sigma,\tau)$ an apparent pair if any of the following equivalent conditions hold true:
    \begin{enumerate}
    \item $ \tau \supseteq \sigma, \text{ s.t. } \textsf{diam}(\tau) = \textsf{diam}(\sigma)$
        \item 
    $ \exists x_{k_{p+2}} \in X:  \text{ s.t. }\max_{i=1}^{p+1}(d(x_{k_i},x_{k_{p+2}})) \leq \textsf{diam}(\sigma) \leq R$
    \end{enumerate}

    \end{lemma}
    \begin{proof}
    
    We show that the first condition is sufficient and then that the second condition of the lemma is equivalent to the first condition.
    
    \textbf{The first condition is sufficient: }

    According to Definition \ref{def:apparent_pair}, an apparent pair 
    \begin{equation}
    (\sigma,\tau) \in (\textsf{ker}(\partial_p)\cap {X \choose p+1}) \times (\{\tau: \partial_{p+1}(\tau) \in \textsf{im}(\partial_{p+1})\}\cap {X \choose p+2}))
    \end{equation}
    has the property that:
    \begin{enumerate}
        \item $\tau$ is the oldest cofacet of $\sigma$ and
        \item there is no other $\sigma'$ younger than $\sigma$ which has the same oldest cofacet.
    \end{enumerate} 

    \textbf{Satisfying Property 1 of an Apparent Pair: }
    
    We show that $\textsf{diam}(\tau) \leq \textsf{diam}(\sigma)$ for some $\tau: \tau\supseteq \sigma$ implies the first condition for $(\sigma,\tau)$ to be an apparent pair:
    
    Let $\tau^*$ be some cofacet of $\sigma$ amongst all $\tau$ with $\textsf{diam}(\tau)\geq \textsf{diam}(\sigma)$ that results in $\textsf{diam}(\tau^*)= \textsf{diam}(\sigma)$. This is the first property of an apparent pair.

    \textbf{Satisfying Property 2 of an Apparent Pair: }
    
    In order to satisfy the second property of an apparent pair, the $\sigma$ cannot have any other $\sigma'$ younger than $\sigma$ sharing $\tau^*$ as its oldest cofacet. For $p=1$ there cannot be a $\sigma'$ when 
    \begin{equation}
        \textsf{diam}(\tau^*)=\textsf{diam}(\sigma)
    \end{equation}
    since $\sigma'$ would have to have 
    \begin{equation}
        \textsf{diam}(\sigma')<\textsf{diam}(\sigma)=\textsf{diam}(\tau^*)
    \end{equation}
    However, this would mean that $(\sigma',\tau)$ cannot form an apparent pair and thus there are no potential conflicts.
    
    If we can find a tie-breaking order that can replace (3) from Section \ref{sec: filtration-order}, the latest amongst $\tau^*$ satisfying $\textsf{diam}(\tau^*)=\textsf{diam}(\sigma)$ becomes an oldest cofacet of $\sigma$. This would make $(\sigma,\tau)$ an apparent pair.
    
    Since both conditions are satisfied, the pair $(\sigma,\tau^*)$, is an apparent pair according to Definition \ref{def:apparent_pair}.
    
    \textbf{The second  condition is equivalent to the first: }
    
    Expanding the diameter function on $\tau$ and $\sigma$ as:
    \begin{equation}
        \textsf{diam}(\tau)= \max_{i,j=1}^{p+2}(d(x_{k_i},x_{k_{j}}))
    \end{equation} and 
    \begin{equation}
        \textsf{diam}(\sigma)= \max_{i,j=1}^{p+1}(d(x_{k_i},x_{k_{j}})),
    \end{equation}
    the equation $\textsf{diam}(\tau)\leq \textsf{diam}(\sigma)$ is equivalent to:
    \begin{subequations}
    \begin{equation}
        \max_{i,j=1}^{p+2}(d(x_{k_i},x_{k_{j}})) \leq \max_{i,j=1}^{p+1}(d(x_{k_i},x_{k_{j}}))\text{ iff }
    \end{equation}
    \begin{equation}
    \max_{i,j=1}^{p+2}(d(x_{k_i},x_{k_{j}}))=\max(\max_{i=1}^{p+1}(d(x_{k_i},x_{k_{p+2}})), \max_{i,j=1}^{p+1}(d(x_{k_i},x_{k_{j}}))) \leq  \max_{i,j=1}^{p+1}(d(x_{k_i},x_{k_{j}}))\text{ iff }
    \end{equation}
    \begin{equation}
        \max_{i=1}^{p+1}(d(x_{k_i},x_{k_{p+2}})) \leq \max_{i,j=1}^{p+1}(d(x_{k_i},x_{k_{j}}))
    \end{equation}
    \end{subequations}
    For any $\sigma \in \textsf{ker}(\partial_p)$, if any point $x \in X$ with $\max_{i=1}^{p+1}d(x,x_{k_i}) \leq \textsf{diam}(\sigma) \leq R$ then amongst all such $x \in X$, we can form a cofacet $\tau: \tau\supseteq \sigma$ with points $\tau=\{x_{k_1},...,x_{k_{p+1}},x^* \}$ that maximizes $\textsf{diam}(\tau)$.
    \end{proof}
    
    In Lemma \ref{lemma: apparentpair-diam-peq1} the proof required that for $p=1$ there would be a unique face of $\tau^*$ forming an apparent pair.
\begin{remark}
     Of course, when $p\geq 2$, this uniqueness of the $p$-dimensional face may not hold anymore. The second condition of the Apparent Pairs Lemma breaks this tie. Thus, the first condition is enough to show the \textbf{existence} of an apparent pair. 
\end{remark}
An elegant algorithmic connection between apparent pairs and the GPU is exhibited by the following Corollary.
	\begin{corollary}
	\label{corollary: massiveparallel-apparent}
	The Apparent Pairs Lemma can be applied for massively parallel operations on every column $[\sigma]$ of the coboundary matrix. 	
	\end{corollary}
    \begin{proof}
	We may generate the cofacets of simplex $\sigma$ and facets of cofacet $\tau$ of $\sigma$ independently with other simplices $\sigma' \neq \sigma$.
	\end{proof}
The effectiveness of the Apparent Pairs Lemma hinges on the following three phenomenon:
    \begin{enumerate}
        \item Theoretically there are favorable bounds on the number of apparent pairs for a full Rips complex induced by $n_0$ points. See Section \ref{sec: apparentbounds}.
        \item An important empirical fact and common dataset property: there are a lot of apparent pairs \cite{zhang2019hypha, bauer2019ripser}. In fact, by Table \ref{tab:all-pairs} in Section \ref{sec: experiments}, in many real world datasets up to 99\% of persistence pairs over all dimensions are apparent pairs. 
        \item According to our synthetic experiments on random distance matrices with unique entries in Section \ref{sec: experiments-apparent} and our  model of approximation to the expected number of apparent pairs. These results further confirm that there are a large number of apparent pairs in a (co-)boundary matrix induced by a Rips filtration.
    \end{enumerate}
    \subsubsection{Nearest Neighbor of a Simplex}
    We would like to rewrite the second sufficient condition of Lemma \ref{lemma: apparentpair-diam-peq1} in more geometric terms. 
    
    In metric geometry the concept of a ``nearest neighbor" of a point $x$ is defined as the point that is closest to the point $x$. In terms of a GMST, we know that the edge $e$ formed by every unordered pair of points $\{p,q\}: p<q$, where $p$ is the  \textbf{nearest neighbor} of  $q$, must belong to the GMST according to Kruskal's algorithm~\cite{kruskal1956shortest}. These are the ``apparent pairs" in the sense of persistent homology. The pairs $(q,e)$ are the first point-edge pairs that must form when computing $H_0(\partial_{1})$. 

    In relation to the second sufficient condition of Lemma \ref{lemma: apparentpair-diam-peq1}, we can generalize this definition of a nearest neighbor to its higher order generalization as follows:
    \begin{definition}
        In a metric space $\gX$ with metric $d$, for a set of points $S \subseteq \gX$ of size $k\geq 1$ and for an integer $j: j\leq k$, a $j$th-\textbf{order nearest neighbor} of $S$ is defined as:
        \begin{equation}
            \text{NN}_j(S)\triangleq\arg\min_{x \in \gX}(d(x,x_{(j)}))
        \end{equation}
        where $x_{(j)}$ is the j-th order point in $S$, meaning:
        \begin{equation}
            x_{(1)}<\cdots <x_{(j)}<\cdots <x_{(n)}: x_{(i)} \in S, i=1,...,k
        \end{equation}
    \end{definition}
    Thus, a $k$th-order nearest neighbor of $S$ is the closest point to \emph{all} the points in $S$.
    
    This gives us the following corollary to Lemma \ref{lemma: apparentpair-diam-peq1}.
    \begin{corollary}
    Assuming the conditions of Lemma \ref{lemma: apparentpair-diam-peq1}, namely $p=1$: 

     $(\sigma,\tau)$ is an apparent pair if:
        \begin{equation}
            d(\text{NN}_k(\sigma), x_i)  \leq \textsf{diam}(\sigma), \forall x_i \in \sigma
        \end{equation}
    \end{corollary}
    \begin{proof}
        This directly follows by the second condition of Lemma \ref{lemma: apparentpair-diam-peq1}.
    \end{proof}
    Thus, apparent pairs $(\sigma,\tau)$ for $\sigma$ a $1$-dimensional simplex can be interpreted as a \emph{``conditional 2nd-order nearest neighbor criterion"} for $\sigma$. 

    \section{Theoretical Bounds on the Number of Apparent Pairs}
	\label{sec: apparentbounds}
    We show a general upper and lower bound on the number of apparent pairs for any dimension $p\geq 1$. These bounds are tight. We then show a probabilistic lower bound for the case of Uniform Sampling. 
    \subsection{The General Upper and Lower Bound}
	Besides being an empirical fact for the existence of a large number of apparent pairs (see Section \ref{sec: experiments}), we show theoretically that there are tight upper and lower bounds to the number of apparent pairs for a dimension $p+1$ full Rips filtration on a  fully connected graph of $n_0$ points $X$ when the simplices containing the maximally indexed point  $n_0-1$ occur oldest in the filtration. We assume the simplex-wise filtration order of Section \ref{sec: filtration-order} throughout this paper. First we prove a useful property of apparent pairs on the filtered $\textsf{Rips}(X)$ depending on the largest indexed vertex. 
	
	\begin{lemma}[Cofacet Lexicographic Property]
	\label{lemma: cofacet-property}
	For any pair of $p$-dimensional simplex and its oldest cofacet: $(\sigma,\tau)$ of a dimension $p+1$ full Rips filtration  generated by a fully connected graph on $n_0$ points $V=\{0,...,n_0-1\}$ where $\sigma$ does not contain the maximum indexed vertex $n_0-1$ and all $p$-dimensional simplices $\sigma'$ containing maximum indexed vertex $n_0-1$ have $\textsf{diam}(\sigma')\leq \textsf{diam}(\sigma)$, $\tau=(v_{p+1},...,v_0)$ must have $v_{p+1}=n_0-1$, where 
	$v_{p+1}>v_{p}>...>v_0$.
	\end{lemma}
	\begin{proof}
	
	Let $\sigma=(w_d,...,w_0)$, with $w_d>w_{d-1}>...>w_0$ and $w_d \neq n_0-1$; The maximum indexed vertex $n_0-1$ forms a cofacet $t=(n_0-1,\sigma)$ of $\sigma$. If $\textsf{diam}(\tau)>\textsf{diam}(\sigma)$, then some facet $\sigma'$ of $\tau$ must be strictly younger (having larger diameter) than $\sigma$ and must contain the vertex $n_0-1$. $s'=(n_0-1,s'')$, with $s''$ a facet of $\sigma$. However, $\sigma'$ is actually strictly older than $\sigma$ by the fact that $\textsf{diam}(\sigma')\leq \textsf{diam}(\sigma)$ by assumption and $\textsf{cidx}(\sigma')>\textsf{cidx}(\sigma)$ since $n_0-1$ is the largest vertex index. 
	This is a contradiction. Thus $\textsf{diam}(\tau)=\textsf{diam}(\sigma)$ (since also $\textsf{diam}(\sigma) \leq \textsf{diam}(\tau)$) and $\textsf{cidx}(\tau)$ is largest amongst all cofacets of $\sigma$ by the existence of point $n_0-1$ in $\tau$. Thus the oldest cofacet of $\sigma$ is in fact $\tau$.
	\end{proof}
	
	\begin{theorem}\label{theorem: apparent bounds}(Equal Diameter Bound on the Number of Apparent Pairs)
		
    
    Let $X$ be a point set of size $n_0$. For a full Rips-filtration $\textsf{Rips}_{\infty}(X)=(\textsf{Rips}_t(X))_{t < \infty}$,
    
    If $p$-dimensional simplices $\sigma' \in \textsf{Rips}_{\infty}(X)$ containing some vertex $x \in X$ have $\textsf{diam}(\sigma')\leq \textsf{diam}(\sigma)$ for all $p$-dimensional simplices $\sigma \in \textsf{Rips}_{\infty}(X)$ not containing vertex $x \in X$, 
    
    Then the ratio of the number of $(p,p+1)$-dimensional apparent pairs to the total number of $p$-dimensional simplices satisfies the following upper and lower bounds: 
    \begin{itemize}
    \item upper bound: $\frac{(n_0-p-1)}{n_0}$; (tight for all $n_0\geq  p+1$ and $p \geq 1$).
    \item lower bound: $\frac{1}{(p+2)}$; (tight for $p \geq 1$).
    \end{itemize}
	\end{theorem}
	\begin{proof}
	
	\textbf{Upper Bound}:
	
	If $p$-dimensional simplex $\sigma$ has vertex $x$, so do all its cofacets $\tau$. Thus by Lemma \ref{lemma: cofacet-property}, the set of $(p,p+1)$-dimensional apparent pairs $(\sigma,\tau)$ must have $\tau=(v_{p+1},...,v_0)$ with $v_{p+1}=x$. There are at most ${n_0-1 \choose p+1}$ such $p$-dimensional simplices $\sigma$ by counting all possible tuples $\sigma=(v_p,...,v_j,...,v_0)$, $v_j \in \{0,...,n_0-1\} \setminus \{x\}$, $\sigma$ a facet of $\tau$ not including point $x$. Since there are a total of ${n_0 \choose p+1}$ $p$-dimensional simplices, we divide the two factors and obtain $\frac{{n_0-1 \choose p+1}}{{n_0 \choose p+1}}=\frac{(n_0-p-1)}{n_0}$ percentage of $p$-dimensional simplices, each of which are paired up with a $(p+1)$-dimensional simplex as an apparent pair. This percentage forms an upper bound on the number of apparent pairs.
	
	\textbf{Tightness of the Upper Bound}:
	
	We show that the upper bound is achievable in the special case where all diameters are equal. The condition of the theorem is certainly still satisfied. In this case we break ties for the simplex-wise refinement of the Rips filtration by considering the lexicographic order of simplices on their decreasing sorted list of vertex indices (See Section \ref{sec: combinatorial_number_system}). 
	
	We exhibit the upper bounding case by forming the corresponding coboundary matrix. The vertex $x$ is set to permuted to $n_0-1$. This does not affect the algorithm since it is permutation invariant on the indices of the vertices.  
	
	By the lexicographic ordering on simplices, in the coboundary matrix there exists a staircase (moving down and to the right by one row and one column) of apparent pair entries starting from the pair $(s_1,t_1)=((p,...,1,0),(n_0-1,p,...,1,0))=(s_1,(n_0-1,s_1))$ and ending on the pair $(s_{{n_0-1 \choose p+1}},t_{{n_0-1 \choose p+1}})=((n_0-2,n-3,...,n_0-p-2),(n_0-1,n_0-2,...,n_0-p-2))=(s_{{n_0-1 \choose p+1}},(n_0-1,s_{{n_0-1 \choose p+1}}))$. See Figure \ref{fig: max-equaldiam}, for the case of $n_0=5$ and $p=1$, where the first ${n_0-1 \choose p+1}$ columns are all apparent. 
	
	The staircase certainly is made up of apparent pairs since each entry has all zeros below and to its left, being the lowest nonzero entries of the leftmost columns of the coboundary matrix. Furthermore, the staircase spans all possible apparent pairs since all ${n_0-1 \choose p+1 }$ rows (the upper bound on number of apparent pairs) of the form $(n_0-1,\sigma')$ are apparent, $\sigma'$ an arbitrary simplex on the points $\{0,...,n_0-2\}$.
	
	\begin{figure}[h]
	\includegraphics[width=1.0\columnwidth]{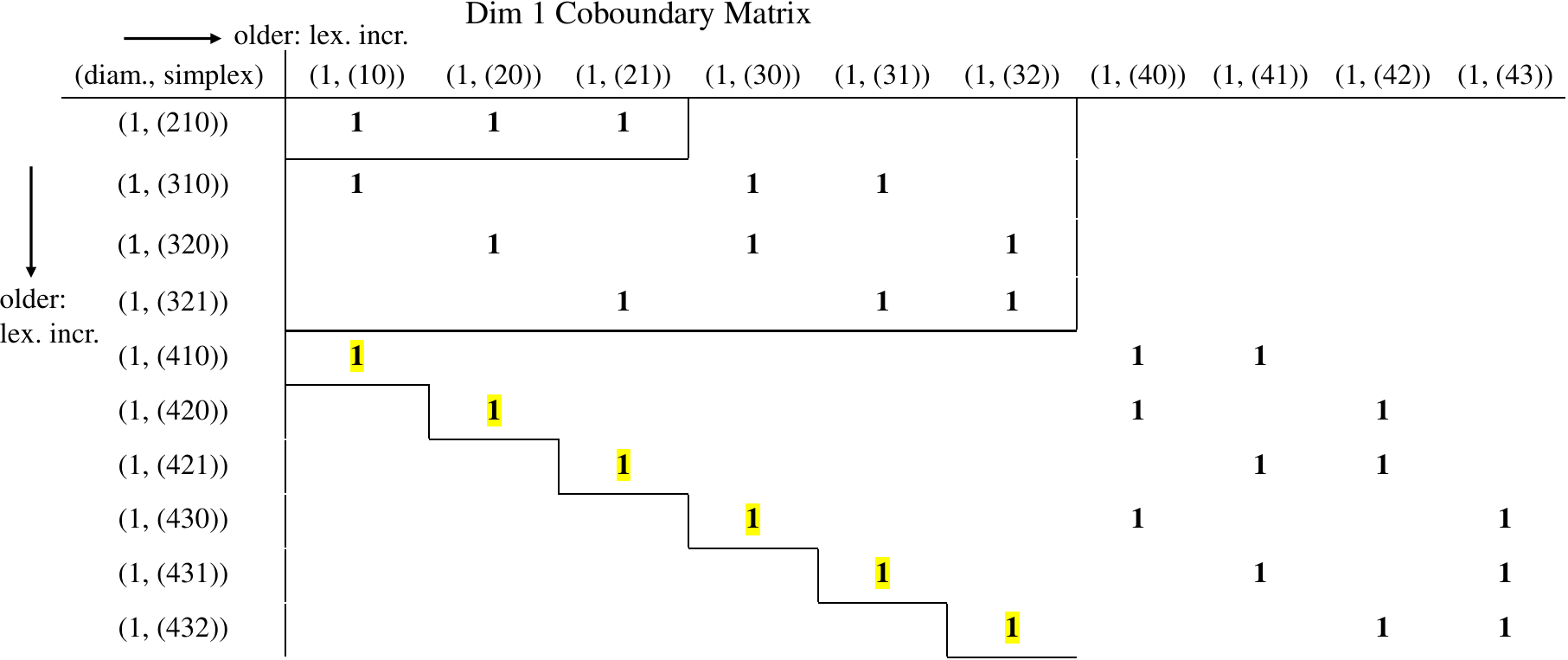}
	\caption{\small{A dimension 1 coboundary matrix of the full Rips filtration of the 2-skeleton on 5 points with all simplices of diameter 1. The yellow highlighted entries above the staircase correspond to apparent pairs.}}
	\label{fig: max-equaldiam}
	\end{figure}

    $\textbf{Lower bound}$:

	The largest number of cofacets of a given $p$-dimensional simplex must be $n_0-p-1$. Thus we will obtain a lower bound if each set of cofacets of the same diameter can be forced to be disjoint from one another. Thus we seek to find a minimal $k$ s.t. ${n_0 \choose p+2} \leq k \cdot (n_0-p-1)$. Upon solving for $k$, divide $k$ by ${n_0 \choose p+1}$, the number of $p$-dimensional simplices, and we get a lower bounding ratio of $\frac{1}{(p+2)}$ $p$-dimensional simplices being paired with $(p+1)$-dimensional simplices as apparent pairs.
 \hfill
 \break
	\textbf{Tightness of the Lower Bound}:
	
	For every dimension $p$, a $(p+1)$-dimension simplex has $p+2$ $p$-dimensional facets. There must be exactly one apparent pair of dimension $(p,p+1)$ in such a $(p+1)$-dimension simplex. For example, if $p=1$, the 2-simplex $X$ has one 1-simplex paired with it out of all $(3=p+2)$ 1-simplices in $X$ (see Figure \ref{fig:apparentpairs} and Figure \ref{fig: geometric-theoretical-upperbound}(a)). 
	\end{proof}
	
	\subsection{The Upper Bound Under Differing Diameters}\label{sec: upperbnd-differingdiam}
    
    The theoretical upper bound of Theorem \ref{theorem: apparent bounds} can still be achieved in dimension 1 with the following diameter assignments (distance matrix); let $d_1>d_2 > ... >d_{n_0\cdot(n_0-1)/2}>0$ be a sequence of diameters. We assign them in increasing lexicographic order on the 1-simplices ordered on the vertices sorted in decreasing order. For example, 1-simplex (10) with vertices 1 and 0 gets assigned $d_1$, 1-simplex (20) with vertices 2 and 0 gets assigned $d_2$. See the following distance matrix, Figure \ref{fig: max-noequaldiam}, for how the diameters are assigned.
    \begin{figure}[h]
	\includegraphics[width=1.0\columnwidth]{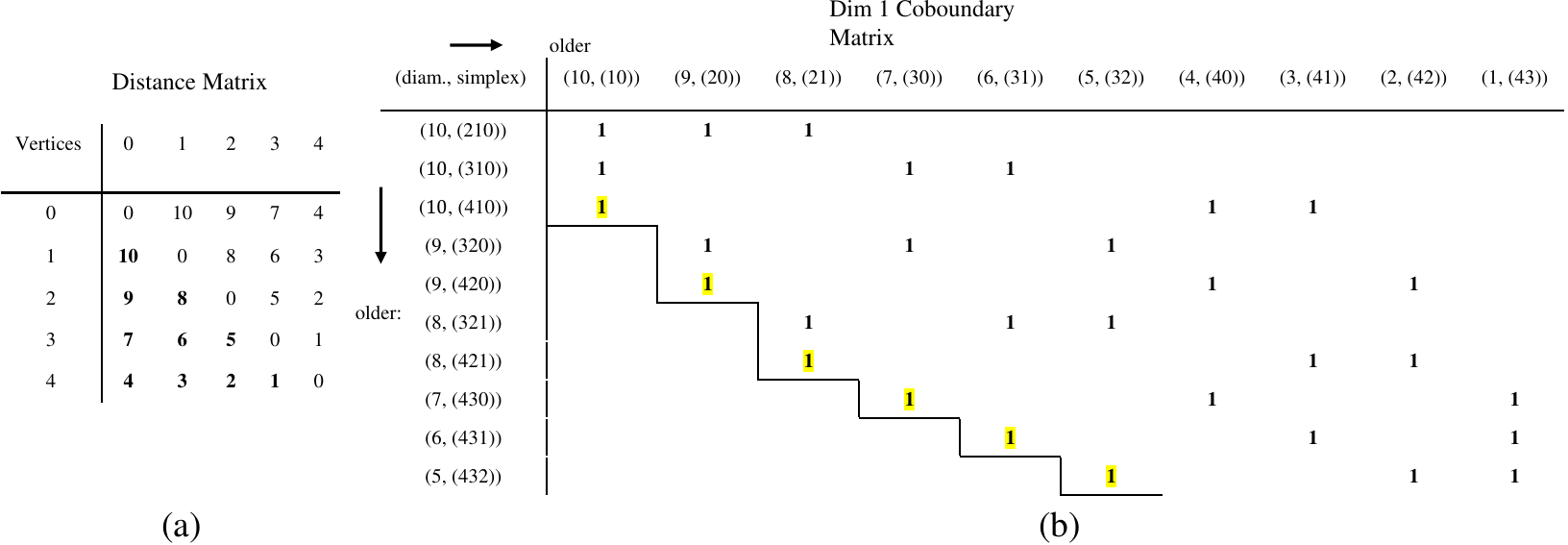}
	\caption{\small{(a) is an example distance matrix with differing edge diameters assigned in decreasing order for increasing lexicographic order on simplices. The barcodes are equivalent up to scaling so long as the distance matrix entries are in the same order (see Observation \ref{obs: reassignment} in Section \ref{sec: experiments-apparent}); thus we set the distance matrix entries to 1,...,10. (b) is a dimension 1 coboundary matrix of the full Rips filtration of the 2-skeleton on 5 points with simplices having diameters given in (a). The yellow highlighted entries above the staircase correspond to apparent pairs. Notice the coboundary matrix is a row permutation transformation from Figure \ref{fig: max-equaldiam}.}}
	\label{fig: max-noequaldiam}
	\end{figure}
 
	Since the increasing lexicographic order from the smallest lexicographically ordered simplex $(p,...,1,0)$ is still used on the columns $\sigma$, we still have the same leftmost ${n_0-1 \choose p+1}$ simplices/columns (but with a different diameter) as in the diameters all the same case (call this the original case). Furthermore, the oldest cofacet $\tau$ of the simplex/column $\sigma$ must still be the same $(p+1)$-dimensional cofacet $\tau$ of $\sigma$ originally. This follows by Lemma \ref{lemma: cofacet-property}, since we assume the lexicographic order on columns is preserved by the diameter assignment and thus that the simplices with vertex $n_0-1$ are oldest in the filtration (the diameter condition in Lemma \ref{lemma: cofacet-property} will be satisfied). These oldest cofacets are all different and are the same simplices as in the original case; Since the columns under consideration are leftmost, each cofacet $\tau$ has as youngest facet the $\sigma$. Thus all $(\sigma,\tau)$ are apparent pairs. Thus we preserve the same apparent pairs as in the original case.
    
    \textbf{The Relationship between the Optimal Upper Bound and the Combinatorial Index: }
    
        In Section \ref{sec: upperbnd-differingdiam}, there is a proof of the upper bound on the number of apparent pairs under differing diameters. In that section, we showed that for any $p\geq 1$ we have that the point of maximum index $n_0-1$ over $n_0$ points acts as a $(p+1)$-th order nearest neighbor for every $p$-dimensional $\sigma \subseteq {X \choose p+1}$.
        
        The distances between pairs of vertices are chosen in a way so that the lexicographic order and these distances are in opposite order. Thus, the largest indexed vertex acts as the tip of a ``fat cone" as in Figure \ref{fig: geometric-theoretical-upperbound}. Since this \textbf{cone} is determined by some \textbf{highest order nearest neighbor} to the previous points, it destroys all existing cycles with the faces of the cone itself.
   \subsection{Geometric Interpretation of the Upper Bound}
	Geometrically the theoretical upper bound in dimension 1 of Theorem \ref{theorem: apparent bounds} is illustrated in Figure \ref{fig: geometric-theoretical-upperbound}. The construction involves adding in a point at a time in order of its vertex index. 
	By the construction, we can alternatively count the number of apparent pairs on $n_0$ points, $T(n_0)$, by the following recurrence relation:
	\begin{equation}
	    T(n_0)= T(n_0-1)+n_0-2,
	\end{equation}
	where $T(3)=1$ and the $n_0-2$ term comes from the number of edges incident to the vertex with maximum index $n_0-2$ in the $(n_0-1)$-point subcomplex that become apparent when adding the new maximum point of index $n_0-1$ to the subcomplex. Solving for $T(n_0)$, we get $T(n_0)= {n_0-1 \choose 2}$ as in Theorem \ref{theorem: apparent bounds}. Thus, assuming the conditions of Theorem \ref{theorem: apparent bounds}, $n_0-2$ is maximal in the recurrence.
	\begin{figure}[h]
	    \includegraphics[width=0.9\columnwidth]{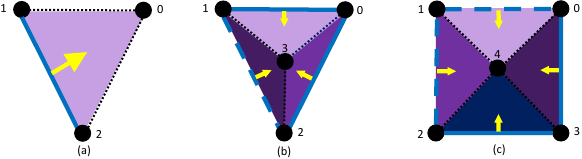}
		\caption{\small{Geometric interpretation of the theoretical upper bound in Theorem \ref{theorem: apparent bounds}. Edge distances are not to scale. (a),(b),(c) (constructed in this order) show the apparent pairs for $p=1$ on the planar cone graph centered around the newest apex point: $n_0-1$ for $n_0=3,4,5$ points. The yellow arrows denote the apparent pairs: blue edges paired with purple or navy triangles. The dashed (not dotted) blue edges denote apparent edges from the previous $n_0-1$ point subcomplex.}}
		\label{fig: geometric-theoretical-upperbound}
		\centering
	\end{figure}

	\subsection{Apparent Pairs are Heavy-Hitters for Large Point Samples}
    We show that the apparent pairs of the Vietoris-Rips barcode have a high probability of occurring for points sampled from a $d_{amb}$-dimensional hypercube with the uniform probability distribution.  Let us first understand how an apparent pair forms:

    \textbf{In terms of the Geometric Minimum Spanning Tree: }

    According to Section \ref{sec: 0-persistence}, for the computation of zero-dimensional persistent homology, we knew that $\{\tau: \partial_1(\tau) \in \textsf{im}(\partial_1(K_n))\}$ is the formal sum of edges from the geometric minimum spanning tree on $X$. The kernel: $\textsf{ker}(\partial_1(K_n))$ is represented by the formal sum of complementary edges of the spanning tree of $X$.

    For Vietoris-Rips complexes, only pairs of points and their distances determine any simplex. We can show that apparent pairs maintain complementary edges on the GMST of a set of $n_0$ points $X$.
    \begin{lemma}
        Let $X \subseteq \mathbb{R}^{ d_{amb}}$ be a set of $n_0$ points and $GMST(X)$ be the geometric minimum spanning tree on $X$.
        
        For any $p\geq 1$, the pair $(\sigma,\tau) \in (z_p,b_{p+1})$, with 
        \begin{equation}
            (z_p,b_{p+1})\in (\textsf{ker}(\partial_p) \cap {X \choose p+1}) \times (\{\tau: \partial_{p+1}(\tau) \in \textsf{im}(\partial_{p+1})\}\cap {X \choose p+2})
        \end{equation} is apparent, then 
        \begin{equation}
            u_{\sigma},v_{\sigma}= \arg\max_{u,v \in \sigma} d(u,v)
        \end{equation} 
        is a complementary edge.
    \end{lemma}
    \begin{proof}
    For each $p\geq 1: $
    
    If $(\sigma,\tau)$ is apparent, then according to Lemma \ref{lemma: apparentpair-diam-peq1}, we simply need $\textsf{diam}(\tau)\leq \textsf{diam}(\sigma)$ for some $p+1$-dimensional cofacet $\tau: \tau \supseteq \sigma$.

    This means that $\tau$ shares the same maximum distance pair of points $\{u_{\sigma},v_{\sigma}\} \subseteq \tau$ as $\sigma \subseteq \tau$. If this pair were not a complementary edge, then it would belong to the GMST. However, by definition of minimum spanning tree, all edges on the GMST are shorter than complementary edges.
    \end{proof}
    \textbf{The event of forming an Apparent Pair:}
    
    \begin{figure}[!ht]
    \centering
\begin{subfigure}{0.495\textwidth}
\centering
\begin{tikzpicture}

\coordinate (A) at (0, 0);      
\coordinate (B) at (4, 0);      
\coordinate (C) at (3, 2);      
\coordinate (D) at (1, 2);      

\draw[ultra thick] (B) -- node {$-$} (C);
\draw[ultra thick] (C) -- node {$\mid$} (D);
\draw[ultra thick] (D) -- node {$-$} (A);
\draw[thick, draw=red] (A) -- node {$\mid\mid$} (B);
\draw[thick, draw=red] (A) -- node {$\mid\mid$}  (C);
\draw[thick, draw=red] (B) -- node {$\mid\mid$}  (D);

\node[below] at (A) {$p_1$};
\node[below] at (B) {$p_2$};
\node[above] at (C) {$p_3$};
\node[above] at (D) {$p_4$};

\end{tikzpicture}
\caption{\small{A random trapezoid from $4$ i.i.d. sampled points}}
\end{subfigure}
\begin{subfigure}{0.495\textwidth}
\centering
\begin{tikzpicture}

\coordinate (A) at (0, 0);      
\coordinate (B) at (2, 0);      
\coordinate (C) at (2, 2);      
\coordinate (D) at (0, 2);      

\draw[thick] (A) -- (B) -- (C) -- (D) -- cycle;
\draw[ultra thick] (A) -- node{$\mid$} (B); 
\draw[ultra thick] (B) -- node{$-$} (C); 
\draw[ultra thick] (C) -- node{$\mid$} (D);
\draw[thick] (D) -- node{$-$} (A);
\draw[thick, draw=red] (A) -- node[pos=0.2]{$\mid\mid$}(C);        
\draw[thick, draw=red] (B) -- node[pos=0.2]{$\mid\mid$}(D);        

\node at (A) [below left] {$p_1$};
\node at (B) [below right] {$p_2$};
\node at (C) [above right] {$p_3$};
\node at (D) [above left] {$p_4$};

\end{tikzpicture}
\caption{\small{A random square from $4$ i.i.d. sampled points}}
    \end{subfigure}
\caption{\small{A set of $4$ i.i.d. points  $\{p_1,p_2,p_3,p_4\}$ sampled from some common distribution whose pair of points form edges. All the edges with the same number of ticks have the same length $p$. The edges with more than one tick 
have length $L'>L$. 
The left figure has the sample of points that form a trapezoid and the right figure has the sample of points form a square.
}}
\end{figure}
    Consider a random sample of $4$ points from the uniform distribution in $d_{amb}$ dimensions, a Rips-complex can form by either:
    \begin{enumerate}
        \item Adding diagonals before forming a full-Rips complex on $4$ points or
        \item Creating a $4$-cycle before closing with diagonals to form a full-Rips complex
    \end{enumerate}
    In the later, case (2), where the cycle cannot close up with its own apparent edge, the event would be considered spurious. For a large number of samples there should be no distinguishable spacing. In particular, with just one more point in the interior of the square, atleast one of the sides of the square becomes apparent. It would be suspected that with a large number of points these spurious cycles would not appear. 

    \subsubsection{Lower Bounds under the Uniform Sampling Condition}\label{sec: lowerbound-uniformsampling}
    
    Based on our understanding of apparent pairs in terms of the point samples, in the following we specify a generic sampling distribution in Euclidean space. 
    \begin{condition}[Hypercube Sampling Condition]\label{cond: uniform-independent-sampling}
    
    \begin{equation}
        P(x)\triangleq P(x \sim U([0,1]^{d_{amb}})) 
    \end{equation}
    $\bullet$ Assume the Euclidean metric over $[0,1]^{d_{amb}}$.
    \end{condition}
    We show that there is a lower bound on the probability of achieving an apparent pair for a given $p$-dimensional simplex $\sigma$. This lower bound holds in the case of uniform independent sampling with growing ambient dimension from Condition \ref{cond: uniform-independent-sampling}. 
    
    \begin{lemma}\label{lemma: apparent-probability}
    Let $n_0 \in \mathbb{N}$ and 
    letting $X$ be a point sample of size $n_0$ from the uniform sampling distribution of Condition \ref{cond: uniform-independent-sampling}. 
    
    For any $p$-simplex $\sigma \in \textsf{ker}(\partial_p) \cap {X \choose (p+1)}$,  
    
    We then have:
    \begin{equation}\label{eq: apparent-lowerbound}
           P(\exists \tau: (\sigma,\tau) \text{ is apparent} \mid \textsf{diam}(\sigma)=t)\geq 1-(1-\frac{t^{d_{amb}}}{(d_{amb}+1)!2^{p+1}})^{n_0}
        \end{equation}
    \end{lemma}
    \begin{proof}
        Using Lemma \ref{lemma: apparentpair-diam-peq1}, we know that given $p$-dimensional simplex $\sigma \in K_n$ with $\sigma=\{x_1,...,x_{p+1}\}$ if all $x \in X$ have that $\max_{i=1}^{p+1}d(x,x_i)\leq t$. Thus, the point $x \sim P(x)$ has the property of forming a $p+1$-dimensional simplex $\tau$ that pairs with $\sigma$ as an apparent pair if and only if each $x$ belongs to the intersection of the $p+1$ $t$-radius balls centered about $x_i, i=1,...,p+1$.

        Since $P(x)$ is a uniform distribution on the hypercube $([0,1])^{d_{amb}}$, we can take the complementary Euclidean volume of the intersection of $p+1$ balls of radius $t$ in $d_{amb}$ dimensions to get the probability of not having an apparent pairing $(\sigma,\tau)$ for some $\tau$. 
        This gives the complementary probability of a point belonging to this intersection: $1-P(\bigcap_{i=1}^{p+1} B(x_i,t))$. 
        
        If none of the $n_0-(p+1)$ points not in $\sigma$ belong to $\bigcap_{i=1}^{p+1} B(x_i,t)$, then there is no cofacet $\tau$ that makes an apparent pair $(\sigma, \tau)$. Thus, the complementary probability is a lower bound to  $P(\exists \tau: (\sigma,\tau) \text{ is apparent}\mid  \textsf{diam}(\sigma)=t)$.

        We now compute this complementary probability:

        We know that 
        \begin{equation}
            \textsf{BB}(\textsf{Cyl}(\sigma,2t)) \supseteq \textsf{Cyl}(\sigma,2t) \supseteq \bigcap_{i=1}^{p+1} B(x_i,t) \supseteq 
        \textsf{inscpoly}(\{x_i\}_{i=1}^{p+1}))
        \end{equation}
        where $\textsf{Cyl}(\sigma,2t) \triangleq \sigma \times [-t,t]$, $\textsf{BB}(\sigma)$ is the $d_{amb}$-dimensional bounding box of a set of points $S$, and $\textsf{inscpoly}(\{x_i\}_{i=1}^{p+1}))$ is a polytope that is contained in $\bigcap_{i=1}^{p+1} B(x_i,t)$. This polytope exists by Proposition \ref{prop: ball-intersection-polytope}.
        
        We have that:
        \begin{subequations}
           \begin{equation}
            \textsf{BB}(\textsf{Cyl}(\sigma,2t)) \leq \frac{(2t)^{d_{amb}}}{2^p}
        \end{equation}
        \begin{equation}
            \textsf{vol}(\textsf{Cyl}(\sigma,2t))= 2t \textsf{vol}(\sigma)
        \end{equation}
        \begin{gather}
        \begin{split}
            \textsf{vol}(\textsf{inscpoly}(
         = \frac{1}{(d_{amb}+1)!}(\frac{t}{\sqrt{2}})^{d_{amb}+1}\frac{2^{d_{amb}}}{2^p}\frac{1}{t}
        \end{split}
        \end{gather}
        \end{subequations}
        where in the inequalities for the inscribed polytope volume, the first lower bound is the volume of a $t$-equilateral $d_{amb}$-dimensional simplex with fixed $p$-dimensional face that is reflected across ${d_{amb}+1 \choose d_{amb}}-(p+1)$ faces. Such a union of simplices is inscribed in $\bigcap_{i=1}^{p+1} B(x_i,t)$
        where the multiplicative factor $2^{(d_{amb}-p)}$ in the lower bound on $\textsf{vol}(\textsf{inscpoly}(
        \{x_i\}_{i=1}^{p+1}))$ comes from the ${d_{amb}+1 \choose d_{amb}}-(p+1)=d_{amb}-p$ reflections on the faces of the $t$-equilateral $d_{amb}$-dimensional simplex. See Proposition \ref{prop: ball-intersection-polytope}.
        
        This results in the following probabilistic inequalities:
        \begin{equation}
            \frac{(2t)^{d_{amb}}}{2^p} \geq \textsf{vol}(\sigma) (2t) \geq \textsf{vol}(\bigcap_{i=1}^{p+1} B(x_i,t))\geq \frac{1}{(d_{amb}+1)!}(\frac{t}{\sqrt{2}})^{d_{amb}+1}\frac{2^{d_{amb}}}{2^p}\frac{1}{t}
        \end{equation}
        Thus 
        \begin{subequations}
        \begin{equation}
            \frac{(2t)^{d_{amb}}}{2^p}\geq P(\exists \tau: (\sigma,\tau) \text{ is apparent}\mid \textsf{diam}(\sigma)=t)= 1-(1-P(\bigcap_{i=1}^{p+1} B(x_i,t)))^{n_0-(p+1)} 
        \end{equation}
        \begin{equation}
            \geq 1-(1-\frac{1}{(d_{amb}+1)!}(\frac{t}{\sqrt{2}})^{d_{amb}+1}\frac{2^{d_{amb}}}{2^p}\frac{1}{t})^{n_0-(p+1)}
        \end{equation}
        \begin{equation}
            \geq 1-(1-\frac{t^{d_{amb}}}{(d_{amb}+1)!2^{p+1}})^{n_0}
        \end{equation}
        \end{subequations}
    \end{proof}

According to Lemma \ref{lemma: apparent-probability}, we thus have the following implication:
\begin{equation}\label{eq: sufficient-apparent}
    1-(1-\frac{t^{d_{amb}}}{(d_{amb}+1)!2^{p+1}})^{n_0}\geq q \Rightarrow P(\exists \tau: (\sigma,\tau) \text{ is apparent}\mid \textsf{diam}(\sigma)=t)\geq q
\end{equation}
Solving for $t$, this gives us the lower bound on all the $t$ so that the implication of Equation \ref{eq: sufficient-apparent} can hold:
\begin{equation}
    t_q^*\triangleq [((d_{amb}+1)!2^{p+1})(1-(1-q)^{\frac{1}{n_0}})]^{\frac{1}{d_{amb}}}
\end{equation}
Taking $n_0 \rightarrow \infty$, we see that $t_q^* \rightarrow 0$, so that the expression in Equation \ref{eq: apparent-lowerbound} gives a lower bound of $0$ probability. Intuitively this is because when there are unlimited number of points in a compact domain with no lower bound on the closest pair of points, then an arbitrarily large number of ``thin", non-apparent $p+1$-dimensional simplices $\sigma$ could occur. 

This is why we must condition on having a lower bound on the closest pair of points in any sample $X \sim P(x)^{n_0}$. Since $P(x)$ has compact support, this is equivalent to upper bounding the aspect ratio. Let $\rho= a(\textsf{Rips}_t(X))$, this gives the lower bound on $t$ as $\frac{1}{\rho}$.

We show in the following theorem a sufficient condition to obtain atleast probability of $q$ chance for any $\sigma$ of a $p$-dimensional cycle creation event of being part of an apparent pair.

We group together the terms in Equation \ref{eq: sufficient-apparent} involving $\rho$ into a single variable. This defined below:
    \begin{equation}\label{eq: Delta}
        \Delta\triangleq (d_{amb}+1)!2^{p+1}\rho^{(d_{amb})}
    \end{equation}
    
    \begin{theorem}
    \label{thm: apparent-probability-numsample}
    Under the same conditions as Lemma \ref{lemma: apparentpair-diam-peq1} where $X$ is a point sample of size $n_0$ from the uniform sampling of Condition \ref{cond: uniform-independent-sampling} and $\rho= a(\textsf{Rips}_t(X))$. 
    
    For any $q \in [0,1]$ 
    and for any $p$-simplex $\sigma  \in \textsf{ker}(\partial_p) \cap {X \choose (p+1)}$,
    
    Let:
    \begin{equation}
        n^*(q) \triangleq \frac{\log(1-q)}{\log(1-\frac{1}{\Delta})} 
    \end{equation}
We must have that 
\begin{equation}\label{eq: samplecomplexity-apparent}
    n_0\geq n^*(q) \Rightarrow P(\exists \tau: (\sigma,\tau) \text{ is apparent}) \geq q 
\end{equation}

Furthermore, if we enforce aspect ratio $\rho=1$ as in Theorem \ref{theorem: apparent bounds}, we obtain the optimal lower bound on $n^*(q)$.
    \end{theorem}
    \begin{proof}
        Using that $\Delta=(d_{amb}+1)!2^{p+1}\rho^{d_{amb}}$,     we have:
    \begin{equation}
         (1-(1-\frac{1}{\Delta})^{n_0}) \geq q \Rightarrow P(\exists \tau: (\sigma,\tau) \text{ is apparent}) \geq q 
    \end{equation}
    Bringing $n_0$ to one side, we obtain the desired implication from Equation \ref{eq: samplecomplexity-apparent}.
    \end{proof}
    \clearpage
    \section{Finding Apparent Pairs in Parallel on GPU}
	\begin{algorithm}[!h]
    \SetKwComment{Comment}{//}{}
    \caption{Finding Apparent Pairs on GPU}
    \label{alg:apparentpairsalgorithm}
    \KwIn{$\pmb{C}$: the simplices to reduce; $\textsf{vertices}(\cdot)$: the vertices of a simplex; $\textsf{diam}(\cdot)$: the diameter of a simplex; $\textsf{cidx}(\cdot)$: the combinatorial index of a simplex; $\textsf{dist}(\cdot)$: the distance between two vertices; $\textsf{enumerate}$-$\textsf{facets}(\cdot)$: enumerates facets of a simplex.}
    \KwOut{$\pmb{A}$: the apparent pair set from the coboundary matrix of dimension $\textsf{dim}$. \Comment*[r]{global to all threads}}
    
    $\text{tid}$: the thread id. \Comment{local to each thread}
    
    $\sigma \gets \pmb{C}[\text{tid}]$ 
    \Comment{Each thread fetches a distinct simplex from the set of simplices}
    $\pmb{V} \gets \textsf{vertices}(\sigma)$ 
    \Comment{This only depends on the combinatorial index of $\sigma$}
    
    \ForEach{{cofacet $\tau$ of $\sigma$ in lexicographically decreasing order}}{
        \ForEach{$v'$ in $\pmb{V}$}{
            $\textsf{diam}(\tau) \gets \max(\textsf{dist}(v', v),\, \textsf{diam}(\sigma))$ 
            \Comment{Calculate the diameter of $\tau$}
        \If{$\textsf{diam}(\tau) = \textsf{diam}(\sigma)$}{
            $\pmb{S} \gets \emptyset$ \;
            $\textsf{enumerate-facets}(\tau,\, \pmb{S})$ 
            \Comment{$\pmb{S}$ are facets of $\tau$ in lexicographical increasing order}
            
            \ForEach{$\sigma'$ in $\pmb{S}$}{
                \If{$\textsf{diam}(\sigma') = \textsf{diam}(\sigma)$}{
                    \If{$\textsf{cidx}(\sigma') = \textsf{cidx}(\sigma)$}{
                        \Comment{$\sigma$ is the youngest facet of $\tau$}
                        $\pmb{A} \gets \pmb{A} \cup \{(\sigma, \tau)\}$ \;
                        \textbf{return}
                        \Comment*[r]{Exit if $(\sigma,\tau)$ is apparent or if $\sigma'$ is strictly younger than $\sigma$}
                    }
                }
            }
        }
    }
    }
\end{algorithm}
    Based on Lemma \ref{lemma: apparent pairs lemma}, finding apparent pairs from a cleared coboundary matrix without explicit coboundaries becomes feasible. There is no dependency for identifying an apparent pair as Corollary \ref{corollary: massiveparallel-apparent} states, giving us a unique opportunity to develop an efficient GPU algorithm by exploiting the massive parallelism.
 
	Algorithm \ref{alg:apparentpairsalgorithm} shows how a GPU kernel finds all apparent pairs in a massively parallel manner. A GPU thread fetches a distinct simplex $\sigma$ from an ordered array of simplices, represented as a (diameter, cidx) struct in GPU device memory
	, and uses the Apparent Pairs Lemma to find the oldest cofacet $\tau$ of $\sigma$ and ensure that $\tau$ has no younger facet $\sigma'$. If the two conditions of the Apparent Pairs Lemma hold, $(\sigma,\tau)$ can form an apparent pair. Lastly, the kernel inserts into a data structure containing all apparent pairs in the
	GPU device memory. 
	
	The complexity of one GPU thread is $O(\log(n_0)(d+1) + (n_0-p-1)(d+1)$), in which $n_0$ is the number of points and $p$ is the dimension of the simplex $\sigma$. The first term represents a binary search for $p+1$ simplex vertices from a combinatorial index, and the second term says the algorithm checks at most $p+1$ facets of all $n_0-p-1$ cofacets of the simplex $\sigma$. Notice that this complexity is linear in $n_0$, the number of points, with dimension $p$ small and constant. 
	\subsection{Review of Enumerating Cofacets in Ripser}
	\label{sec: enumerating-cofacets}
	Enumerating cofacets/facets in a lexicographically decreasing/increasing order is substantial to our algorithm. The cofacet enumeration algorithm differs for full Rips computation and for sparse Rips computation. Cofacet enumeration is already implemented in Ripser. 
	Details of the cofacet enumeration from Ripser are presented in Algorithm \ref{alg:cofacets_vr} of Section \ref{sec: enumerating-cofacets}. For enumerating cofacets of a simplex in the sparse case, we utilize the sparsity of edge relations as in Ripser. 
	Algorithm \ref{alg:cofacets_sparse} shows such an enumeration. 

	In Algorithm \ref{alg:cofacets_vr}, we enumerate cofacets of a simplex $\sigma$ by iterating through all vertex indices $v \in X: X=\{0,...,n_0-1\}$. We keep track of an integer $k$. If $v$ matches a vertex of $\sigma$, then we decrement $k$ and $v$ until we can add ${v \choose k}$ legally as a binomial coefficient of the combinatorial index of a cofacet of $\sigma$. This algorithm assumes a dense distance matrix, meaning all its entries represent a finite distance between any two points in $X$.
 
		\begin{algorithm}[!h]
        \SetKwComment{Comment}{//}{}
    \caption{Enumerating Cofacets of a Simplex}
    \label{alg:cofacets_vr}
    
    \KwIn{$\pmb{X}=\{0,...,n_0-1\}$: the set of indices of a finite metric space; $\sigma$: a simplex with vertices in $\pmb{X}$; $\textsf{vertices}(\cdot)$: the vertices of a simplex; $\textsf{cidx}(\cdot)$: the combinatorial index of a simplex.}
    \KwOut{$\pmb{S}$: the facets of $\sigma$ in lexicographically decreasing order.}
    
    $\pmb{V} \gets \textsf{vertices}(\sigma)$ \;
    $\textsf{cidx}(\sigma'_{high}) \gets 0$ \;
    $\textsf{cidx}(\sigma'_{low}) \gets \textsf{cidx}(\sigma)$ \;
    
    \While{$v \in \pmb{X} = \{0,...,n_0-1\}$}{
        \uIf{$v \notin \pmb{V}$}{
            $\textsf{cidx}(\sigma') \gets \textsf{cidx}(\sigma'_{high}) + {v \choose k} + \textsf{cidx}(\sigma'_{low})$ \;
            $v \gets v - 1$ \;
        }
        \Else {
            \While{$v \in \pmb{V}$}{
                $\textsf{cidx}(\sigma'_{high}) \gets \textsf{cidx}(\sigma'_{high}) + {v \choose k+1}$ 
                \Comment{$\textsf{vertices}(\sigma'_{high}) \gets \textsf{vertices}(\sigma'_{high}) \cup \{v\}$}
                $\textsf{cidx}(\sigma'_{low}) \gets \textsf{cidx}(\sigma'_{low}) - {v \choose k}$ \Comment{$\textsf{vertices}(\sigma'_{low}) \gets \textsf{vertices}(\sigma'_{low}) \setminus \{v\}$}
                $v \gets v - 1$ \;
                $k \gets k - 1$ \;
            }
        }
        $\text{append}(\pmb{S}, \sigma')$ \;
    }
\end{algorithm}

	\subsection{Computation Induced by Sparse 1-Skeletons}
	Due to the exponential growth in the number of simplices by dimension during computation of Vietoris-Rips filtrations, which demands a high memory capacity and causes long execution time, we also consider enumerating cofacets induced by sparse distance matrices. A sparse distance matrix $D$ simply means that we consider certain distances between points to be ``infinite." This prevents certain edges from contributing to forming a simplex since a simplex's diameter must be finite. Computationally when considering cofacets of a simplex, we only consider finite distance neighbors. This results in a reduction in the number of simplices to consider when performing cofacet enumeration for matrix reduction and filtration construction, potentially saving both execution time and memory space if the distance matrix is ``sparse" enough. 
	 
	There are two uses of sparse distance matrices. One usage is to compute Vietoris-Rips barcodes up to a particular diameter threshold, using a sparsified distance matrix. This results in the same barcodes as in the dense case on a truncated filtration due to the diameter threshold. The second usage is to approximate a finite metric space at the cost of a multiplicative approximation factor on barcode lengths ~\cite{dey2019simba, sheehy13linear, cavanna2015geometric}. The paper \cite{cavanna2015geometric} has an algorithm to approximate dense distance matrices with sparse matrices for such barcode approximation and can be used with Ripser++; this approach is a particular necessity for large dense distance matrices where the persistence computation has dimension $p$ such that the resulting filtration size becomes too large.
	
	\subsection{Enumerating Cofacets of Simplices Induced by a Sparse 1-Skeleton in Ripser}
	
	The enumeration of cofacets for sparse edge graphs (sparse distance matrices involving few neighboring relations between all vertices) must be changed from Algorithm \ref{alg:cofacets_vr} for performance reasons. The sparsity of neighboring relations can significantly reduce the number of cofacets that need to be searched for. 
	Similar to the inductive algorithm described in \cite{zomorodian2010fast}, cofacet enumeration is in Algorithm \ref{alg:cofacets_sparse}.
    
	\begin{algorithm}[!h]
    \SetKwComment{Comment}{//}{}
    \SetKw{Continue}{\text{continue}}
    \caption{Enumerating Cofacets of Simplex for Sparse 1-Skeletons}
    \label{alg:cofacets_sparse}
    \KwIn{$\pmb{X}= \{0,...,n_0-1\}$: a finite metric space; $\sigma$: a simplex with vertices in $\pmb{X}$; $\textsf{vertices}(\cdot)$: the vertices of a simplex; $\textsf{cidx}(\cdot)$: the combinatorial index of a simplex; $\textsf{cidx}_{vert}(\cdot)$: calculate the combinatorial index from the vertices.}
    \KwOut{$\pmb{S}$: the facets of $\sigma$ in lexicographically decreasing order.}
    
    $\pmb{V} \gets \textsf{vertices}(\sigma)$ \;
    fix some $v_0 \in \pmb{V} \subseteq \pmb{X}$ \;
    
    \ForEach{\text{neighbor $v' \in \pmb{X} - \pmb{V}$ of $v_0$ where  $v' \neq v_0$ in decreasing order}}{
        \ForEach{\text{$w \in \pmb{V}$, $w \neq v_0$ and $w \neq v'$}}{
            \If{\text{$w$ is a neighbor of $v'$}}{
                \Continue
                \Comment*[r]{Jump to the next iteration of the inner loop}
            }
            \Else {
                \If{\text{all neighboring vertices to $w$ are greater than $v'$}}{
                    \Return 
                    \Comment*[r]{There are no more cofacets that can be enumerated}
                }
                \Else {
                    \textbf{goto} try\_next\_vertex 
                    \Comment*[r]{Jump to line 13}
                }
            }
        }
        $\sigma' \gets \textsf{cidx}_{vert}(V \cup \{v'\})$ \;
        \text{append}($\pmb{S}, \sigma'$) \;
        try\_next\_vertex: ;
    }
\end{algorithm}
\clearpage
	\subsection{Enumerating Facets in Ripser++}
	Algorithm \ref{alg:facets} shows how to enumerate facets of a simplex as needed in Algorithm \ref{alg:apparentpairsalgorithm}. A facet of a simplex is enumerated by removing one of its vertices. Due to properties of the combinatorial number system, if we remove vertex indices in a decreasing order, the combinatorial indices of the generated facets will increase (the simplices will be generated in lexicographically increasing order). Algorithm \ref{alg:facets} is used in Algorithm \ref{alg:apparentpairsalgorithm} for GPU in Ripser++ and does not depend on sparsity of the distance matrix since a facet does not introduce any new distance information. In fact, there are only $p+1$ facets of a simplex to enumerate. On GPU we can use shared memory per thread block to cache the few $p+1$ vertex indices. 
\begin{algorithm}[!h]
\SetKwComment{Comment}{//}{}
    \caption{Enumerating Facets of a Simplex}
    \label{alg:facets}
    
    \KwIn{$\pmb{X} = \{0,...,n_0-1\}$: $n_0$ points of a finite metric space; $\sigma$: a simplex with vertices in $\pmb{X}$; $\textsf{vertices}(\cdot)$: the vertices of a simplex; $\textsf{cidx}(\cdot)$: the combinatorial index of a simplex; $\textsf{last}(\cdot)$: the last simplex of a sequence.}
    \KwOut{$\pmb{S}$: the facets of $\sigma$ in lexicographically increasing order.}
    \SetKwProg{Fn}{Function}{:}{}
    \SetKwFunction{enumfacets}{enumerate-facets}
    \Fn{\enumfacets{$\sigma$, $\pmb{S}$}}{
        $\pmb{V} \gets \textsf{vertices}(\sigma)$ \;
        $\text{prev} \gets \emptyset$ \;
        $k \gets \lvert \pmb{V}\rvert $ \;
        
        \For{$v \in \pmb{V} \subseteq \pmb{X}$ in decreasing order}{
            \If{$\text{prev} \neq \emptyset$}{
                $\textsf{cidx}(\sigma') \gets \textsf{cidx}(\textsf{last}(\pmb{S})) - \binom{v}{k} + \binom{[ \text{prev} ] }{k}$ 
                \Comment*[r]{$[\text{prev}] $ is the only element of the singleton $\text{prev}$}
            }
            \Else {
                $\textsf{cidx}(\sigma') \gets \textsf{cidx}(\textsf{last}(\pmb{S})) - \binom{v}{k}$ \;
            }
            }
            \text{append}($\pmb{S}, \sigma'$) 
            \Comment*[r]{Append $\sigma'$ to the end of $\pmb{S}$}
            $\text{prev} \gets \{v\}$ \;
            $k \gets k - 1$ \;
        }
\end{algorithm}
\section{The Expected Complexity of Ripser++}
    We prove the expected complexity of Ripser++ under the Uniform Sampling Condition as well as the $k$-Cavity Uniform Sampling Condition. 
\subsection{Under Uniform Sampling Conditions}
    \begin{theorem}(Complexity Under Uniform Sampling)\label{thm: expectedcomplexity-ripserpp}
    
    Assuming the Uniform Sampling Condition,  when  $n_0$ satisfies 
        \begin{equation}
    {n_0}\geq \Omega(\max(C^{\frac{\log(1+\frac{1}{\Delta})\log(\Delta)}{(\log(1+\frac{1}{\Delta})\log(\Delta) -1) }}, (\textsf{dim}(\textsf{Rips}(X)))^2\log^2(\Delta))
    \end{equation}
    for some constant $C>0$, then the expected Betti numbers of the Rips complex satisfy $0\leq \mathbb{E}[\beta_p]\leq 1, \forall p\geq 1$
    
    This makes the expected work complexity of Ripser++: \begin{equation}
        O(n_{p}\log(n_0)+n_{p+1})
    \end{equation}
    with expected depth complexity of \begin{equation}
        O(n_{p}\log(n_0)+n_0 (p+1))
    \end{equation}
    which is the complexity of computing $\textsf{ker}([\partial_p])$.

    To compute the Betti number only, the expected complexity of Ripser++ is \begin{equation}
        O(n_{p+1})
    \end{equation}
    and the expected depth complexity is \begin{equation}
        O(n_0(p+1))
    \end{equation}
    \end{theorem}
    \begin{proof}
    Let:
        \begin{equation}
            q_{n_0,p}\triangleq 1-(\frac{C}{n_0})^{(p+1)\log(1+\frac{1}{\Delta})\log(\Delta)}
        \end{equation} 
         where $q_{n_0,p}$ is the constant lower bound parameterizing Equation \ref{eq: apparent-lowerbound}. 
        
      \textbf{Claim 1. } Using $q_{n_0,p}$, there is an asymptotically constant Expected Number of Nonapparent Simplices $\sigma$ for column $[\sigma] \in \textsf{ker}([\partial_p])$: 

        If $n_0 \geq C^{\frac{\log(1+\frac{1}{\Delta})\log(\Delta)}{\log(1+\frac{1}{\Delta})\log(\Delta) -1 }}$, then we obtain from Theorem \ref{thm: apparent-probability-numsample} and the definition of $1-q_{n_0,p}$ that:
        \begin{equation}
            (1-q_{n_0,p})n_0^{p+1}= n_0^{p+1} (\frac{C}{n_0})^{(p+1)\log(1+\frac{1}{\Delta})\log(\Delta)} = (\frac{C^{(p+1)\log(1+\frac{1}{\Delta})\log(\Delta)}}{n_0^{(p+1)(\log(1+\frac{1}{\Delta})\log(\Delta)-1)}}) 
        \end{equation} 
        \begin{equation}
            \leq (\frac{C^{(p+1)\log(1+\frac{1}{\Delta})\log(\Delta)}}{C^{(p+1)(\log(1+\frac{1}{\Delta})\log(\Delta))}})=1
        \end{equation}
        where the last inequality follows by $n_0 \geq C^{\frac{\log(1+\frac{1}{\Delta})\log(\Delta)}{\log(1+\frac{1}{\Delta})\log(\Delta) -1 }}$.
        
        Since $n_0^{p+1}$ is an upper bound on the number of simplices $\sigma \in \textsf{ker}(\partial_p)\cap {X \choose (p+1)}$ that can be apparent, by union bound, we have the following bound on the expected number of simplices $\sigma$ that cannot be apparent:
        \begin{equation}
            \mathbb{E}[\mathbf{1}[{\sigma \in \textsf{ker}(\partial_p)\cap {X \choose (p+1)}: \not \exists \tau, (\sigma,\tau) \text{ is apparent}}]] \leq (1-q_{n_0,p})n_0^{p+1} \leq 1
        \end{equation}
       \textbf{2. }A Sufficient Condition to achieving $P(\exists \tau: (\sigma,\tau) \text{ is apparent}) \geq q_{n_0,p}$: 
       
        Let $n^*(q_{n_0,p})= \frac{\log((1-q_{n_0,p})^{-1})}{\log((1-\frac{1}{\Delta})^{-1})}$ from Theorem \ref{thm: apparent-probability-numsample}. We know that:
        \begin{subequations}
         \begin{equation}
            n^*(q_{n_0,p}) = \frac{\log((\frac{n_0}{C})^{(p+1)\log(1+\frac{1}{\Delta})\log(\Delta)})}{\log((1-\frac{1}{\Delta})^{-1})} \leq \frac{{(p+1)\log(1+\frac{1}{\Delta})\log(\Delta)}\log(\frac{n_0}{C})}{\log(1+\frac{1}{\Delta})} 
        \end{equation}
        \begin{equation}
            \leq (p+1)\log(\Delta)\log(\frac{n_0}{C})\leq (p+1)\log(\Delta)\sqrt{n_0}
        \end{equation}
        \end{subequations}
        This allows us to design a sufficient condition for Theorem \ref{thm: apparent-probability-numsample}, namely that:
        \begin{subequations}
        \begin{equation}
            {n_0} \geq (p+1)\log(\Delta)\sqrt{n_0} \Rightarrow 
        \end{equation}
        \begin{equation}
            n_0 \geq n^*(q_{n_0,p})= \frac{\log((1-q_{n_0,p})^{-1})}{\log((1-\frac{1}{\Delta})^{-1})} \Rightarrow P(\exists \tau: (\sigma,\tau) \text{ is apparent}) \geq q_{n_0,p}
        \end{equation}
        \end{subequations}
       Thus, taking both sufficient conditions: 
       
       ${n_0} \geq (p+1)\log(\Delta)\sqrt{n_0}$ and $n_0 \geq C^{\frac{\log(1+\frac{1}{\Delta})\log(\Delta)}{\log(1+\frac{1}{\Delta})\log(\Delta) -1 }}$, we have that:
       \begin{subequations}
       \begin{equation}
           {n_0}\geq \Omega(\max(C^{\frac{\log(1+\frac{1}{\Delta})\log(\Delta)}{\log(1+\frac{1}{\Delta})\log(\Delta) -1 }}, (p+1)^2\log^2(\Delta))) 
       \end{equation}
       \begin{equation}
           \Rightarrow \mathbb{E}[\mathbf{1}[{\sigma: [\sigma] \in \textsf{ker}([\partial_1]), \not\exists \tau, (\sigma,\tau) \text{ is apparent}}]] \leq 1
       \end{equation}
       \end{subequations}
       Since $\mathbb{E}[\mathbf{1}[{\sigma: [\sigma] \in \textsf{ker}([\partial_p]), \not\exists \tau, (\sigma,\tau) \text{ is apparent}}]]\geq \mathbb{E}[\beta_p]$, we must have $0\leq \mathbb{E}[\beta_p]\leq 1$.
       
        In \cite{zhang2025computing} a complexity bound is shown for computing the dual matrix reduction problem for persistent homology. Since Ripser++ is also based on dual matrix reduction, we obtain the work complexity bound for Ripser++ is:
        \begin{equation}
           O( \beta_{p}n_{p}\log(n_0)+n_{p+1}) 
        \end{equation}
        where the $O(n_{p+1})$ is the total amount of simplices from the coboundary that must be enumerated. 

        Since Ripser++ enumerates these coboundaries in parallel, this part of the work complexity becomes $O(n_0 (p+1))$, which is the work it takes to enumerate a single coboundary.

        This thus gives a depth complexity of $O( \beta_{p}n_{p}\log(n_0)+n_0 (p+1))$
    \end{proof}
    \subsection{Under Non-trivial Constant Number of ``Topological Punctures"}
    Consider a finitely ``topologically punctured" uniform sampling distribution over the hypercube. We define this realistic setting as a  sampling condition.
    
 e.g.    A probability integral transform of 3D points sensed by lidar~\cite{wandinger2005introduction} from a room where the separated furniture cannot be penetrated.
    \begin{condition}[$k$-Cavity Sampling Condition]\label{cond: uniform-k-cavity}
        \begin{equation}
          P_{k\text{-cavity}}(x)\triangleq P(x \sim U([0,1]^{d_{amb}}\setminus \bigcup_{C \in \mathcal{C}_k}C)) 
    \end{equation}
    where $\mathcal{C}_k$ is a collection of $k$ disjoint cavities of $[0,1]^{d_{amb}}$ with
    \begin{equation}
        \sum_{C \in C_k}\textsf{vol}(C)=V\text{ where } V<1
    \end{equation}
    $\bullet$ Assume the Euclidean distance.
    \end{condition}
    We show that with 
    the new sampling condition, Condition \ref{cond: uniform-k-cavity}, our algorithm can recover the number of cavities and the expected complexity is the same as Theorem \ref{thm: expectedcomplexity-ripserpp} up to a multiple of the parameter $k$.
    \begin{corollary}(Complexity Under $k$-Cavity Sampling Condition)
    \label{cor: k-cavity-corollary}
    
        Assuming the $k$-Cavity Uniform Sampling Condition (Condition \ref{cond: uniform-k-cavity}), the distance $R$ from Lemma \ref{lemma: appendix-k-cavities-betti}, and the sampling lower bound conditions of Theorem \ref{thm: expectedcomplexity-ripserpp}.
        Let: \begin{equation}
            v_{cell}:= \frac{1}{1-V}(\frac{R}{2\sqrt{d_{amb}}})^{d_{amb}}
        \end{equation}
        be the probability of a random sample from $P_{k\text{-cavity}}$ hit a cell of length $\frac{R}{2\sqrt{d_{amb}}}$.
        
        Assume that for any pair of points with: 
        \begin{equation}
           d(p,q)\leq R\Rightarrow  \mathbf{n}(p)\cdot \mathbf{n}(q)>0, \forall p,q \in \partial(C)
        \end{equation}
        For any $\delta: 0<\delta<1$ and 
        \begin{equation}
            {n_0}\geq \Omega(\max(C^{\frac{\log(1+\frac{1}{\Delta})\log(\Delta)}{(\log(1+\frac{1}{\Delta})\log(\Delta) -1) }}, (\textsf{dim}(\textsf{Rips}(X)))^2\log^2(\Delta), \frac{\log(\frac{\delta}{v_{cell}})}{\log(1-v_{cell})}))
        \end{equation}
        
        Then the expected Betti numbers of the Rips complex satisfy:
        \begin{equation}
           k\leq  \mathbb{E}[\beta_{p}]\leq k+1, \forall p: 1\leq p\leq d_{amb}-1   
        \end{equation}
        with probability atleast $1-\delta$.

        This makes the expected work complexity of Ripser++: \begin{equation}
        O( kn_{p}\log(n_0)+n_{p+1})
    \end{equation}
    with expected depth complexity of \begin{equation}
        O( kn_{p}\log(n_0)+n_0 (p+1))
    \end{equation}
    which is the complexity of computing $\textsf{ker}([\partial_p])$.

    To compute the Betti number only, the expected complexity of Ripser++ is \begin{equation}
        O(n_{p+1})
    \end{equation}
    and the expected depth complexity is \begin{equation}
        O(n_0(p+1))
    \end{equation}
    These happen with probability atleast $1-\delta$.
    \end{corollary}
    \begin{proof}
        Since the $k$-Cavity uniform distribution is still uniform. It is just uniform over a restricted support. We must have:
        \begin{equation}
        P_{k\text{-cavity}}(E)=\frac{U_{[0,1]^{d_{amb}}}(E)}{1-V}, \forall E \in \sigma(P_{k\text{-cavity}})
        \end{equation}
        where $\sigma(P_{k\text{-cavity}})$ is the sigma-algebra of $P_{k\text{-cavity}}$ and $U_{[0,1]^{d_{amb}}}$ is the uniform distribution over the hypercube $[0,1]^{d_{amb}}$. 

        In Proposition \ref{prop: ball-intersection-polytope} the inscribed polytope from the intersection of $p+1$ balls only needs to span $p+1$ points. This is because we cannot guarantee a reflection on the polytope. The lower bound from Lemma \ref{lemma: apparent-probability} remains.

        Since the probability of not sampling from a $d_{amb}$-dimensional cell of length $\frac{R}{2\sqrt{d_{amb}}}$ from $[0,1]^{d_{amb}}$ in $n_0$ trials is :
        \begin{equation}
            1-v_{cell}
        \end{equation}
        and since there are (counting fractional cells)
        \begin{equation}
            m_{ncells}:=\frac{1}{v_{cell}}
        \end{equation}
        many $d_{amb}$-dimensional cells over the volume $1-V$ in $[0,1]^{d_{amb}}$.
        With a probability of:
        \begin{equation}
              P_{n_0}=1-m_{ncells}(1-v_{cell})^{n_0}
        \end{equation}
        we can fill each $d_{amb}$-dimensional cell of length $\frac{R}{2\sqrt{d_{amb}}}$ to obtain the sampling conditions of Lemma \ref{lemma: appendix-k-cavities-betti} of the Appendix.

        Thus, letting $1-\delta\leq  P_{n_0}$, we need:
        \begin{equation}
        n_0\geq \frac{\log(\frac{\delta}{m_{ncell}})}{\log(1-v_{cell})}
\end{equation}
to obtain the sufficient condition with probability atleast $1-\delta$ for Lemma \ref{lemma: appendix-k-cavities-betti}.


        By Lemma \ref{lemma: appendix-k-cavities-betti}, we then have: 
        \begin{equation}
            k\leq \mathbb{E}[\beta_p] , \forall p=1,...,d_{amb}-1
        \end{equation}
       The condition:
       \begin{equation}
            \mathbf{n}(p)\cdot \mathbf{n}(q)>0
        \end{equation}
       means that the angle between the normals of any pair of points on $\partial(C)$ must be atmost $90^{\circ}$. This means their tangents are atleast $(180-90)^{\circ}$ apart. This is larger than $60^{\circ}$. 
       
       For any set of $p$ points $\sigma$ outside $C$, any apex point that completes $\sigma$ to a $p$-dimensional simplex $\tau$ must form angles more than $60^{\circ}$ with any pair from $\sigma$. Thus, $\tau$ must form a strict constant fraction of any inscribed polytope of $p+1$ radius-$R$ balls from Proposition \ref{prop: ball-intersection-polytope}. This means that in Lemma \ref{lemma: apparent-probability} and Theorem \ref{thm: apparent-probability-numsample}, we can lower bound $\Delta$ from Equation \ref{eq: Delta} by the number:
       \begin{equation}
           C'(p)\frac{\Delta}{2^{p+1}}: C'(p)<1
       \end{equation}
       where $C'(p)$ is the ratio between a $p+1$-dimensional equilateral simplex and a simplex attaching an apex to the largest volume $p$-dimensional simplex $\sigma_{vol}\in \textsf{Rips}(X)$ at $90^{\circ}$ to all pairs of points in $\sigma_{vol}$. 

    We have that all pairs of points outside $\mathcal{C}_k$ satisfy Lemma \ref{lemma: apparent-probability}. Thus, we satisfy the conditions of Theorem \ref{thm: expectedcomplexity-ripserpp} on the uniform distribution $P_{k\text{-cavity}}$, we thus have that:
        \begin{equation}
        \mathbb{E}[\beta_p-k] \leq 1, \forall p=1,...,d_{amb}-1
        \end{equation}
        with probability $P_{n_0}$.
        Thus:
        \begin{equation}
            k\leq \mathbb{E}[\beta_p] \leq k+1, \forall p=1,...,d_{amb}-1
        \end{equation}
        The expected complexity becomes as in Theorem \ref{thm: expectedcomplexity-ripserpp}, but with a nontrivial expected Betti number which is $O(k)$.
    \end{proof}
    \subsection{From Geometry to Combinatorics through Locality}
    A space $X$ where every sample $x \in X$ can intervene with an independent \textbf{local metric}, meaning that for every $x \in X$ there is a function $d_x: X \rightarrow \mathbb{R}^+\cup \{\infty\}$, is equivalent to a weighted directed graph. Certainly each $x\in X$ has a neighborhood defined by $\{y \in X: d_x(y)\in \mathbb{R}^+\}$, thus we can form the graph \begin{equation}
        G_{X,d_{\bullet}}\triangleq(X,\{(x,y): d_x(y)<\infty\})
    \end{equation}
    
    If there is symmetry across local metrics on pairs of samples, meaning: $d_x(y)=d_y(x), \forall x,y \in X$ then $(X,d_x)$ forms a weighted undirected complete graph also defined similarly as \begin{equation}
        G_{X,d_{\bullet},sym}=(X,\{(x,y): d_x(y)<\infty \text{ and } d_x(y)=d_y(x) , \forall x,y \in X\})
    \end{equation}

    As in the metric case, we can define an \textbf{aspect ratio} on a weighted (un)directed graph $G_{X,d_{\bullet}}$, denoted $a(G_{X,d_{\bullet}})_{d_{\bullet}}$, by:
    \begin{equation}
        a(G_{X,d_{\bullet}})_{d_{\bullet}}\triangleq \frac{\max_{x,y \in X} d_x(y)}{\min_{x,y \in X} d_x(y)}
    \end{equation}

    In analogy to the case of the uniform distribution on the hypercube, for  $G_{X,d_{\bullet},sym}$ we can also derive similar probabilistic bounds. 
    \begin{lemma}\label{lemma: combinatorical-apparent-probability}
    For a $p$-dimensional simplex $\sigma \in \textsf{ker}(\partial_p)\cap {X \choose (p+1)}$ with $\textsf{diam}(\sigma)=t$, for $t \in [0,1]$.
    Let $G_{X,d_{\bullet}}$ denote a complete weighted undirected graph. 
    
        If the local metrics $d_x$ are chosen for each $x \in X$ to  satisfy $d_x(y) \in U([0,1]), \forall y \in X$ with the only dependency of $d_x(y)=d_y(x)$, then 
        \begin{equation}
           P(\exists \tau: (\sigma,\tau) \text{ is apparent} \mid \textsf{diam}(\sigma)=t)\geq  {t}^{p+1}
        \end{equation}
    \end{lemma}
    \begin{proof}
        According to Lemma \ref{lemma: apparentpair-diam-peq1}, $\sigma \in z_p$, for some $ z_p \in  \textsf{ker}(\partial_p)$ has some $\tau: \partial_{p+1}(\tau) \in \textsf{im}(\partial_{p+1})$ with $(\sigma,\tau)$ apparent if there is a $x \in X$ with 
        \begin{equation}\label{eq: locally-apparent}
            d_x(x_i)\leq t, \forall x_i \in \sigma
        \end{equation} 
        Such a $x \in X$ can be found if the local metric $d_x$ satisfies the distance bounds of Equation \ref{eq: locally-apparent}. By assumption of the chosen local metrics $d_x$, this will occur with probability $t^{p+1}$. This gives the lower bound.
    \end{proof}
    \begin{theorem}\label{thm: localmetric-expectedcomplexity}
    
        Assuming the conditions of Lemma \ref{lemma: combinatorical-apparent-probability}, let $a(\textsf{Rips}(X))_{d_{\bullet}}=\rho$:
        \begin{equation}
            (1-q_{n_0,p}) \triangleq (\frac{C}{n_0})^{p}
        \end{equation}
        
        If $\rho$ satisfies: 
        \begin{equation}\label{eq: localmetric-lowerbound-q}
            \frac{\rho^{-(p+1)}\rho^{-({p+1 \choose 2} +1)}}{p+1+ {p+1 \choose 2}}\geq q_{n_0,p}, 
        \end{equation}
    Then the expected Betti numbers of the Rips complex satisfy $\mathbb{E}[\beta_p]=O(1), \forall p\geq 1$. 
    
    As in Theorem \ref{thm: apparent-probability-numsample}, we get similar expected complexity bounds. 
    
    The expected work complexity of Ripser++ is \begin{equation}
        O(n_{p}\log(n_0)+n_{p+1})
    \end{equation} with expected depth complexity of 
    \begin{equation}
        O(n_{p}\log(n_0)+n_0 (p+1))
    \end{equation}
    which is the complexity of computing $\textsf{ker}([\partial_{p}])$.
    \end{theorem}
    \begin{proof}
        For each $\sigma \in \textsf{ker}(\partial_p) \cap {X \choose (p+1)}$, we can compute that
        \begin{subequations}
        \begin{equation}
            P(\exists \tau: (\sigma,\tau) \text{ is apparent}) 
        \end{equation}
        \begin{equation}
            \geq \int_{t \in [0,1]} P(\exists \tau: (\sigma,\tau) \text{ is apparent} \mid  \textsf{diam}(\sigma)=t)P(\textsf{diam}(\sigma)=t) dt
        \end{equation}
        \begin{equation}
            \geq \int_{t \in [0,1]}t^{p+1}t^{p+1 \choose 2} dt
        \end{equation}
        \begin{equation}
            \geq \frac{t^{p+1}t^{{p+1 \choose 2}+1}}{p+1+ {p+1 \choose 2}}
        \end{equation}
        \begin{equation}
            \geq \frac{\rho^{-(p+1)}\rho^{-({p+1 \choose 2}+1)}}{p+1+ {p+1 \choose 2}}
        \end{equation}
        \begin{equation}
            \geq q_{n_0,p}
        \end{equation}
        \end{subequations}
        Taking the complementary event of $E_{\sigma}\triangleq (\exists \tau: (\sigma,\tau) \text{ is apparent})$, and using the definition of $1-q_{n_0,p}$, we have by union bound that:
        \begin{subequations}
        \begin{equation}
            \mathbb{E}[\beta_1]\leq \mathbb{E}[\mathbf{1}[\sigma \in \textsf{ker}(\partial_p): \not \exists \tau: (\sigma, \tau) \text{ is apparent}]]\leq (1-q_{n_0,p})n_0^p 
        \end{equation}
        \begin{equation}
            =(1-\frac{\rho^{-(p+1)}\rho^{-({p+1 \choose 2}+1)}}{p+1+ {p+1 \choose 2}})n_0^p\leq \frac{C n_0^p}{n_0^p}=
            O(1)
        \end{equation}
        \end{subequations}

        The remainder of the proof follows verbatim as in Theorem \ref{thm: apparent-probability-numsample}.
    \end{proof}
    In Theorem \ref{thm: localmetric-expectedcomplexity}, the bound on the aspect ratio of Equation \ref{eq: localmetric-lowerbound-q}, requires a bound on $\rho$. In fact, $\rho$ must be bounded by a constant. 
    \begin{corollary}
        In order for Theorem \ref{thm: localmetric-expectedcomplexity} to hold, the aspect ratio must satisfy:
        \begin{equation}
            \rho=O(1)
        \end{equation}
    \end{corollary}
    \begin{proof}
    Writing out the equality for $1-q_{n_0,p}$. 
    \begin{equation}
        1-q_{n_0,p}=1-\frac{\rho^{-(p+1)}\rho^{-({p+1 \choose 2}+1)}}{p+1+ {p+1 \choose 2}}\leq \frac{C}{n_0^p}
        \end{equation}
        Taking a lower bound on $1-q_{n_0,p}$:
        \begin{equation}
            \exists C'>0: 1-\frac{1}{C'\rho^{p^2} p^2}\leq 1-q_{n_0,p}\leq \frac{C}{n_0^p}  
        \end{equation}
        Since $p$ is a constant, we must have that for $n_0\rightarrow \infty$ that $\rho \not\rightarrow \infty$. If $\rho\rightarrow \infty$, then we would have $1\leq \frac{C}{n_0^p}$ for $n_0 \rightarrow \infty$, Contradiction.
        
        Thus, $\rho$ cannot have any dependency with $n_0$. It is thus bounded by a constant.
    \end{proof}
    \begin{remark}
    In Theorem \ref{thm: expectedcomplexity-ripserpp}, the aspect ratio can depend on $n_0$ and still achieve the expected complexities. This is because in Lemma \ref{lemma: apparent-probability}, we have the following lower bound:
    \begin{equation}
         P(\exists \tau: (\sigma,\tau) \text{ is apparent} \mid \textsf{diam}(\sigma)=t)\geq(1-(1-\frac{1}{(d_{amb}+1)!2^{p+1}\rho^{d_{amb}}})^{n_0})
    \end{equation}
    For large $n_0$ this goes to $1$ even when $\rho$ has a dependency on $n_0$. 

    However, in Lemma \ref{lemma: combinatorical-apparent-probability}, when we assume that each edge obtains a distance i.i.d. uniform from $[0,1]$, we have that:
    \begin{equation}
           P(\exists \tau: (\sigma,\tau) \text{ is apparent} \mid \textsf{diam}(\sigma)=t)\geq  {t}^{p+1}
    \end{equation}
    This does not go to one with $n_0$, which is the reason for the stronger requirement of $\rho=O(1)$. Without the compactness from the uniform distribution on the points, increasing the number of points does not bring the points closer to each other. This prevents an asymptotic convergence with respect to $n_0$ for the probability of an apparent pair. This puts the constant bound on $\rho$.
    \end{remark}
	\section{GPU and System Kernel Development for Ripser++}
	We have laid the mathematical and algorithmic foundation for Ripser++ in Section \ref{sec: contributions}. GPU has two distinguished features to deliver high performance compared with CPU. First, GPU has very high memory bandwidth for massively parallel data accesses between computing cores and the on-chip memory (or device memory). Second, Warp is a basic scheduling unit consisting of multiple threads. GPU is able to hide memory access latency by warp switching, benefiting from zero-overhead scheduling by hardware. 
	
	To turn the elegant mathematics proofs and parallel algorithms into GPU-accelerated computation for high performance, we must address several technical challenges. First, the capacity of GPU device memory is much smaller than the main memory in CPU. Therefore, GPU memory performance is critical for the success of Ripser++. Second, only a portion of the computation is suitable for GPU acceleration. Ripser++ must be a hybrid system, effectively switching between GPU and CPU, which increases system development complexity. Finally, we aim to provide high performance computing service in the TDA community without a requirement for users to do any GPU programming. Thus, the interface of Riper++ is GPU independent and easy to use. In this section, we will explain how we address these issues for Ripser++.
    \begin{figure}[!h]\includegraphics[width=1.0\columnwidth]{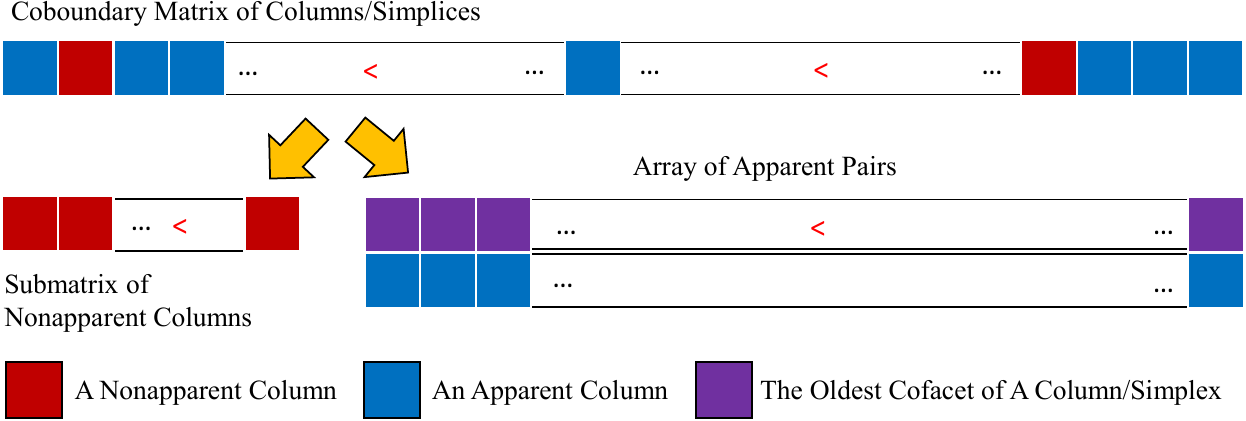}
		\caption{\small{After finding apparent pairs, we partition the coboundary matrix columns into apparent and nonapparent columns. The apparent columns are sorted by the coboundary matrix row (the oldest cofacet of an apparent column) and stored in an array of pairs; while the nonapparent columns are collected and sorted by coboundary matrix order in another array for submatrix reduction.}}
		\label{fig: post-apparent}
		\centering
	\end{figure}
    
    \subsection{Core System Optimizations}
    \label{sec: performance-optimization}
	The expected performance gain of finding apparent pairs on GPU comes from not only the parallel computation on thousands of cores but also the concurrent memory accesses at a high bandwidth, where the apparent pairs can be efficiently aggregated. In a sequential context, an apparent pair (a row index and a column index) of the coboundary matrix may be kept in a hashmap as a key-value pair with the complexity of $O(1)$. However building a hashmap is not as fast as constructing a sorted continuous array \cite{Kim:2009:SVH:1687553.1687564} in parallel. So in our implementation, the apparent pairs are represented by a key-value pair $(\tau,\sigma)$ where $\tau$ is the oldest cofacet of simplex $\sigma$  
	and stored in an aligned continuous array of pairs. This slightly lowers the read performance because we need a binary search to locate a desired apparent pair. But this is a cost-effective implementation since the number of insertions of apparent pairs are actually three orders of magnitude higher than that of reads (See Table \ref{tab: hashtable} in Section \ref{sec: experiments}) after finding apparent pairs. Figure \ref{fig: post-apparent} presents how we collect apparent pairs on GPU, where each thread works on a column of coboundary matrix and writes to the output array in parallel.
	
	On top of the sorted array, we add a hashmap as one more layer to exclusively store persistence pairs discovered during the submatrix reduction. 
	Apparent pairs, and in fact persistence pairs in general can be stored as key-value pairs since no two persistence pairs $(\sigma,\tau)$ and $(s',t')$ have the possibility of $s=\sigma'$ or $t=t'$, as any equality would contradict Algorithm \ref{alg:standard-algorithm}.
	Figure \ref{fig: pairs-datastructure} explains our two layer design of a key-value storage data structure for persistence pairs in detail.

    \begin{figure}[!h]
		\includegraphics[width=1.0\columnwidth]{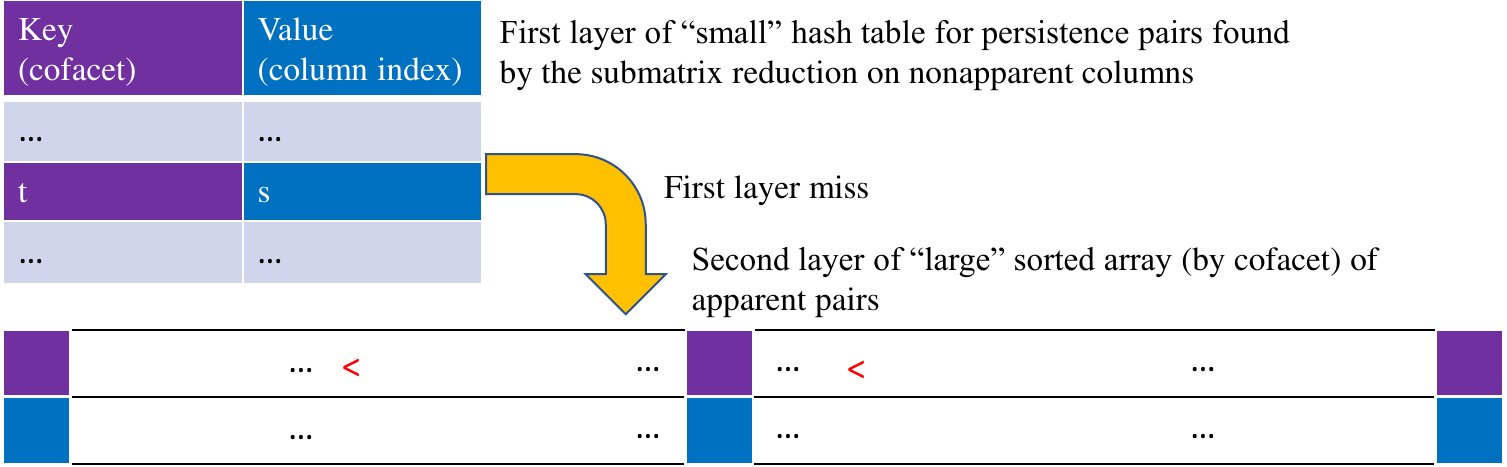}
		\caption{\small{Two-layer data structure for persistence pairs. Apparent pair insertion to the second layer of the data structure is illustrated in Figure \ref{fig: post-apparent}, followed by persistence pair insertion to a small hashmap during the submatrix reduction on CPU. A key-value read during submatrix reduction involves atmost two steps: first, check the hashmap; second, if the key is not found in the hashmap, use a binary search over the sorted array to locate the key-value pair (see the arrow in the figure).}}
		\label{fig: pairs-datastructure}
		\centering
	\end{figure} 
\subsection{Filtration Construction with Clearing}
	Before entering the matrix reduction phase, the input simplex-wise filtration must be constructed and simplified to form coboundary matrix columns. We call this Filtration Construction with Clearing. This requires two steps: filtering and sorting. Both of which can be done in parallel, in fact massively in parallel. Filtering removes simplices that we don't need to reduce as they are equivalent to zeroed columns. As presented in Algorithm \ref{alg:filtering}, these simplices are filtered out: the ones having higher diameters than the $\text{threshold}$ (see Section \ref{sec: appendix-enclosing-radius} for the enclosing radius condition that can be applied even when no $\text{threshold}$ is explicitly specified) and paired simplices (the clearing lemma \cite{chen2011persistent}).
	\begin{figure}[!h]
		\includegraphics[width=0.9\columnwidth]{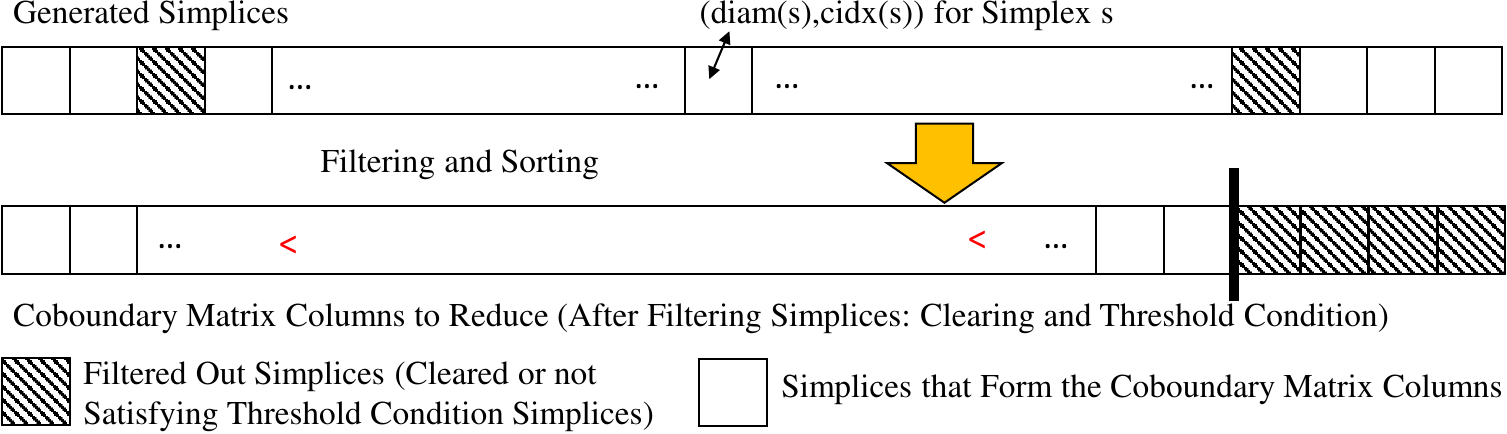}
		\caption{\small{The Filtration Construction with Clearing Algorithm for Full Rips Filtrations}}
		\label{fig:filtrationconstruction}
		\centering
	\end{figure}
    
	Sorting in the reverse of the order given in Section \ref{sec: filtration-order} is then conducted over the remaining simplices. This is the order for the columns of a coboundary matrix. 
	The resulting sequence of simplices is then the columns to reduce for the following matrix reduction phase. Algorithm \ref{alg:fullripsalgorithm-sorting} presents how we construct the full Rips filtration with clearing. Our GPU-based algorithms leverage the massive parallelism of GPU threads and high bandwidth data processing in GPU device memory. For a sparse Rips filtration, our construction process uses the cofacet enumeration of Ripser per thread, as in Algorithm \ref{alg:cofacets_sparse}, and is similar to \cite{zomorodian2010fast}.
	\begin{algorithm}[!h]
    \SetKwComment{Comment}{//}{}
    \caption{Filtering the Columns on GPU}
    \label{alg:filtering}
    \KwIn{$\pmb{P}$: the persistence pairs in the form (cofacet, simplex) discovered in the previous dimension; $\text{threshold}$: the max diameter allowed for a simplex; $\textsf{diam}(\cdot)$: the diameter of a simplex; $\textsf{cidx}(\cdot)$: the combinatorial index of a simplex.}
    \KwOut{$\pmb{C}$: an array of simplices, where each element includes a diameter paired with a combinatorial index; $\text{tid}$: an array of flags marking which columns are kept (filtered in).}
    \SetKwProg{proc}{Procedure}{:}{}
    \SetKwFunction{filtercols}{filter-columns-kernel}
    \proc{\filtercols{$\pmb{C}$, $\pmb{P}$, $\text{threshold}$, $\text{tid}$ }}{
        $\textsf{cidx}(\sigma) \gets \text{tid}$ \;
        \uIf{$\nexists t \text{ such that } (\tau, \sigma) \in \pmb{P} \text{ AND } \textsf{diam}(\sigma) \leq \text{threshold}$}{
            $\textsf{diam}(\pmb{C}[\text{tid}]) \gets \textsf{diam}(\sigma)$ \;
            $\textsf{cidx}(\pmb{C}[\text{tid}]) \gets \textsf{cidx}(\sigma)$ \;
            $\text{tid}[\text{tid}] \gets 1$ \;
        }
        \Else{
            $\textsf{diam}(\pmb{C}[\text{tid}]) \gets -\infty$ \;
            $\textsf{cidx}(\pmb{C}[\text{tid}]) \gets +\infty$ \;
            $\text{tid}[\text{tid}] \gets 0$ \;
        }
    }
    
\end{algorithm}
	\begin{algorithm}[h]
    \SetKwComment{Comment}{//}{}
    \caption{Use GPU for Full Rips Filtration Construction with Clearing}
    \label{alg:fullripsalgorithm-sorting}
    \KwIn{$\pmb{P}$, $\text{threshold}$, $\text{tid}$: as defined in Algorithm \ref{alg:filtering}; $n_0$: the number of points; $p$: the current dimension for simplices to construct.}
    \KwOut{$\pmb{C}$: as defined in Algorithm \ref{alg:filtering}.}
    
    $\pmb{C} \gets \emptyset$ \;
    $\text{tid} \gets \{0, \ldots, 0\}$ \;
    \texttt{filter-columns-kernel}($\pmb{C}$, $\pmb{P}$, $\text{threshold}$, $\text{tid}$) \Comment{${n_0 \choose p+1}$ threads launched}
    $\text{len} \gets \texttt{GPU-reduction(flagarray)}$ \;
    \texttt{GPU-sort}($\pmb{C}$)  \Comment{Sort entries of $\pmb{C}$ in coboundary filtration order: decreasing diameters, increasing combinatorial indices; restrict $\pmb{C}$ to indices 0 to $\text{len}-1$ afterwards.}
    
\end{algorithm}
	\subsection{Warp-based Filtering}
	\label{sec: warp-filtering}
	There is also a standard technique for filtering on GPU which is warp-based. A warp is a unit of 32 threads that work in SIMT (Single Instruction Multiple Threads) fashion, executing the same instruction on multiple data. This concept is very different from MIMD (Multiple Instruction Multiple Data) parallelism \cite{zhang2019hypha}. Warp-filtering does not change the complexity of filtering in $O(N)$, where $N$ is the number of elements to filter; however it can allow for insertion into an array using 32 threads (a unit of a warp) at a time in SIMT fashion. We use warp-based filtering for sparse computation and as an equivalent algorithm to Algorithm \ref{alg:fullripsalgorithm-sorting}. Warp-based filtering involves each warp atomically grabbing an offset to the output array and communicating within the warp to determine which thread will write what selected array element to the array beginning at the thread's offset within the warp.
	
	\subsection{Using Ripser++}
	The command line interface for Ripser++ is the same as Ripser to make the usage of Ripser++ as easy as possible to TDA specialists. However, Ripser++ has a {\tt -{}-sparse} option which manually turns on, unlike in Ripser, the sparse computation algorithm for Vietoris-Rips barcode computation involving a sparse number of neighboring relations between points. Python bindings for Ripser++ will be available to allow users to write their own preprocessing code on distance matrices in Python as well as to aid in automating the calling of Ripser++ by removing input file I/O.    
	

 \section{Experiments}
\label{sec: experiments}
    All experiments are performed on a powerful computing server. It consists of an NVIDIA Tesla V100 GPU that has 5120 FP32 cores and 2560 FP64 cores for single- and double-precision floating-point computation. The GPU device memory is 32 GB High Bandwidth Memory 2 (HBM2) that can provide up to 900 GB/s memory access bandwidth. The server also has two 14 core Intel XEON E5-2680 v4 CPUs (28 cores in total) running at 2.4 GHz with a total of 100 GB of DRAM. The datasets are taken from the original Ripser repository on Github \cite{ripser} and the repository of benchmark datasets from \cite{otter2017roadmap}.

\subsection{The Empirical Relationship amongst Apparent Pairs, Emergent Pairs, and Shortcut Pairs}
	There exists three kinds of persistence pairs of the Vietoris-Rips filtration, in fact for any filtration with a simplex-wise refinement. Using the terminology of \cite{bauer2019ripser}, these are apparent (Definition \ref{def:apparent_pair}) \cite{delgado2014skeletonization, henselman2016matroid, bauer2019ripser, mendoza2017parallel}, shortcut \cite{bauer2019ripser}, and emergent pairs \cite{bauer2019ripser, zhang2019hypha}. By definition, they are known to form a tower of sets ordered by inclusion (expressed by Equation (\ref{eq:pairs})). We will show a further empirical relationship amongst these pairs involving their cardinalities. 
	\begin{equation}
	\label{eq:pairs}
    \overbrace{
	\underbrace{\text{apparent pairs}}_\text{large cardinality} \subseteq \text{shortcut pairs} \subseteq \text{emergent pairs} \subseteq \text{persistence pairs}}^\text{the difference in cardinalities is ``small"}
    \end{equation}
	The first empirical property is that the cardinality difference amongst all of the sets of pairs is very small compared to the number of pairs
	, assuming Ripser's framework of computing cohomology and using the simplex-wise filtration ordering in Section \ref{sec: filtration-order}. Thus there exist a very large number of apparent pairs. The second is that the proportion of apparent pairs to columns in the cleared coboundary matrix increases with dimension (see Section \ref{sec: appendix-apparent-percentage} in Appendix), assuming no diameter threshold criterion as in the first property.
	
	Table \ref{tab:all-pairs} shows the percentage of apparent pairs up to dimension $p$ is extremely high, around 99\%. Since the number of columns of a cleared coboundary matrix equals to the number of persistence pairs, the number of nonapparent columns for submatrix reduction is a tiny fraction of the original number of columns in Ripser's matrix reduction phase.

\begin{table*}[h]
	\centering
	\caption{\small{Empirical Results on Apparent, Shortcut, Emergent Pairs}}
	\label{tab:all-pairs}
	\begin{tabular}{@{}lrrrrrrr@{}} \toprule
		
		&   &   & \scriptsize apparent         &\scriptsize  shortcut  &\scriptsize emergent &\scriptsize all & \scriptsize percentage of \\
		\scriptsize Datasets & \scriptsize $n$ &\scriptsize $p$ & {\scriptsize pairs} & \scriptsize pairs &\scriptsize pairs & \scriptsize pairs & \scriptsize apparent pairs\\
		
		\midrule
		\scriptsize {\it celegans}           & \scriptsize 297 & \scriptsize 3 & \scriptsize 317,664,839 & \scriptsize 317,723,916 & \scriptsize 317,723,974 & \scriptsize 317,735,650 & \scriptsize 99.9777139\% \\
		\scriptsize {\it dragon1000}         & \scriptsize 1000 & \scriptsize 2 & \scriptsize  166,132,946  & \scriptsize 166,160,587   & \scriptsize 166,160,665 & \scriptsize 166,167,000 & \scriptsize 99.9795062\%\\
		\scriptsize {\it HIV}                & \scriptsize  1088 & \scriptsize 2 & \scriptsize 214,000,996 & \scriptsize 214,030,431 & \scriptsize 214,040,521 & \scriptsize 214,060,736 & \scriptsize 99.9720920\%\\
		\scriptsize {\it o3} (sparse: $t=1.4$) & \scriptsize 4096 & \scriptsize 3 & \scriptsize 43,480,968 & \scriptsize 43,940,030 & \scriptsize 43,940,686 & \scriptsize 44,081,360 & \scriptsize 98.6379912\%\\
		\scriptsize {\it sphere\_3\_192}     & \scriptsize 192 & \scriptsize 3 & \scriptsize 54,779,316 & \scriptsize 54,871,199 & \scriptsize 54,871,214 & \scriptsize 54,888,625 & \scriptsize 99.8008531\%\\
		\scriptsize {\it Vicsek300\_of\_300} & \scriptsize 300 & \scriptsize 3 & \scriptsize  330,724,672 & \scriptsize 330,818,491 & \scriptsize 330,818,507 & \scriptsize 330,835,726 & \scriptsize 99.9664323\%\\
		\bottomrule
	\end{tabular}
\end{table*}

\subsection{Execution Time and Memory Usage}

    We perform extensive experiments that demonstrate the execution time and memory usage of Ripser++. We further look into the performance of both the apparent pairs search algorithm and the management of persistence pairs in the two layer data structure after finding apparent pairs. Variables $n$ and $p$ for each dataset are the same for all experiments.
	
	Table \ref{tab: performance} shows the comparisons of execution time and memory usage for computation up to dimension $p$ between Ripser++ and Ripser with six datasets, where R. stands for Ripser and R.++ stands for Ripser++. Memory usage on CPU and total execution time were measured with the {\tt /usr/time -v} command on Linux. GPU memory usage was counted by the total displacement of free memory over program execution.
	
\begin{table*}[!h]	\centering
	\caption{\small{Total Execution Time and CPU/GPU Memory Usage}}
	\label{tab: performance}
	\begin{tabular}{@{}lrrrrrrrr@{}} \toprule
		
		&&\scriptsize    & \scriptsize R.++  & \scriptsize R.   &\scriptsize R.++ GPU  &\scriptsize R.++ CPU &\scriptsize R. CPU &  \\
		\scriptsize Datasets &n& \scriptsize d &\scriptsize time & {\scriptsize time} & \scriptsize mem. &\scriptsize mem. & \scriptsize mem. & \scriptsize Speedup\\
		
		\midrule
		\scriptsize {\it celegans}           & \scriptsize 297 & \scriptsize 3 & \scriptsize 7.30 s & \scriptsize 228.56 s & \scriptsize 16.84 GB & \scriptsize 10.53 GB & \scriptsize 23.84 GB &\scriptsize 31.33x \\
		\scriptsize {\it dragon1000}         & \scriptsize 1000 & \scriptsize 2 & \scriptsize  5.79 s  & \scriptsize 48.98 s   & \scriptsize 8.81 GB & \scriptsize 3.75 GB & \scriptsize 5.79 GB& \scriptsize 8.46x\\
		\scriptsize {\it HIV}                & \scriptsize  1088 & \scriptsize 2 & \scriptsize 7.11 s & \scriptsize 147.18 s & \scriptsize 11.36 GB & \scriptsize 6.68 GB & \scriptsize 14.59 GB & \scriptsize 20.69x\\
		\scriptsize {\it o3} (sparse: $t=1.4$) & \scriptsize 4096 & \scriptsize 3 & \scriptsize 11.62 s & \scriptsize 64.18 s & \scriptsize 18.76 GB & \scriptsize 2.77 GB & \scriptsize 3.86 GB & \scriptsize 5.52x\\
		\scriptsize {\it sphere\_3\_192}     & \scriptsize 192 & \scriptsize 3 & \scriptsize 2.43 s & \scriptsize 36.96 s & \scriptsize 2.92 GB & \scriptsize 2.03 GB & \scriptsize 4.32 GB & \scriptsize 15.21x\\
		\scriptsize {\it Vicsek300\_of\_300} & \scriptsize 300 & \scriptsize 3 & \scriptsize 9.98 s & \scriptsize 248.72 s & \scriptsize 17.53 GB & \scriptsize 11.46 GB & \scriptsize 27.78 GB & \scriptsize 24.92x\\
		\bottomrule
	\end{tabular}
\end{table*}

    Table \ref{tab: performance} shows Ripser++ can achieve 5.52x - 31.33x speedups of total execution time over Ripser in the evaluated datasets. The performance improvement mainly comes from massive parallel operations of finding apparent pairs on GPU
    , and from the fast filtration construction with clearing by GPU using filtering and sorting. We also notice that the speedups of execution time varies in different datasets. That is because the percentages of execution time in the submatrix reduction are different among datasets.   

	It is well known that the memory usage of full Vietoris-Rips filtration grows exponentially in the number of simplices with respect to the dimension of persistence computation. For example, 2000 points at dimension 4 computation may require ${2000}\choose{4+1}$ $\times 8$ bytes = 2 million GB memory. Algorithmically, we avoid allocating memory in the cofacet dimension and keep the memory requirement of Ripser++ asymptotically same as Ripser. Table \ref{tab: performance} also shows the memory usage of Ripser++ on CPU and GPU. Ripser++ can actually lower the memory usage on CPU. This is mostly because Ripser++ offloads the process of finding apparent pairs to GPU and the following matrix reduction only works on much fewer columns than that of Ripser (as the submatrix reduction). Table \ref{tab: performance} also shows that the GPU device memory usage is usually lower than the total memory usage of Ripser. However, in the sparse computation case (dataset \textit{o3}) the algorithm must change; Ripser++ thus allocates memory depending on the hardware instead of the input sizes.

	\subsection{Throughput of Apparent Pairs Discovery with Ripser++ vs. Throughput of Shortcut Pairs Discovery in Ripser}
	
	\begin{figure}[h]
		\includegraphics[width=1.0\columnwidth]{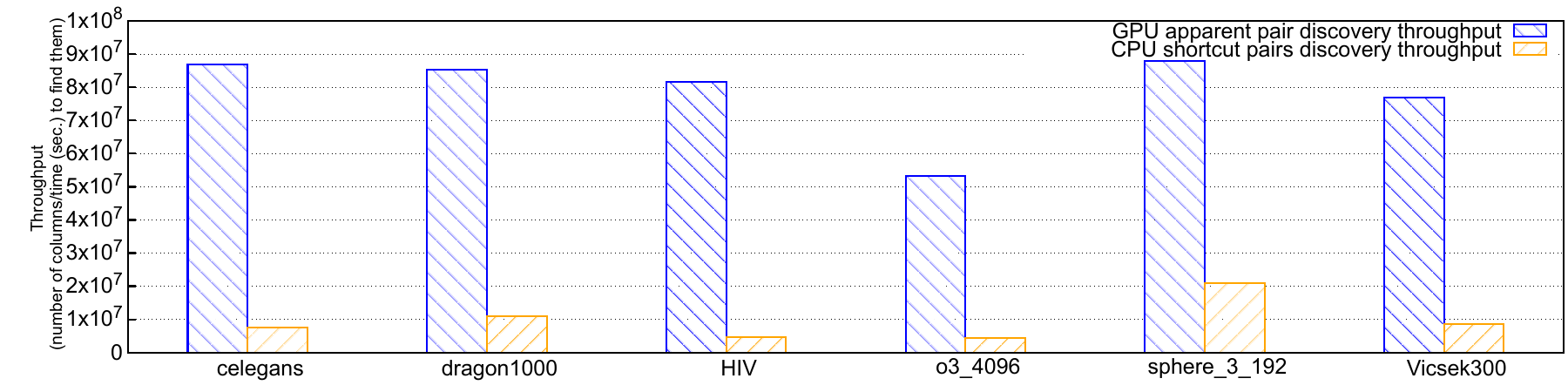}
		\caption{\small{A comparison of column discovery throughput 
		of apparent pair discovery with Ripser++ vs. Ripser's shortcut pair discovery. 
		The time is greatly reduced due to the parallel algorithm of finding apparent pairs on GPU (see Algorithm \ref{alg:apparentpairsalgorithm}). 
		}}
		\label{fig: apparentvsshortcut}
		\centering
	\end{figure}
	
	Discovering shortcut pairs in Ripser and discovering apparent pairs in Ripser++ account for a significant part of the computation. Thus, we compare the discovery throughput of these two types of pairs in Ripser and Ripser++, respectively. The throughput is calculated as the number of a specific type of pair divided by the time to find and store them. The results are reported in Figure \ref{fig: apparentvsshortcut}. We can find for all datasets, our GPU-based solution outperforms the CPU-based algorithm used in Ripser by 4.2x-12.3x. Since the number of the two types of pairs are almost the same (see Table \ref{tab:all-pairs}), such throughput improvement can lead to a significant saving in computation time.

	\subsection{Two-layer Data Structure for Memory Access Optimizations}
	
	\begin{table*} \centering
    \caption{\small{Hashmap Access Throughput, Counts, and Times Comparisons}}
    \label{tab: hashtable}
    \footnotesize
    \begin{tabular}{@{}lrrrrrr@{}} \toprule
              & R.++ write & R. write   & Num. of &  Num. of& R.++  &  R. \\
              & throuput &  throughput &   R.++ reads  &  R. reads     &  read & read \\
     Datasets & (pairs/s) & (pairs/s) & to data struct. & to hashmap & time (s) & time (s) \\
    \midrule
    {\it celegans}             & $7.21\times 10^8$ & $6.98\times 10^7$ & $3.22\times 10^4$ & $5.81\times 10^8$ & $0.00100$ & 11.43  \\
    {\it dragon1000}           & $7.62\times 10^8$ & $6.29\times 10^7$ & $1.19\times 10^5$ & $1.12\times 10^8$ & $0.00460$ & 1.28 \\
    {\it HIV}                  & $7.06\times 10^8$ & $8.85\times 10^7$ & $1.57\times 10^5$ & $3.10\times 10^8$ & $0.00130$ & 5.52 \\
    {\it o3} (sparse: $t=1.4$) & $4.78\times 10^8$ & $6.88\times 10^7$ & $1.65\times 10^6$ & $8.85\times 10^7$ & $0.01500$ & 0.56 \\
    {\it sphere\_3\_192}       & $7.32\times 10^8$ & $9.41\times 10^7$ & $2.71\times 10^5$ & $9.37\times 10^7$ & $0.00068$ & 0.30 \\
    {\it Vicsek300\_of\_300}   & $6.80\times 10^8$ & $8.82\times 10^7$ & $2.12\times 10^5$ & $5.67\times 10^8$ & $0.00053$ & 10.81 \\
    \bottomrule
    \end{tabular}
    \end{table*}

\label{sec:hashmap-performance}
	
	
	Table \ref{tab: hashtable} presents the write throughput of persistence pairs in pairs/s in the 2nd and 3rd columns. 
	In Ripser, we use the measured time of writing pairs to the hashmap to divide the total persistence pair number; while in Ripser++, the time includes writing to the two-layer data structure and sorting the array on GPU. The results show that Ripser++ consistently has one order of magnitude higher write throughput than that of Ripser. 
	
	Table \ref{tab: hashtable} also gives the number of reads in the 4th, 5th, and 6th columns as well as the time consumed in the read operations (in seconds) in the last column. 
	The number of reads in Ripser means the number of reads to its hashmap, while Ripser++ counts the number of reads to the data structure. 
	The reported results confirm that Ripser++ can reduce at least two orders of magnitude memory reads over Ripser. 
	A similar performance improvement can also be observed in the measured read time. 
	
	\subsection{Breakdown of Accelerated Components}
	\label{sec: appendix-breakdown}
	We breakdown the speedup on the two accelerated components of Vietoris-Rips persistence barcode computation over all dimensions $\geq 1$: matrix reduction vs. filtration construction with clearing. Ripser++ accelerates both stages of computation; however, which stage is accelerated more varies. For most datasets, it appears the filtration construction with clearing stage is accelerated more than the matrix reduction stage. This is because this stage is massively parallelized in its entirety while accelerated matrix reduction only parallelizes the finding and aggregation/management of apparent pairs. Speedups on filtration construction with clearing range from 2.87x to 42.67x while speedups on matrix reduction range from 5.90x to 36.71x. 
	
	\begin{figure}[h]
		\includegraphics[width=1.0\columnwidth]{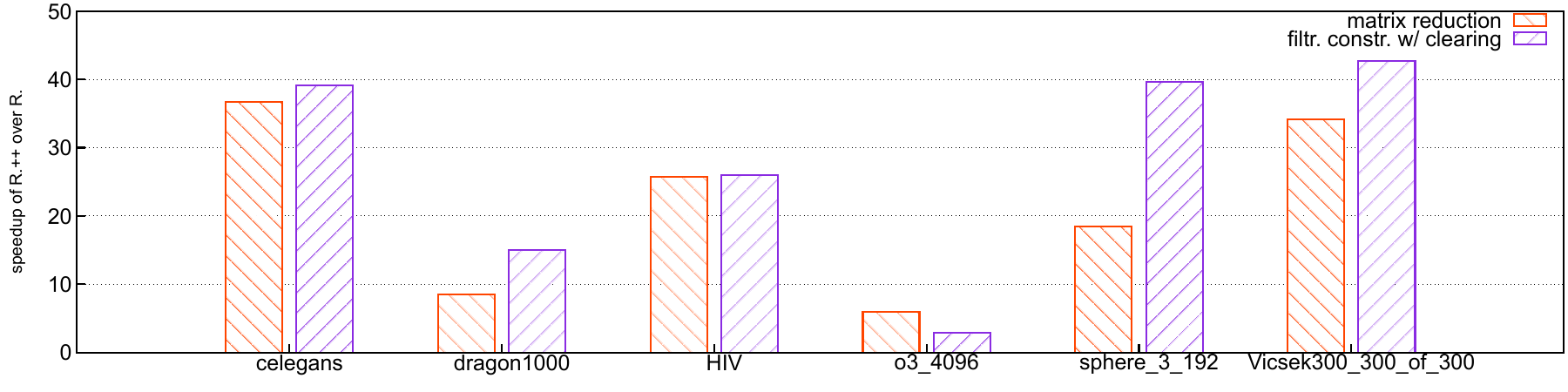}
		\caption{\small{A breakdown of the speedup of Ripser++ over Ripser for computation beyond dimension 0 into the two stages: matrix reduction and filtration construction with clearing.  
		}}
		\label{fig: breakdown}
		\centering
	\end{figure}
	

	\subsection{Experiments on the Apparent Pairs Rate in Dimension 1}
	\label{sec: experiments-apparent}
	We run extensive experiments to analyze the number of apparent pairs in the average case of a random distance matrix input. We make an assumption and one observation about the values of a random distance matrix.
	
    \begin{assumption}
    \label{ass: different diameters}
    In practice, distances between points are almost never exactly equal. Thus we assume the entries of the distance matrix are all different.
    \end{assumption}

    \begin{observation}
    \label{obs: reassignment}
    The persistence barcodes (see Section \ref{sec: preliminaries-computation} on definition of barcodes) executed by the persistent homology algorithm do not change up to endpoint reassignment if we reassign the distances of the input distance matrix while preserving the total order amongst all distances. 
    \end{observation}
    
    Thus the setup for our experiments is to uniformly at random sample permutations of the integers 1 through $n(n_0-1)/2$ to fill the lower triangular portion of a distance matrix, where $n_0$ is the number of points. We run Ripser++ for $n_0$ = 50, 100, 200, 400, ..., 1000, ..., 9000 with 10 uniformly random samples with replacement of $n_0(n_0-1)/2$-permutations for a fixed random seed for dimension 1 persistence. We consider the general combinatorial case where the distance matrix does not necessarily satisfy the triangle inequality and thus that the set of points may not form a finite metric space. This is still valid input, as Vietoris-Rips barcode computation is dependent only on the edge relations between points (e.g. the 1-skeleton). 
    \begin{figure}[!h]
		\includegraphics[width=1.0\columnwidth]{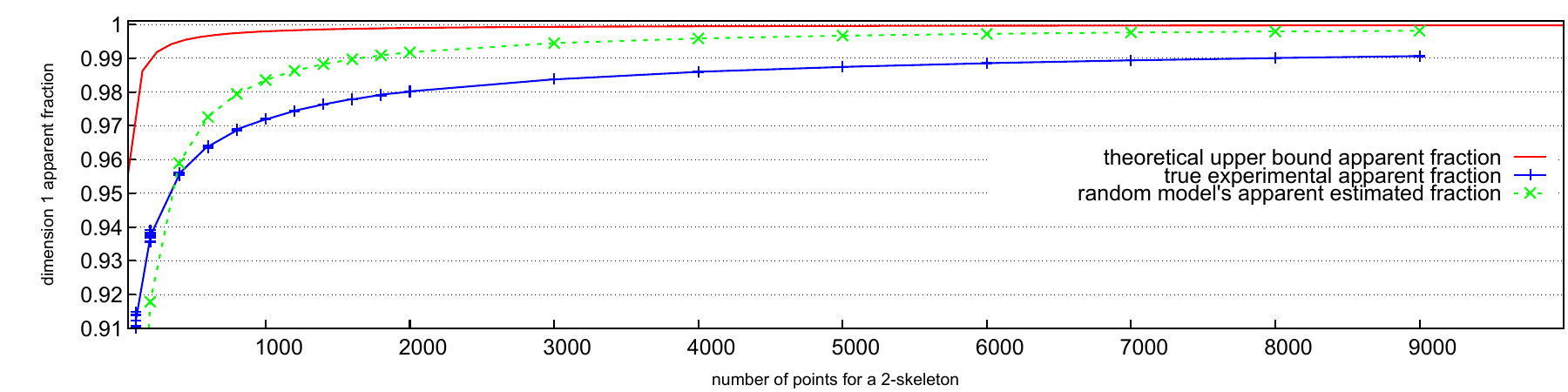}
		\caption{\small{Three different curves of the apparent fraction for a 1-dimensional coboundary matrix as a function of the number of points. The theoretical upper bound is for the case of all diameters the same, but also can be achieved when all diameters are different. Interestingly, the variance of the apparent fraction for each point count from the experiments is very low even though the lower distance matrix entries are uniformly at random permuted. The dotted curve is the piecewise linear interpolated curve of the random model that matches the shape of the empirical and theoretical curve. For more details about the random mathematical model, see the Section \ref{sec: appendix-apparent-model}}}
		\label{fig: apparent-ratios}
		\centering
	\end{figure}
    Figure \ref{fig: apparent-ratios} shows the plot for the percentage of apparent pairs with respect to the total number of 1-simplices for a 1-dimensional coboundary matrix (the apparent fraction) as a function of the number of points of a full 2-skeleton for three different contexts. The first is the theoretical upper bound of $(n_0-2)/n_0$ proven in Theorem \ref{theorem: apparent bounds}, the second is the actual percentage found by our experiments, and the third is the percentage predicted by our random model for the case of dimension 1 coboundary matrices (see Section \ref{sec: appendix-apparent-model}).  
	
	From the experiments, we notice that in the average case there is still a large number of apparent pairs and this number is close to and closely guided by the theoretical upper bound found in Theorem \ref{theorem: apparent bounds}. Furthermore, the model's curve and the theoretical upper bound are asymptotic to 1.00 as $n_0 \rightarrow \infty$. This is calculated by some algebraic manipulations of radical equations. Furthermore, we performed the true apparent pairs fraction search experiment up to 20000 points, with consistent monotonic behavior toward 1.00. For example, at 10000, 20000 points we obtain  0.991127743, and 0.993733522 average apparent fractions respectively.  
    
    \subsection{Algorithm for Randomly Assigning Apparent Pairs on a Full 2-Skeleton}
	
\begin{algorithm}[h] 
\SetKwComment{Comment}{//}{}
    \caption{Algorithm for Random Apparent Pairs Construction}
    \label{alg: apparentconstruction}
    
    \KwIn{$d_i$: a sequence of diameters with $d_1 > d_2 > \ldots > d_{n \choose 2}$; $K = (V, E, T)$: a full 2-skeleton on $n$ points where $V$ is a set of $n$ vertices, $E$ is a set of ${n \choose 2}$ edges, and $\tau$ is a set of ${n \choose 3}$ triangles.}
    \KwOut{A sequence of apparent pairs of edges and triangles $(e_i, t_i)$ with $\textsf{diam}(e_i) = \textsf{diam}(t_i) = d_i$.}
    \SetKwProg{Fn}{Function }{:}{}
   \SetKwFunction{RDA}{RandomDiameterAssignment}{}
   \Fn{\RDA{$K$}} {
    \While{there are triangles left in $\tau$}{
        Uniformly at random pick a 1-dimensional simplex $e_i \in E$ \;
        Assign edge $e_i$ a diameter $d_i$ strictly less than all $d_k$ for $k < i$  \Comment{(e.g., $d_i = {n \choose 2} - i + 1$)} 
        \If{there are triangles incident to $e_i$}{
            Pair up $e_i$ with its oldest cofacet, the unique triangle $t_i \in T$ of highest lexicographic order amongst remaining triangles incident to $e_i$ \;
            Emit $(e_i, t_i)$ \;
            Remove $e_i$ from $E$ and all triangles $t'_i$ containing $e_i$ in their boundary from $\tau$ since these triangles must all have the same diameter $d_i$ \;
            }
        }
    }
\end{algorithm}
    
    Algorithm \ref{alg: apparentconstruction} assigns diameters to a subset of the edges in decreasing order so that each diameter value $d_i$ at iteration $i$ results (if possible) in an apparent pair $(e_i,t_i)$ with edge $e_i$ and triangle $t_i$ both of diameter $d_i$. In the Algorithm \ref{alg: apparentconstruction} at line 5, since we assume at iteration $i, i>i'$ that $d_i<d_{i'}$ and that all triangles of diameter greater than or equal to $d_{i'}$ have already been removed, at iteration $i$ all remaining cofacets $t'$ of $e_i$ must have the same diameter as $e_i$. 
    Recalling Assumption \ref{ass: different diameters} and Observation \ref{obs: reassignment}, this random algorithm is equivalent to uniformly at random assigning permutations of the numbers 1,...,${n \choose 2}$ to the lower triangular part of a symmetric distance matrix $D$ and counting the number of apparent pairs in the 1-dimensional coboundary matrix induced by $D$.
    
    \subsection{A Greedy Deterministic Distance Assignment}
    We consider Algorithm \ref{alg: apparentconstruction} with line 3 changed to pick $e_i \in E$ with maximum number of cofacets remaining, with largest combinatorial index if there is a tie. This greedy deterministic algorithm serves as a lower bound to Algorithm \ref{alg: apparentconstruction}. We plot the resulting apparent fraction and experimentally verify that $\frac{1}{(p+2)}$ is a theoretical lower bound for $p=1$; (notice the theoretical lower bound did not depend on the diameter condition on simplices containing the maximum indexed point of Theorem \ref{theorem: apparent bounds}). 
    It is currently unknown what the theoretical relationship is between the greedy deterministic algorithm and any theoretical lower bounding curve.
    
    \begin{figure}[h]
		\includegraphics[width=1.0\columnwidth]{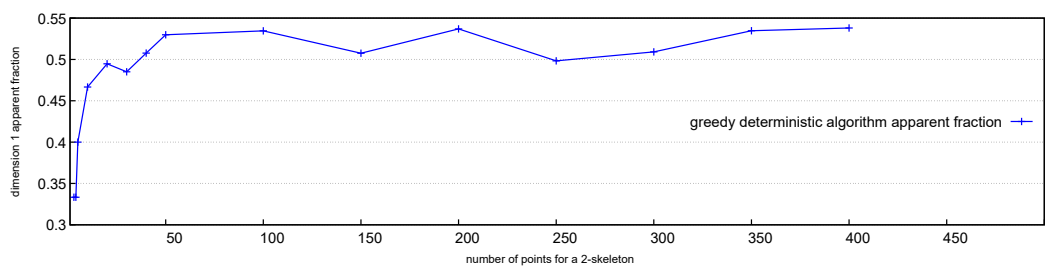}
		\caption{\small{The deterministic greedy apparent fraction curve. Notice the theoretical lower bound of 0.3333 is confirmed experimentally by the experimental curve. The experiments show the apparent fraction stays within a neighborhood of 0.5 as n gets large enough.}}
		\label{fig: greedy-ratio}
		\centering
	\end{figure}
    
    \subsection{A Random Approximation Model for the Number of Apparent Pairs in a Full 2-Skeleton on n Points}
	\label{sec: appendix-apparent-model}
    
    We construct a model for the analysis of random Algorithm \ref{alg: apparentconstruction}. We model Algorithm \ref{alg: apparentconstruction} by analyzing a modified algorithm. Let there be a full 2-skeleton $K= (V,E,T)$ as in Algorithm \ref{alg: apparentconstruction}. Let $E' \subseteq E$ be the subset of edges not including the single point $v \in V$ of highest index: $n_0-1$ and $T' \subseteq T$ be the subset of triangles induced by $E'$. Modify Algorithm \ref{alg: apparentconstruction} to let $j \leq {n_0-1 \choose 2}$ be the fixed number of iterations of the loop, replacing line 8. Modify Algorithm \ref{alg: apparentconstruction} at line 3 to choose uniformly at random from $E'$ instead of $E$, outputting a sequence $C$ with $j$ different edges and having diameters $d_1>d_2>...>d_j$.
    
    We pick edges from $E'$ since this ensures that the If in line 4 of Algorithm \ref{alg: apparentconstruction} will always evaluate to true by the existence of triangles containing vertex $n_0-1$ and thus that there are at least $\|C\|=j$ number of apparent pairs in $K$. $\|C\|=j$ is equal to the number of iterations of the algorithm. After choosing $j$ edges, we count how many triangles are still left in $T'$ in expectation. 
    
    We define a Bernoulli random variable for each triangle $t \in T'$ of the full 2-skeleton $K$. 
   
    $$
    X_{t,j} = \left\{
     \begin{array}{lr}
       1 &  \text{if triangle t}\in T'\text{ is not incident to any edges in } C\\
       0 &  \text{otherwise}
     \end{array}
   \right.$$
    
    We notice that for every triangle, the same random variable can be defined on it, all identically distributed. 
    
    Let $$p_{t,j}= \frac{({n_0-1 \choose 2}-3)\cdot({n_0-1 \choose 2}-4)\cdots({n_0-1 \choose 2}-3-j+1)}{({n_0-1 \choose 2}\cdot({n_0-1 \choose 2}-1)\cdots({n_0-1 \choose 2}-j+1))}$$ be the probability of triangle $t \in T$ not containing any of the $j$ chosen edges in its boundary of 3 edges.
    
    We thus define the random variable $T_j$= $\Sigma_{t \in T'}X_{t,j}$ to count the number of triangles remaining after $j$ edges are chosen in sequence.
    
    Taking expectation, we get 
    $$E[T_j] = \Sigma_{t \in T'}E[X_t] = \Sigma_{t \in T'} 1 \cdot p_{t,j} = {n_0-1 \choose 3} \cdot \frac{({n_0-1 \choose 2}-3)\cdot({n_0-1 \choose 2}-4)\cdots({n_0-1 \choose 2}-3-j+1)}{({n_0-1 \choose 2}\cdot({n_0-1 \choose 2}-1)\cdots({n_0-1 \choose 2}-j+1))}$$
    
    by linearity of expectation, the definition of $T'$ and the definition of $p_{t,j}$.

    Set $E[T_j]=\tau$, with $\tau$ the number of triangles reserved to not be incident to the sequence $C$ of $j$ apparent edges of $K$. Then solve for $j$ from the equation $E[T_j]=\tau$ with a numerical equation solver system; then divide $j$ by ${n \choose 2}$, the total number of edges, and call this the ratio $r_{\tau}(n)$. Since we are just building a mathematical model to match experiment, we fit our curve $r_{\tau}(n)$ to the true experimental curve from Figure \ref{fig: apparent-ratios}. We choose $\tau=500$ to minimize the least squares error on the sampled values of n in Figure \ref{fig: apparent-ratios} (we assume $n,\tau$ s.t. $\binom{n}{3} > \tau$). to the averaged experimental curve, formed by averaging the apparent fraction for each $n$. $\tau$ was chosen in units of 100s due to the high computational cost of solving for $j$ for every $n$. We then obtain the dotted curve in Figure \ref{fig: apparent-ratios}, $r_{500}(n)$: a function of $n$, the number of points.
    
    The shape of the model's curve, which matches experiment and stays within theoretical bounds is the primary goal of our model. The constant, $\tau$=500, suggests that as the number of points increases, in practice the expected percentage of triangles not a cofacet of an apparent edge decreases to 0 and that the the expected value is approximately a constant value. See Section \ref{sec: appendix-equivalent-model} for an equivalent model that counts edges and triangles differently.
    
    \section{The ``width" and ``depth" of Computing Vietoris-Rips Barcodes}
	We are motivated by a common phenomenon in computation of Vietoris-Rips barcodes found in \cite{zhang2019hypha} for the matrix reduction stage: the required sequential computation is concentrated on a very few number of columns in comparison with the filtration size. We further generalize a principle to quantify parallelism in the computation as a guidance for parallel processing.

	We define two concepts: ``computational width" as a measurement of the amount of independent operations and ``computational depth" as a measurement of the amount of dependent operations. Their quotients measure the level of parallelism in computation. We consider rough upper and lower bounds on parallelism for Vietoris-Rips barcode computation using these quotients. 
	
	For an upper bound, let the ``computational width" be 2 $\times$ the number of simplices in the filtration or 2 $\times$ the total number of uncleared columns of the coboundary matrix where the 2 comes from the two stages of persistence computation: filtration construction and matrix reduction. (This quantifies the maximum amount of independent columns achievable if all columns were independent of each other). This rationale comes from the existence of a large percentage of columns that are truly independent of each other (e.g. apparent columns during matrix reduction) as well as the independence amongst simplices for their construction and storage during filtration construction (assume the full Rips computation case). 
	Let the ``computational depth" be 1 + the amount of column additions amongst columns requiring at least one column addition. (The + 1 is to prevent dividing by zero). 
	In this case, the ``computational width" divided by the ``computational depth" thus quantifies an upper bound on the amount of parallelism available for computation. 
	
	For a lower bound one could similarly consider the ``computational width" as the number of apparent columns divided by the ``computational depth" as the sum of all column additions amongst columns plus any dependencies during filtration construction.
	
	
	These bounds along with empirical results on columns additions \cite{zhang2019hypha}, the percentage of apparent pairs in Table \ref{tab:all-pairs}, and the potentially several orders of magnitude factor difference in number of simplices compared to column additions, suggest that there can potentially be a high level of hidden parallelism suitable for GPU in computing Vietoris-Rips barcodes. We are thus led to aim for two objectives for effective performance optimization:
	\begin{enumerate}
		\item To massively parallelize in the ``computational width" direction (e.g. parallelize independent simplex-wise operations).
		\item To exploit locality in the ``computational depth" direction (while using sparse representations).
	\end{enumerate}
	
	Objective 1 is well achievable on GPU while Objective 2 is known to be best done on CPU. In fact depth execution such as column additions are best done sequentially due to the few number of ``tail columns" \cite{zhang2019hypha} of the coboundary matrix of Vietoris-Rips filtrations.
	
	Our ``width" and ``depth" principle gives bounds for developing potential parallel algorithms to accelerate, for example, Vietoris-Rips barcode computation. However, real-world performance improvements must be measured empirically, as in Table \ref{tab: performance}. 
	
	\begin{figure}[h]
		\includegraphics[width=1.0\columnwidth]{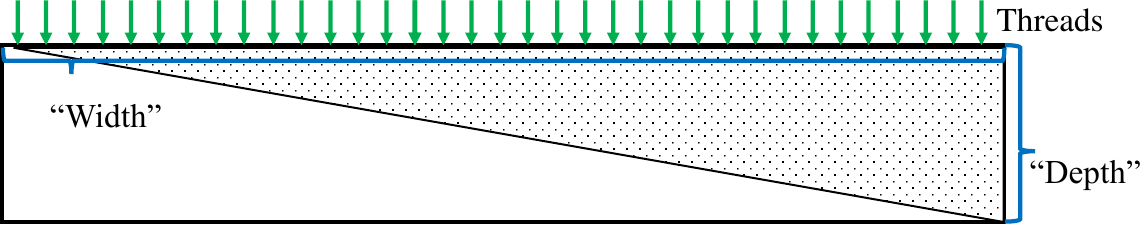}
		\caption{\small{Illustration of the ``width" and ``depth" of computation. The area of the triangle represents the total amount of work performed. Furthermore, the ``width" divided by the ``depth" quantifies the level of parallelism in computation.}}
		\label{fig: widthanddepeth}
		\centering
	\end{figure}

 \subsection{The Growth of the Proportion of Apparent Pairs} \label{sec: appendix-apparent-percentage}
	In our experiments on smaller datasets with high dimension and a low number of points, all persistence pairs eventually become trivial at a relatively low dimension compared to the number of points. Theoretically, of course, there always exists some dimension at which the filtration collapses in persistence (all pairs become 0-persistence), e.g. $n$ points with an $n-1$ dimensional simplex. Apparent pairs are a subset of the 0-persistence pairs, we show empirically in Table \ref{tab: appendix-apparent-percentage} that the proportion of apparent pairs of dimension $p$ to the number of pairs (including infinite persistence ``pairs") in dimension $p$ grows with dimension. The only exception to this is the o3 dataset with sparse computation and a restrictive threshold. We believe this outlier is due to the threshold condition. The other datasets were computed without any threshold condition.
	
	\begin{table*} [!h]\centering
    \caption{\small{Apparent Pair Percentage (Out of All Persistence Pairs) per Dimension}}
    \label{tab: appendix-apparent-percentage}
    \begin{tabular}{@{}lrrrrrr@{}} \toprule
              &  & & \% apparent & \% apparent & \% apparent \\
     Datasets & n      &d  &   in dim 1 &  in dim 2 &  in dim 3 \\
    \midrule
    {\it celegans} & 297&3 &99.661017\% & 99.962021\% & 99.977972\%  \\
    {\it dragon1000} & 1000 & 2 & 99.937011\% & 99.972234\% &   \\
    {\it HIV} & 1088 & 2 & 99.92071\% & 99.972234\% &   \\
    {\it o3} (sparse: $t=1.4$) & 4096 & 3 & 98.928064\% & 98.651461\% & 98.634918\%  \\
    {\it sphere\_3\_192} & 192 & 3 & 99.707909\% & 99.734764\% & 99.802291\%\\
    {\it Vicsek300\_of\_300} & 300 & 3 & 99.849611\% & 99.925678\% & 99.966999\% \\
    \bottomrule
    \end{tabular}
    \end{table*}
	\subsection{Empirical Properties of Filtering by Diameter Threshold and Clearing}
	\begin{table*}[!h]
	\centering
	\caption{\small{Empirical results on Clearing and Threshold Restriction}}
	\label{tab:clearing-threshold}
	\begin{tabular}{@{}lrrrrrrr@{}} \toprule
		
		&   &   & \scriptsize possible         &\scriptsize  num. cols.  &\scriptsize simpl. removed &\scriptsize cols. & \scriptsize \% of simplices \\
		\scriptsize Datasets & \scriptsize n &\scriptsize d & \scriptsize num. simpl. & \scriptsize to reduce &\scriptsize by diameter & \scriptsize cleared & \scriptsize sel. for red. \\
		
		\midrule
		\scriptsize {\it celegans}           & \scriptsize 297 & \scriptsize 3 & \scriptsize 322,058,286 & \scriptsize 256,704,712 & \scriptsize 61,576,563 & \scriptsize 3,777,011 & \scriptsize 79.7075322\% \\
		\scriptsize {\it dragon1000}         & \scriptsize 1000 & \scriptsize 2 & \scriptsize  166,666,500  & \scriptsize 56,110,140   & \scriptsize 110,237,919 & \scriptsize 318,441 & \scriptsize 33.6661177\%\\
		\scriptsize {\it HIV}                & \scriptsize  1088 & \scriptsize 2 & \scriptsize 214,652,064 & \scriptsize 155,009,693 & \scriptsize 59,123,662 & \scriptsize 518,709 & \scriptsize 72.2143967\%\\
		\scriptsize {\it o3} (sparse: $t=1.4$) & \scriptsize 4096 & \scriptsize 3 & \scriptsize $1.17\times10^{13}$ & \scriptsize 44,081,360 & \scriptsize $1.17\times 10^{13}$& \scriptsize 4,347,112 & \scriptsize 0.00037604\%\\
		\scriptsize {\it sphere\_3\_192}     & \scriptsize 192 & \scriptsize 3 & \scriptsize 55,004,996 & \scriptsize 46,817,416 & \scriptsize 8,159,941 & \scriptsize 1,072,739 & \scriptsize 85.114843\%\\
		\scriptsize {\it Vicsek300\_of\_300} & \scriptsize 300 & \scriptsize 3 & \scriptsize  335,291,125 & \scriptsize 283,441,085 & \scriptsize 47,803,132 & \scriptsize 4,046,908 & \scriptsize 84.5358150\%\\
		\bottomrule
	\end{tabular}
\end{table*}
	
	We have done an empirical study on the filtration construction with clearing stage of computation. By Table \ref{tab:clearing-threshold}, except for dragon1000, all full Rips filtrations, after applying the enclosing radius condition (see Section \ref{sec: appendix-enclosing-radius}) and clearing (see Section \ref{sec: clearing-lemma}), result in a large percentage of simplices selected for reduction. More importantly, Algorithm \ref{alg:fullripsalgorithm-sorting} sorts all the possible number of simplices for full Rips computation. If the number of simplices selected is close to the number of simplices sorted, the sorting is effective. By our experiments, even in the dragon1000 case, sorting all simplices is still faster than CPU filtering or even warp-based filtering on GPU (see Section \ref{sec: warp-filtering}).
	
	By Section \ref{sec: appendix-enclosing-radius}, the enclosing radius eliminates simplices when added to a growing simplicial complex, will only contribute 0-persistence pairs. This is equivalent to truncating the coboundary matrix to its bottom right block (or zeroing such rows and columns). There are usually fewer columns zeroed by the clearing lemma than zeroed by the threshold condition, however they correspond to columns that must be completely zeroed during reduction and a lot will form ``tail columns" that dominate the time of matrix reduction.
	
	Notice that the o3\_4096 dataset has a predefined threshold less than the enclosing radius. The number of possible simplices is several orders of magnitude larger than the actually number of columns needed to reduce after clearing. Thus we must use the sparse option of computation and avoid using Algorithm \ref{alg:fullripsalgorithm-sorting}. This means memory is not allocated as a function of the number of possible simplices and is instead allocated with respect to the GPU's memory capacity  as we grow the number of columns for matrix reduction. Sorting is not used and a warp-base filtering is used instead (see Section \ref{sec: warp-filtering}). 

    \section{Related Persistent Homology Computation Software}
	In Section \ref{sec: preliminaries-computation}, we briefly introduce several software to compute persistent homology, paying special attention on computing Vietoris-Rips barcodes. In this section, we elaborate more on this topic. 
	
	The basic algorithm upon which all such software are based on is given in Algorithm \ref{alg:standard-algorithm}. Many optimizations are used by these software, and have made significant progress over the basic computation of Algorithm \ref{alg:standard-algorithm}. Amongst all such software, Ripser is known to achieve state of the art time and memory performance in computing Vietoris-Rips barcodes~\cite{bauer2019ripser, otter2017roadmap}. Thus it should suffice to compare our time and memory performance against Ripser alone. We overview a few, and certainly not all, of the related software besides Ripser. 
	
	Gudhi \cite{gudhi:urm} is a software that computes persistent homology for many filtration types. It uses a general data structure called a simplex tree \cite{boissonnat2012simplex, boissonnat:hal-01883836} for general simplicial complexes for storage of simplices and related operations on simplices as well as the compressed annotation matrix algorithm \cite{boissonnat2013compressed} for computing persistent cohomology. It can compute Vietoris-Rips barcodes.
	
	Eirene \cite{henselman2016eirene} is another software for computing persistent homology. It can also compute Vietoris-Rips barcodes. One of its important features is that it is able to compute cycle representatives. Paper \cite{hylton2017performance} details how Eirene can be optimized with GPU.
	
	Hypha (a hybrid persistent homology matrix reduction accelerator) \cite{zhang2019hypha} is a recent open source software for the matrix reduction part of computing persistent homology for explicitly represented boundary matrices, similar in style to \cite{bauer2017phat, bauer2014distributed}. Hypha is one of the first publicly available softwares using GPU. A framework based on the separation of parallelisms is designed and implemented in Hypha due to the existence of atleast two very different execution patterns during matrix reduction. Hypha also finds apparent pairs on GPU and subsequently forms a submatrix on multi-core similar to in Ripser++. 
	
\section{Conclusion}
\label{sec:impact}

Ripser++ can achieve significant speedup (up to 20x-30x) on representative datasets in our work and thus opens up unprecedented opportunities in many application areas. For example, fast streaming applications \cite{syzdykbayev2019persistent} or point clouds from neuroscience \cite{bendich2016persistent} that spent minutes can now be computed in seconds, significatly advancing the domain fields.



We identify specific properties of Vietoris-Rips filtrations such as the simplicity of diameter computations by individual threads on GPU for Ripser++. Related discussions, both theoretical and empirical, suggest that our approach be applicable to other filtration types such as cubical \cite{bauer2014distributed}, flag \cite{luetgehetmann2019computing}, and alpha shapes \cite{gudhi:urm}. We strongly believe that our acceleration methods are widely applicable beyond computing Vietoris-Rips persistence barcodes. 

We have described the mathematical, algorithmic, and experimental-based foundations of Ripser++, a GPU-accelerated software for computing Vietoris-Rips persistence barcodes. Computationally, we develop massively parallel algorithms directly tied to the GPU hardware, breaking several sequential computation bottlenecks. These bottlenecks include the Filtration Construction with Clearing stage, Finding Apparent Pairs, and the efficient management of persistence pairs for submatrix reduction. Theoretically we have looked into properties of apparent pairs, including the Apparent Pairs Lemma for massively parallel computation and a Theorem of upper and lower bounds on their significantly large count. Empirically we have performed extensive experiments, showing the true consistent behavior of apparent pairs on both random distance matrices as well as real-world datasets, closely matching the theoretical upper bound we have shown. Furthermore, we have measured the time, memory allocation, and memory access performance of our software against the original Ripser software. We achieve up to 2.0x CPU memory efficiency, besides also significantly reducing execution time. We hope to lead a new direction in the field of topological data analysis, accelerating computation in the post Moore's law era and turning theoretical and algorithmic opportunities into a high performance computing reality.

%% file: appendix.tex
	\appendix
    \clearpage
	\section{Some Proofs}
    
    We can connect $p$-dimensional simplices with their next higher $p+1$-dimension via a single point, which we call the \textbf{apex}.
    \begin{definition}
    Let $p \leq d_{amb}$.
    
    For a $t$-equilateral $p$-dimensional simplex $\sigma \subseteq \mathbb{R}^{d_{amb}}$, we call any point $a \in \mathbb{R}^{d_{amb}}$ that makes $\textsf{conv}(\sigma \cup \{a\})$ a $t$-equilateral $p+1$-dimensional simplex as a $t$-equilateral-\textbf{apex} of $\sigma$.

    Let $\textbf{t\text{-equilateral-apexes}}(\sigma)$ denote the set of all $t$-equilateral-apex points for $\sigma$.
    \end{definition}\begin{proposition}\label{prop: twoorthoapex}
    Amongst all apexes for a $t$-equilateral $p$-dimensional simplex $\sigma$, there are exactly two apexes that pass through the line passing through the centroid of a $p$-dimensional simplex $\sigma$ in the direction of one of its normals $\overrightarrow{n}_{\sigma} \in N_{\sigma}$. 
    \end{proposition}
    \begin{proof}
       The line passing through the centroid of a $p$-dimensional simplex $\sigma$ in the direction of its normal comprises the set of points:
    $\{c_{\sigma}+t\overrightarrow{n}_{\sigma}: t \in \mathbb{R}\}$.

    Since the line only has two directions, it can have atmost two intersection points with \\$\textbf{t\text{-equilateral-apexes}}(\sigma)$. Both intersection points are apexes since the line is centered at $c_{\sigma}$ as can be checked. 
    \end{proof}
    \begin{proposition}(Inscribing a  Polytope)\label{prop: ball-intersection-polytope}
    
    Let $t>0$ and $p,d_{amb} \in \mathbb{N}, 0\leq p\leq d_{amb}$.

    For $p+1$ $t$-equidistant points $\sigma_p=\{x_1,...,x_{p+1}\}$, there exists a polytope $P$ with  
    \begin{equation}
        P \subseteq \bigcap_{i=1}^{p+1}B(x_i,t)  \subseteq \mathbb{R}^{d_{amb}}, P \neq \emptyset
    \end{equation} 
    spanning $2^{d_{amb}-p}(p+1)$ many points. 
    \end{proposition}
    \begin{proof}
        Let \begin{equation}
            S_{p+1}= \bigcap_{i=1}^{p+1} \partial(B(x_i,t)) 
        \end{equation}
        for $p\leq d_{amb}$. 

        Certainly by definition of the boundary of a ball, all $s \in S_{p+1}$ have the property that $\|s-x_i\|_2 =t, \forall i=1,...,p+1$.

        We construct $P$ inductively.

        Let $P_{q}, p \leq q \leq d_{amb}$ be a collection of point sets we construct from $\sigma_p$. This can be defined inductively:
        
        \begin{enumerate}
            \item $P_p\gets \{\sigma_p\}$.
            \item $P_{q+1}\gets P_{q} \cup \{\textsf{conv}((\{u\} \cup \tau): \forall u \in (S_{p+1} \cap\text{apexes}(\tau) \cap \{c_{\tau}+t \overrightarrow{n}_{\tau} : t \in \mathbb{R}, \text{ for a single }\overrightarrow{n}_{\tau} \in N_{\tau}\}), \forall \tau \in P_{q}\} $
        \end{enumerate}
        From $q$ to $q+1$, we add $t$-equilateral apex points for each simplex $\tau$ in $P_q$ on the orthogonal axis passing through the centroid $c_{\tau}$ of $\tau$. By Proposition \ref{prop: twoorthoapex}, we are adding two diametrically opposing apex points. 
        
        Since taking the convex hull of a $t$-equidistant simplex with its apex maintains $t$-equidistance, all simplices in $P_q$ are $t$-equidistant simplices. Furthermore, since $\tau \subseteq \textsf{conv}((\{u\} \cup \tau) , \forall \tau \in P_q, p\leq q\leq d_{amb}$ and $\sigma_p \in P_p$, we must have that $\sigma_p\subseteq \tau, \forall \tau \in P_q, p\leq q\leq d_{amb}$. Thus, each $\tau$ is the convex hull of $\sigma_p$ with points $t$-equidistant from all of the points $x_1,....,x_{p+1}$. 

        Since $\bigcap_{i=1}^{p+1}B(x_i,t)$ is convex, each $\tau \in P_q$ is the smallest convex set containing $x_1,...,x_{p+1}$ and the apex points used to form $\tau$, and all apex points belong to $S_{p+1}$. We must have that $\tau \subseteq \bigcap_{i=1}^{p+1}B(x_i,t)$.
         
        Since constructing $P_{q+1}$ from $P_q$ introduces $2\lvert P_q \rvert$ many new apex points. From $P_p$ to $P_{d_{amb}}$ we must double the size of the starting set of $p+1$ equidistant points $d_{amb}-p$ times. 

        Let $P\gets \bigcup_{\tau \in P_{d_{amb}}} \tau$. 
        
       There are in total $2^{d_{amb}-p}(p+1)$ many points used to span $P$. 

       It is a union of $d_{amb}$-dimensional simplices where by construction of $P_{d_{amb}}$, each simplex intersects another at a $d_{amb}-1$ dimensional face and all simplices were constructed from $d_{amb}-1$-dimensional faces. Thus $P$ is a polytope.
    \end{proof}
    \subsection{Cavities in Hypercubes}
    Here we discuss cavities in hypercubes. This section supports Corollary \ref{cor: k-cavity-corollary}. 
    Here we define a cavity in a hypercube, which is intuitively a convexly deformed high dimensional closed ball.
    \begin{definition}
        A cavity in a hypercube $H=[0,1]^{d_{amb}}$ is a strict open subset $C\subsetneq X$ where $C$ is a convex connected compact set diffeomorphic to the hypercube $[0,1]^{d_{amb}}$.

        The boundary of a cavity, $\partial(C)$ is the subset of $C$ that is diffeomorphic to the $d_{amb}$-dimensional sphere: $\mathbb{S}^{d_{amb}-1}$. 
    \end{definition}
    \begin{definition}
        For a cavity $C$, a Gauss map $\mathbf{n}:  \partial(C)\rightarrow \mathbb{S}^{d_{amb}-1}$ defines the unit normal vector at each point $p \in\partial(C)$. This is the vector orthogonal to the tangent space of $\partial(C)$ at $p$. 
    \end{definition}
    \begin{definition}
        For a cavity $C \subsetneq [0,1]^{d_{amb}}$, we can define the medial axis $M(C) \subseteq C$ as the the set of points where every point on $M(C)$ has atleast two closest points from $\partial(C)$: 
        \begin{equation}
           M(C)\triangleq \{x \in C: 
           \lvert \arg\min_{y\in \partial(C)} d(x,y) \rvert \geq 2\}
        \end{equation}
        \end{definition}
        
    \begin{definition}
    Let $C$ be a cavity of a hypercube $[0,1]^{d_{amb}}$, the in-Radius over $C$ is defined as follows: 
\begin{equation}
        \textsf{inRadius}(C)\triangleq \min_{x \in \partial(C)} d(M,x)
    \end{equation}
    \end{definition}
    \begin{lemma}\label{lemma: appendix-k-cavities-betti}
        For a hypercube $H=[0,1]^{d_{amb}}$  containing $k$ cavities $C_1,...,C_k$ and let $0<R<1$ satisfy the following inequality: 
        \begin{equation}
         R< \min(\textsf{inRadius}(C_i),d(C_i,\partial(H)), \min_j\min_{x \in C_i} d(x,C_j), \forall i=1,...,k
        \end{equation}
        where the last two quantities in the minimum are the distances between cavities and the distances from cavities to the boundary. 
        
        If for any finite sample set $S \subseteq H$ every unit cell 
        \begin{equation}
            c_{j_1,...,j_{d_{amb}}}:=\Pi_{i=1}^{d_{amb}}[\frac{j_iR}{2\sqrt{d_{amb}}}, \frac{(j_i+1)R}{2\sqrt{d_{amb}}}] \subseteq H\setminus \bigcup_{i=1}^k C_i
        \end{equation}
        contains atleast one point from $S$, 
       then: 
        \begin{equation}
             \beta_p\geq k, \forall p=1,...,d_{amb}\text{ for the Betti numbers on }\textsf{VR}_{R}(S)
        \end{equation}
    \end{lemma}
    \begin{proof}
        For each cavity $C$, consider the following ball:
        \begin{equation}
            B:=\{x \in C: d(x,x^*)<R,\min_{y \in \partial(C)}d(x^*,y)=\textsf{inRadius}(C)\}
        \end{equation}
        Certainly $B \subsetneq C$ since $R< \textsf{inRadius}(C)$.

        \textbf{Claim 1. }We claim that no pair of boundary points on $\partial(C)$ that pass through $B$ can have distance less than or equal to $2R$:

        Say there is a pair $(x,y) \in \partial(C)\times \partial(C)$ with $d(x,y) \leq 2R$ and so that there is a point $z \in B$ with: 
        \begin{equation}
            d(x,y)=d(x,z)+d(z,y)\leq 2R
        \end{equation}
        By minimality of $x^*$, we have upon replacing with the suboptimal $z$:  
        \begin{subequations}
        \begin{equation}
            d(x,z)\geq \min_{x' \in \partial(C)} d(x',z)> R
        \end{equation}
        \begin{equation}
            d(y,z)\geq \min_{y' \in \partial(C)} d(y',z)> R
        \end{equation}
        \end{subequations}
        Thus:
        \begin{equation}
        d(x,y)> 2R
        \end{equation}
        This is a contradiction. 

       \textbf{2. }By Claim 1, we must have that no pair of points on $\partial(C)$ can pass through the ball $B$ with distance less than $2R$. 

       By assumption, we have that the grid of resolution $\frac{R}{2\sqrt{d_{amb}}}$ on $H\setminus \bigcup_{i=1}^k C_i$ has atleast one point in each cell $c_{j_1,...,j_{d_{amb}}}$. 

       For the cavity $C$, we know that no cells completely contained in $C$ can have a single point from $S$. This is by definition of the $k$-cavity sampling criterion. 
       
        Consider the set of cells $N(C)$ adjacent to the missing cells. These are the cells adjacent by a length of $\frac{R}{2\sqrt{d_{amb}}}$ in an orthogonal direction to some missing cell. Let 
        \begin{equation}
            \hat{C}:= \{x \in H\setminus \bigcup_{i=1}^kC_i: x \in N(C)\}
        \end{equation}
       
       \textbf{Claim 2a. } 
    $\textsf{Rips}_R(\hat{C})$ forms atleast one $p$-dimensional homological cycle for $p=1,...,d_{amb}-1$.

    Proof by contradiction: 

    Say $\hat{C}$ does not form a $p$-dimensional homological cycle, then there is some $p-1$ dimensional face of a simplex $\sigma \in \hat{C}$ that has no incident $p$-dimensional simplices from $\hat{C}$. This means that there are $p+1$ points $P$ that did not form a $p$-dimensional simplex in $\hat{C}$. In $\textsf{VR}_R(S)$, these $p+1$ points must have that $\textsf{diam}(P)>R$. This contradicts the largest distance between two points: $2\frac{R}{2}$, which is twice the length of the diameter of the cells: $\frac{R}{2}$. We thus have that $\textsf{Rips}_R(\hat{C})$ contains a $p$-dimensional homological cycle.
    
    Since $\hat{C}$ forms atleast one $p$-dimensional homological cycle, there must be atleast one $p$-dimensional simplex causing a $p$-dimensional creation event in $\textsf{VR}_R(S)$. This creation event has infinite persistence since it can never close the $d_{amb}$-dimensional  ball $B$ by Claim 1. Thus, we get the inequality:
       \begin{equation}
           \beta_{p}(\textsf{VR}_R(S))\geq k, \forall p= 1,...,d_{amb}
       \end{equation}
    \end{proof}
	\subsection{Oblivious Column Reduction Proof}
	\label{sec: appendix-oblivious-proof}
	Recall the following notation: let $D$ and $R$ be as in Algorithm \ref{alg:standard-algorithm}, let $D_j$ denote the $jth$ column of $D$, $R_j$ denote a fully reduced column, i.e. the $jth$ column of $R$ after Algorithm \ref{alg:standard-algorithm} terminates, and let $R[j]$ denote the $jth$ column of $R$ during Algorithm \ref{alg:oblivious}, partially reduced.
    \begin{proof} 
    $\bf{base\ case}$:
    The first nonzero column $j_0$ requires no column additions. $R_{j_0}=D_{j_0}$ is equivalent to a fully reduced column by standard algorithm.

    
    $\bf{induction\ hypothesis}$: 
    We have reduced all columns from 0 to $j \geq j_0$ by the oblivious matrix reduction algorithm. Each such column $R[j']$, $j \geq j' \geq j_0$ (initially $D_{j'}$) was reduced by a sum of a sequence of $D_{i'}$, $i'<j'$, equivalent to a sum of a sequence of fully reduced $R_i, i<j'$ from the standard algorithm. 
    
    $\bf{induction\ step}:$
    Let the $\it{next}$ column $k$, to the right of column $j$, be a nonzero partially reduced column that needs column additions and call it $R[k]$ (initially $D_k$). Let $j= lookup[low(R[k])]$, the column index with matching lowest $\bf{1}$ with column $k$, that must add with column $k$. 

    If $R[j]=D_j$, then since column $k$ adds with $D_j$, certainly $R[k] \gets R[k]+(R_j=R[j]=D_j)$

    Otherwise if $R[j]$ $\neq$ $D_j$, add column $D_j$ with $R[k]$ and call this new column $R[k]'$ and notice that all nonzeros from $low(R[k])+1$ down to $low(R[k]')$ (viewing the column from top to bottom) of the working column $R[k]'$ are now exactly equivalent to the nonzeros from index $low(R[k])+1$=$low(R_j)+1$ down to $low(R[k]')=low(D_j)$ of column $D_j$. This is because $low(R[k])$ is the lowest $\bf{1}$ so all entries below it are zero, so we can recover a block of nonzeros equivalent to a bottom portion of column $D_j$ upon adding $D_j$ to $R[k]$. 
    
    We have recovered the exact same nonzeros of column $D_j$ from $low(R_j)+1=low(R[k])+1$ down to $low(D_j)$. Thus by the oblivious algorithm, before $low(R[k]')$ rises above $low(R[k])$ during column reduction, the sequence of columns to add to $R[k]'$ is equivalent to the sequence of columns to add to $D_j$. By the induction hypothesis on column $D_j$ $R_j=\Sigma_{i<j} D_i$, where the right hand side comes from Algorithm \ref{alg:oblivious}. 
    We thus have, $R[k] \gets R[k]+(R_j=\Sigma_{i<j} D_i)$.
    \end{proof}

    \section{An Equivalent Model to Section \ref{sec: appendix-apparent-model} for Analyzing Algorithm \ref{alg: apparentconstruction}}
    \label{sec: appendix-equivalent-model}
    Consider the same random variable $T_j$, as before, as the number of triangles in $T'$ not incident to any edge chosen from a sequence of edges $C$ of length $j$. Recall that in the model, there are only $j$ iterations of the modified Algorithm \ref{alg: apparentconstruction}. Consider the recurrence relation:
    
    $$T_j= T_{j-1}-Y_j$$ 
    where $Y_j$ is the random variable for the number of triangles removed at step $j$. 
    
    Taking expectations on both sides and using linearity of expectation, we get:
    $$E[Y_j] = \frac{\Sigma_{u_1...u_j}Y_j(u_1...u_{j-1},u_j)}{({n-1 \choose 2}-j+1)\cdot({n-1 \choose 2}-j+2)\cdots {n-1 \choose 2}}$$ 
    $$= \frac{\Sigma_{u_1...u_{j-1}}\Sigma_{u_j}     Y_j(u_1...u_{j-1},u_j)}{({n-1 \choose 2}-j+1)\cdot({n-1 \choose 2}-j+2)\cdots {n-1 \choose 2}}$$
    $$=\frac{\Sigma_{u_1...u_{j-1}} 3 \cdot ({n-1 \choose 3}-\Sigma_{k\leq j-1}Y_{k}(u_1...,u_{k}))}{({n-1 \choose 2}-j+1)\cdot({n-1 \choose 2}-j+2)\cdots {n-1 \choose 2}}$$ by the fact that the total number of triangles incident to all remaining edges at the $j$th step must be (3 $\cdot$ the remaining number of triangles after $j$-1 iterations). (Think of the bipartite graph between edges (left nodes) and triangles (right nodes); the total number of remaining bipartite edges at the jth step in this bipartite graph is what we are counting. These bipartite edges represent the triangles incident to the remaining edges or equivalently, the edges incident to the remaining triangles.)
    $$=\frac{\frac{\Sigma_{u_1...u_{j-1}} 3 \cdot T_{j-1}}{({n-1 \choose 2}-j+2)\cdots {n-1 \choose 2}}}{{n-1 \choose 2}-j+1}= \frac{3 \cdot E[T_{j-1}]}{{n-1 \choose 2}-j+1}$$
    

    We thus have the recurrence relation:
    $$E[T_j] = E[T_{j-1}]\cdot(1-\frac{3}{{n-1 \choose 2}-j+1})$$ with $E[T_0]= {n-1 \choose 3}$.
    Solving the recurrence, we get:
    $$= {n-1 \choose 3} \cdot \Pi_{i=0}^{j-1}(1-\frac{3}{{n-1 \choose 2}-i})= {n-1 \choose 3}\cdot \Pi_{i=0}^{j-1}\frac{{n-1 \choose 2}-i-3}{{n-1 \choose 2}-i}$$
    
    $$= {n-1 \choose 3} \cdot \frac{({n-1 \choose 2}-3)\cdot({n-1 \choose 2}-4)\cdots({n-1 \choose 2}-3-j+1)}{({n-1 \choose 2}\cdot({n-1 \choose 2}-1)\cdots({n-1 \choose 2}-j+1))}$$
    
    Notice this is the same equation as in Section \ref{sec: appendix-apparent-model}.

    